\documentclass[aps,prl,reprint,superscriptaddress,footinbib,longbibliography]{revtex4-2}

\usepackage{newpxtext,newpxmath}

\let\coloneqq\relax

\usepackage{xcolor}

\usepackage{newpxtext,newpxmath}

\usepackage[latin1]{inputenc}
\usepackage{amsthm}
\usepackage{amssymb}
\usepackage{amsmath}
\usepackage{bbold}
\usepackage{bbm}
\usepackage[pdftex, backref=page]{hyperref}
\usepackage{braket}
\usepackage{dsfont}
\usepackage{mathdots}
\usepackage{mathtools}
\usepackage{enumerate}
\usepackage[shortlabels]{enumitem}
\usepackage{csquotes}
\usepackage{stmaryrd}
\usepackage[cal=boondox]{mathalfa}
\usepackage{graphicx}
\usepackage{stackengine}
\usepackage{scalerel}
\usepackage{tensor}       %
\usepackage{array}
\usepackage{makecell}
\newcolumntype{x}[1]{>{\centering\arraybackslash}p{#1}}
\usepackage{tikz}
\usepackage{pgfplots}
\usetikzlibrary{shapes.geometric, shapes.misc, positioning, arrows, arrows.meta, decorations.pathreplacing, decorations.pathmorphing, patterns, angles, quotes, calc}
\usepackage{booktabs}
\usepackage{xfrac}
\usepackage{siunitx}
\usepackage{centernot}
\usepackage{comment}
\usepackage{chngcntr}

\newtheorem{thm}{Theorem}
\newtheorem*{thm*}{Theorem}

\newtheorem*{prop*}{Proposition}
\newtheorem{lemma}[thm]{Lemma}
\newtheorem*{lemma*}{Lemma}

\newtheorem*{cor*}{Corollary}

\newtheorem*{cj*}{Conjecture}

\newtheorem*{Def*}{Definition}

\newtheorem*{question*}{Question}

\newtheorem*{problem*}{Problem}

\makeatletter
\def\thmhead@plain#1#2#3{%
  \thmname{#1}\thmnumber{\@ifnotempty{#1}{ }\@upn{#2}}%
  \thmnote{ {\the\thm@notefont#3}}}
\let\thmhead\thmhead@plain
\makeatother

\theoremstyle{definition}

\makeatletter
\newcommand{\manualifempty}[3]{%
  \edef\@tempa{#1}%
  \ifx\@tempa\@empty
    #2%
  \else
    #3%
  \fi
}
\makeatother

\makeatletter
\newtheoremstyle{manualstyle}
  {3pt}{3pt}{\itshape}{}{\bfseries}{.}{ }{}

\theoremstyle{manualstyle}

\newtheorem{manualthminner}{Theorem}
\newenvironment{manualthm}[1]{%
  \def\@tempa{#1}%
  \ifx\@tempa\@empty
    \renewcommand{\themanualthminner}{\unskip}%
  \else
    \renewcommand{\themanualthminner}{#1}%
  \fi
  \begin{manualthminner}
}{\end{manualthminner}}

\newtheorem{manualpropinner}{Proposition}

\newtheorem{manuallemmainner}{Lemma}

\newtheorem{manualcorinner}{Corollary}

\makeatother

\newcommand{\bb}{\begin{equation}\begin{aligned}\hspace{0pt}}
\newcommand{\bbb}{\begin{equation*}\begin{aligned}}
\newcommand{\ee}{\end{aligned}\end{equation}}
\newcommand{\eee}{\end{aligned}\end{equation*}}

\renewcommand{\epsilon}{\varepsilon}
\newcommand{\G}{\mathrm{\scriptscriptstyle G}}
\newcommand{\dd}{\mathrm{d}}

\newcommand{\Z}{\mathds{Z}}
\newcommand{\C}{\mathds{C}}

\DeclareMathOperator{\Tr}{Tr}

\DeclareMathAlphabet{\pazocal}{OMS}{zplm}{m}{n}

\newcommand{\T}{\pazocal{T}}

\newcommand{\lsmatrix}{\left(\begin{smallmatrix}}
\newcommand{\rsmatrix}{\end{smallmatrix}\right)}

\stackMath

\stackMath

\makeatletter
\newcommand*\rel@kern[1]{\kern#1\dimexpr\macc@kerna}
\newcommand*\widebar[1]{%
  \begingroup
  \def\mathaccent##1##2{%
    \rel@kern{0.8}%
    \overline{\rel@kern{-0.8}\macc@nucleus\rel@kern{0.2}}%
    \rel@kern{-0.2}%
  }%
  \macc@depth\@ne
  \let\math@bgroup\@empty \let\math@egroup\macc@set@skewchar
  \mathsurround\z@ \frozen@everymath{\mathgroup\macc@group\relax}%
  \macc@set@skewchar\relax
  \let\mathaccentV\macc@nested@a
  \macc@nested@a\relax111{#1}%
  \endgroup
}

\counterwithin*{equation}{part}
\counterwithin*{thm}{part}
\counterwithin*{figure}{part}

\tikzset{meter/.append style={draw, inner sep=10, rectangle, font=\vphantom{A}, minimum width=30, line width=.8, path picture={\draw[black] ([shift={(.1,.3)}]path picture bounding box.south west) to[bend left=50] ([shift={(-.1,.3)}]path picture bounding box.south east);\draw[black,-latex] ([shift={(0,.1)}]path picture bounding box.south) -- ([shift={(.3,-.1)}]path picture bounding box.north);}}}
\tikzset{roundnode/.append style={circle, draw=black, fill=gray!20, thick, minimum size=10mm}}
\tikzset{squarenode/.style={rectangle, draw=black, fill=none, thick, minimum size=10mm}}

\definecolor{Blues5seq1}{RGB}{239,243,255}
\definecolor{Blues5seq2}{RGB}{189,215,231}
\definecolor{Blues5seq3}{RGB}{107,174,214}
\definecolor{Blues5seq4}{RGB}{49,130,189}
\definecolor{Blues5seq5}{RGB}{8,81,156}

\definecolor{Greens5seq1}{RGB}{237,248,233}
\definecolor{Greens5seq2}{RGB}{186,228,179}
\definecolor{Greens5seq3}{RGB}{116,196,118}
\definecolor{Greens5seq4}{RGB}{49,163,84}
\definecolor{Greens5seq5}{RGB}{0,109,44}

\definecolor{Reds5seq1}{RGB}{254,229,217}
\definecolor{Reds5seq2}{RGB}{252,174,145}
\definecolor{Reds5seq3}{RGB}{251,106,74}
\definecolor{Reds5seq4}{RGB}{222,45,38}
\definecolor{Reds5seq5}{RGB}{165,15,21}

\allowdisplaybreaks

\usepackage[most,breakable]{tcolorbox}
\makeatletter

\def\boxed@gobegin#1{\def\@tempa{#1}\def\@tempb{orange}\ifx\@tempa\@tempb\begin{tcolorbox}[colback=red!15,colframe=orange!70,breakable,enhanced]\else\begin{tcolorbox}[colback=Blues5seq1,colframe=Blues5seq5,breakable,enhanced]\fi}
\def\boxed@gobegin@empty{\begin{tcolorbox}[colback=Blues5seq1,colframe=Blues5seq5,breakable,enhanced]}
\makeatother
\inputencoding{utf8}
\backrefparscanfalse
\SetSymbolFont{stmry}{bold}{U}{stmry}{m}{n}
\usepackage{microtype}
\hypersetup{colorlinks=true,allcolors=blue,%
 hypertexnames=false,
 pdftitle={Bosonic Quantum Communication beyond the Thermal Threshold},
 pdfauthor={Francesco Anna Mele, Giuseppe Catalano, Marco Fanizza, Vittorio Giovannetti, Ludovico Lami}}
\allowdisplaybreaks[1]
\newcommand{\Ic}{I_{\mathrm c}}
\newcommand{\proj}[1]{\ket{#1}\!\bra{#1}}
\renewcommand{\T}{\mathcal{T}}
\newcommand{\Gcal}{\mathcal{G}}

\newcommand{\gfun}{g}
\newcommand{\grrank}{\operatorname{rank}}
\newcommand{\grdiag}{\operatorname{diag}}

\theoremstyle{plain}
\newtheorem{theorem}{Theorem}
\newtheorem{grproposition}[thm]{Proposition}
\newtheorem{grlemma}[thm]{Lemma}
\newtheorem{grcorollary}[thm]{Corollary}
\newcommand{\newrev}[1]{#1}

 \newenvironment{gaaddition}{}{}

\usepackage[normalem]{ulem}

\begin{document}

\title{Bosonic Quantum Communication beyond the Thermal Threshold}

\author{Francesco Anna Mele}
\email{fmele@caltech.edu}
\affiliation{Institute for Quantum Information and Matter,
California Institute of Technology, Pasadena, California 91125, USA}

\author{Giuseppe Catalano}
\email{giuseppe.catalano@sns.it}
\affiliation{Scuola Normale Superiore, Piazza dei Cavalieri 7, 56126 Pisa, Italy}

\author{Marco Fanizza}
\email{marco.fanizza@inria.fr}
\affiliation{Inria, T\'el\'ecom Paris--LTCI, Institut Polytechnique de Paris, Palaiseau, France}

\author{Vittorio Giovannetti}
\email{vittorio.giovannetti@sns.it}
\affiliation{Scuola Normale Superiore, Piazza dei Cavalieri 7, 56126 Pisa, Italy}

\author{Ludovico Lami}
\email{ludovico.lami@sns.it}
\affiliation{Scuola Normale Superiore, Piazza dei Cavalieri 7, 56126 Pisa, Italy}

\begin{abstract}
The thermal threshold has long been the highest level of thermal-loss noise below which bosonic quantum communication was known to be possible. Here, we show that this threshold can be surpassed. Simple single-mode non-Gaussian states yield positive quantum capacity beyond it at every nonzero bath temperature, even when all multimode Gaussian states yield zero rates. Even a two-mode Gaussian state suffices at low bath temperature, establishing a phenomenon of Gaussian superadditivity. In the weak-loss regime, we also determine the asymptotic expansion of the full quantum capacity up to and including the term linear in the loss.
\end{abstract}

\maketitle

\paragraph{Introduction.—}
The quantum capacity $Q(\mathcal N)$ is the largest asymptotic rate at which a noisy channel $\mathcal N$ can transmit unknown quantum states with vanishing error~\cite{WildeBook2017,KhatriWilde2024}. By the Lloyd--Shor--Devetak theorem~\cite{Lloyd1997,DevetakShor2005,Devetak2005} it equals the regularised coherent information,
\bb
 Q(\mathcal N) = \Ic^\infty(\mathcal N)
 \coloneqq \lim_{n\to\infty}\frac{\Ic(\mathcal N^{\otimes n})}{n},
 \label{capacity}
\ee
where \(\Ic(\mathcal N)\coloneqq\sup_\rho\Ic(\rho,\mathcal N)\) and
\bb
 \Ic(\rho,\mathcal N) = S(\mathcal N(\rho)) - S(\mathcal N^{\rm c}(\rho)),
\ee
where $\mathcal N^{\rm c}$ is a complementary channel~\cite{Lloyd1997,DevetakShor2005,Devetak2005}. Even in finite dimension, the regularisation makes $Q$ difficult to evaluate, while for bosonic channels we also face the complication of an infinite-dimensional Hilbert space~\cite{Serafini2017,Weedbrook2012}.

The thermal attenuator $\Phi_{\eta,\nu}$ models photon loss in a thermal bath~\cite{Serafini2017}. It mixes the input mode $A$ with a thermal mode $E$ of mean photon number $\nu$ on a beam splitter of transmissivity $0\leq\eta\leq1$. The annihilation operator of the received mode $B$ is

\bb
\hat b
=
\sqrt{\eta}\,\hat a
+
\sqrt{1-\eta}\,\hat e,
\qquad
\operatorname{Tr}\!\left[
\hat e^\dagger\hat e\,\tau_\nu
\right]
=
\nu,
\label{channel}
\ee
Here $\hat a$ and $\hat e$ are the annihilation operators of $A$ and $E$, and $\tau_\nu$ is the bath state. Holevo and Werner obtained the standard lower bound by evaluating coherent information on thermal inputs~\cite{HolevoWerner2001}. At $\nu=0$, degradability for $\eta \geq 0.5$ and antidegradability below $0.5$ determine the capacity exactly~\cite{Wolf2007}. At $\nu>0$, however, the attenuator is generally nondegradable and its quantum capacity remains unknown~\cite{Bradler2015,Sharma2018,WildeQi2018,Rosati2018,LimUpper2019,NohAlbertJiang2019,NohPirandolaJiang2020,Fanizza2021,Kianvash2024}. Note that the degradability properties of single-mode Gaussian channels were studied in Refs.~\cite{CarusoGiovannetti2006,CarusoGiovannettiHolevo2006,Holevo2007,Wolf2007,CarusoEisert2008}.

Related lines of work include two-way assisted and secret-key capacities~\cite{Pirandola2017,PirandolaGarciaPatronBraunsteinLloyd2009,GarciaPatronPirandolaLloydShapiro2009,TakeokaGuhaWilde2014Squashed,Takeoka2014,WildeTomamichelBerta2017,KaurWilde2017,DavisShirokovWilde2018,PirandolaChannelSimulation2018,LaurenzaTserkisBanchiBraunsteinRalphPirandola2019,MeleLamiGiovannetti2025,OrtolanoPirandolaBanchi2025}, activation and memory effects in Gaussian channels~\cite{SmithSmolinYard2011,LimActivation2019,LupoGiovannettiMancini2010,Lami2020,MeleLamiGiovannetti2022PRL,MeleLamiGiovannetti2022PRA,MeleDePalmaFanizzaGiovannettiLami2024,PirandolaEnvironment2021,CatalanoFanizzaMeleDePalmaGiovannetti2026}, rates for finitely many channel uses~\cite{WildeRenesGuha2016,MeleBarbarinoGiovannettiFanizza2025,MeleBarbarinoGiovannettiFanizzaMemory2025}, other continuous-variable noise models~\cite{ArqandMemarzadehMancini2020,LeviantXuJiangRosenblum2022,LamiWilde2023,ZhongOhJiang2023,MeleSalekGiovannettiLami2024,LinNoh2024} and GKP codes for pure loss~\cite{ZhengHeLeeNohJiang2025}.

Gaussian inputs over a single-channel use are optimal for several bosonic-channel problems~\cite{Wolf2007,Giovannetti2004,GiovannettiHolevoGarciaPatron2015,DePalma2017}. For the thermal attenuator, however, beating the best single-mode Gaussian input may not be enough: correlations across channel uses may improve the Gaussian rate. Can a non-Gaussian input outperform arbitrary multimode Gaussian states? The best-known lower bound on the quantum capacity was established by Holevo and Werner in a seminal 1999 work~\cite{HolevoWerner2001}. It is obtained from the single-use coherent information of thermal inputs in the limit of infinite energy:
\begin{equation}
 Q(\Phi_{\eta,\nu})\geq\left[\log_2\frac{\eta}{1-\eta}-\gfun(\nu)\right]_+,
 \label{thermal_input_lower_bound}
\end{equation}
where $\gfun(\nu)
\coloneqq
(\nu+1)\log_2(\nu+1)-\nu\log_2\nu$ and $[x]_+\coloneqq\max\{x,0\}$. The bound is positive precisely above the \emph{thermal threshold}:
\begin{equation}
 \eta_{\rm G}(\nu)\coloneqq\frac{2^{\gfun(\nu)}}{1+2^{\gfun(\nu)}}.
 \label{gaussian_threshold}
\end{equation}
On the other hand, the thermal attenuator is antidegradable, and hence has zero quantum capacity, for $\eta\leq\eta_{\rm AD}(\nu)$~\cite{LamiKhatriAdessoWilde2019}, where
\begin{align}
 \eta_{\rm AD}(\nu)
 &\coloneqq \frac{\nu+\frac12}{\nu+1}.
 \label{antidegradability_threshold}
\end{align}
Thus, in the intermediate region $\eta_{\rm AD}(\nu)<\eta\leq\eta_{\rm G}(\nu)$, the {thermal-state lower bound vanishes}, although we can no longer conclude from antidegradability that the quantum capacity is zero. We therefore ask the complementary question:
\begin{center}
 \emph{Is $Q(\Phi_{\eta,\nu})>0$ in this intermediate region?}
\end{center}

We answer both questions by constructing finite-energy non-Gaussian inputs with positive coherent information where every multimode Gaussian input has a nonpositive value [Fig.~\ref{highground_figure}(a)]. The comparison rests on a general result: for every single-mode Gaussian channel that is not equivalent to a thermal attenuator up to Gaussian unitaries applied before and after the channel, one use suffices to attain the regularised Gaussian optimum. For channels that are equivalent to the thermal attenuator, instead, at most two uses suffice. Specifically, at bath occupation $\nu\geq\nu_{\rm c}\approx0.296$, the thermal threshold is the exact positivity threshold for arbitrary multimode Gaussian inputs. For $0<\nu<\nu_{\rm c}$, however, a two-mode Gaussian input sent through two uses of the channel can exceed the single-use Gaussian optimum, at and just below the thermal threshold. We call this phenomenon \emph{Gaussian superadditivity}, as the Gaussian counterpart of coherent-information superadditivity~\cite{DiVincenzoShorSmolin1998,SmithSmolin2007}. \space

Finally, we examine the regime of high transmissivity, where we show that, at each fixed positive bath occupation, the thermal-state bound agrees with the full capacity through the term linear in the loss.

\medskip

\paragraph{The Gaussian optimum.---}
For a single-mode Gaussian channel $\mathcal N$, let $\Ic^\G(\mathcal N^{\otimes m})$ denote the supremum of the coherent information over all $m$-mode Gaussian inputs. The regularised rate is
\bb
 \Ic^{\G,\infty}(\mathcal N)\coloneqq\sup_{m\geq1}\frac{\Ic^\G(\mathcal N^{\otimes m})}{m}.
 \label{gaussian_regularized_definition}
\ee
In this notation, Gaussian superadditivity means $\Ic^{\G,\infty}(\mathcal N)>\Ic^\G(\mathcal N)$.

\begin{theorem}[(Gaussian optimum, Gaussian superadditivity, and thermal threshold)]
\label{gaussian_two_uses_thm}
One use suffices for single-mode Gaussian channels $\mathcal N$ outside the thermal-attenuator class (up to Gaussian input and output unitaries):
\bb
 \Ic^{\G,\infty}(\mathcal N)=\Ic^\G(\mathcal N).
 \label{gaussian_additivity_main}
\ee
For the thermal attenuator ($0<\eta<1$, $\nu>0$), two uses suffice:
\bb
 \Ic^{\G,\infty}(\Phi_{\eta,\nu})
 =\tfrac12\Ic^\G(\Phi_{\eta,\nu}^{\otimes2}).
 \label{gaussian_two_uses_main}
\ee
For $0<\nu<\nu_{\rm c}\approx0.296$, Gaussian superadditivity occurs at the thermal threshold $\eta=\eta_{\rm G}(\nu)$ and throughout an interval immediately below it:
\bb
 \Ic^{\G,\infty}(\Phi_{\eta,\nu})>\Ic^\G(\Phi_{\eta,\nu})=0.
 \label{gaussian_superadditivity_main}
\ee
For $\nu\geq\nu_{\rm c}$, one use is already optimal:
\bb
 \Ic^{\G,\infty}(\Phi_{\eta,\nu})=\Ic^\G(\Phi_{\eta,\nu}).
 \label{gaussianbenchmark}
\ee
For every $\nu>0$, the single-use Gaussian optimum equals the Holevo--Werner bound in Eq.~\eqref{thermal_input_lower_bound}, which vanishes precisely at and below $\eta_{\rm G}(\nu)$.
\end{theorem}
Theorem~\ref{gaussian_two_uses_thm} is proved in Secs.~\ref{gaussian_regularization_section} and~\ref{gaussian_noise_transition_section} of the Supplemental Material (SM)~\cite{SupplementalMaterial}. To prove it, we first use Gaussian input and output unitaries to reduce to phase-insensitive channels; the remaining canonical cases can be evaluated directly. After a symmetry reduction of the input, we express it as a pure Gaussian state with added classical displacement noise. We show that each noise strength can be replaced, for an upper bound, by zero or an arbitrarily large value. At these endpoints, we factorise the determinant ratio in the entropy integral by means of a generalised eigenvalue decomposition, turning its logarithm into a sum of scalar terms. We bound each scalar integral by a rate approached by the correlated, separable two-mode Gaussian input of Fig.~\ref{state_preparations_figure}(b)(i), prepared from two independent squeezed vacua with shared Gaussian random displacements. Substituting the pair’s limiting eigenvalue curve for the original one gives an upper bound on the entropy integrand at every value of the integration variable. A rank argument then shows that each independent correlated contribution requires a distinct additional mode, or helper. Unpaired contributions are bounded by the single-use Gaussian optimum. The regularised Gaussian optimum is therefore approached by independent single-mode inputs or products of these pairs, leaving only one real parameter to optimise.

 For $0<\nu<\nu_{\rm c}$, suitable squeezing and displacements give positive two-use Gaussian coherent information at and just below the thermal threshold, where no single-mode Gaussian input yields a positive value. In contrast to the fixed-energy improvements of Ref.~\cite{NohPirandolaJiang2020}, we allow arbitrarily large input energy in this comparison.

\begingroup
\medskip
\paragraph{Analytic non-Gaussian advantage.---}
A simple non-Gaussian state --- a mixture of two displaced squeezed vacua --- gives positive coherent information at the thermal threshold and even below it, for every $\nu>0$. For $\nu\geq\nu_{\rm c}$, this already goes beyond what Gaussian inputs can achieve, even with arbitrarily many modes.
\begin{samepage}
\begin{theorem}[(Positive capacity and non-Gaussian advantage)]
\label{analytic_thm}
For every $\nu>0$, the quantum capacity at the thermal threshold is strictly positive, $Q(\Phi_{\eta_{\rm G}(\nu),\nu})>0$. Moreover, positive quantum capacity persists in an explicit interval below the threshold:
\bb
 Q(\Phi_{\eta,\nu})>0
 \qquad \text{for all }\eta_*(\nu)<\eta\leq\eta_{\rm G}(\nu),
\ee
where the function $\eta_*$ is defined by
\bb
 \eta_*(\nu)=\frac{8\nu^2+9\nu+2}{4(2\nu+1)(\nu+1)}
 \label{analytic_threshold}
\ee
and satisfies $\eta_{\rm AD}(\nu)<\eta_*(\nu)<\eta_{\rm G}(\nu)$ for every $\nu>0$. For each such channel, a finite-energy single-mode non-Gaussian state has positive coherent information. For $\nu\geq\nu_{\rm c}$, this outperforms Gaussian inputs with any number of modes, since $\Ic^{\G,\infty}(\Phi_{\eta,\nu})=0$.
\end{theorem}
\end{samepage}

\begin{proof}[Proof sketch]
\newrev{We consider} \newrev{two pure Gaussian states $\ket{\varphi_0}$ and $\ket{\varphi_1}$ with covariance matrices $V_0,V_1$ and first moments $\boldsymbol\mu_0,\boldsymbol\mu_1$ given by}
\bb
V_0&=\operatorname{diag}(s^{-1},s),\qquad
 V_1=\operatorname{diag}(s,s^{-1}),\\[-1mm]
 \boldsymbol\mu_0&=-\boldsymbol\mu_1
 =-\frac{d}{2\sqrt2}(1,1)^{\mathsf T},
\label{squeezed_inputs}
\ee
where $s>1$ controls the squeezing and $d>0$ the distance between their first moments. Suppose Alice chooses $X\in\{0,1\}$ uniformly at random and prepares $\ket{\varphi_{X}}$. The corresponding average input is $\rho_{s,d}=\tfrac12\proj{\varphi_0}+\tfrac12\proj{\varphi_1}$ [Fig.~\ref{state_preparations_figure}(b)(ii)]. Since each component is pure, its outputs at Bob and the purified environment have equal entropy, so we can write
\bb
 \Ic(\rho_{s,d},\Phi_{\eta,\nu})
 =H(X|E)-H(X|B),
 \label{analytic_balance}
\ee
where $B$ denotes Bob's output and $E$ includes the purification of the thermal bath. Here $H(X|B)$ and $H(X|E)$ measure the uncertainty about which state Alice prepared, given access to Bob's output and the environment, respectively. Thus, positive coherent information means precisely that Bob is less uncertain about Alice's choice than the environment.

Increasing $d$ makes the states easier to distinguish for both Bob and the environment. We choose the squeezing so that Bob's uncertainty decreases faster. Symmetry lets us bound the conditional entropies through state-discrimination errors and determine their decay rates from Gaussian overlaps~\cite{PirandolaLloyd2008,SeshadreesanLamiWilde2018}. Whenever $\eta>\eta_*(\nu)$, a sufficiently large but finite $s$ makes $H(X|B)$ decay faster than $H(X|E)$. Keeping $s$ fixed and taking $d$ sufficiently large then gives positive coherent information.

For $\nu\geq\nu_{\rm c}$, Theorem~\ref{gaussian_two_uses_thm} gives a zero Gaussian optimum throughout the stated interval, proving the non-Gaussian advantage. Section~\ref{analytic_construction_subsection} of the SM gives the full proof.
\end{proof}
\endgroup
\medskip

\begin{figure*}[t]
 \centering
 \begin{minipage}[t]{0.5\textwidth}
  \vspace{0pt}\raggedright
  \includegraphics[width=0.96\linewidth]{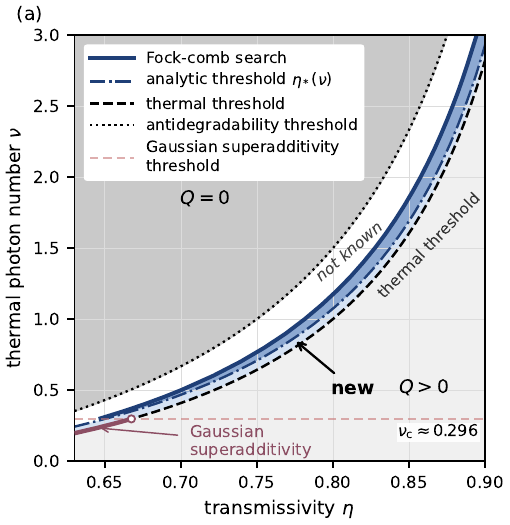}
 \end{minipage}%
 \begin{minipage}[t]{0.5\textwidth}
  \vspace{0pt}\raggedleft
  \begin{minipage}[t]{0.96\linewidth}
   \vspace{0pt}%
\resizebox{\linewidth}{!}{%
\begin{tikzpicture}[x=1cm,y=1cm,font=\fontsize{12.5}{14.4}\selectfont,
  overviewwire/.style={line width=.65pt,-{Latex[length=1.7mm,width=1.2mm]}},
  overviewcontrol/.style={densely dashed,line width=.6pt,-{Latex[length=1.6mm,width=1.1mm]}},
  overviewbox/.style={draw,line width=.65pt,fill=black!4,minimum height=.75cm}]

\node[anchor=west,font=\fontsize{15}{17}\selectfont] at (0,3.85) {\textbf{(b)}};

\node[anchor=west,font=\fontsize{14}{16}\selectfont] at (0,3.28) {\textbf{(i)} Two-mode Gaussian input};
\node[align=center] at (2.05,2.22) {Local\\squeezing};
\node[align=center] at (6.0,2.22) {Gaussian random\\displacement};
\node[align=center] at (12.25,2.22) {Two channel\\uses};
\node at (.28,1.15) {$\lvert0\rangle$};
\node at (.28,-1.15) {$\lvert0\rangle$};
\draw[overviewwire] (.62,1.15)--(1.45,1.15);
\draw[overviewwire] (.62,-1.15)--(1.45,-1.15);
\node[overviewbox,minimum width=1.15cm] at (2.02,1.15) {$S_1$};
\node[overviewbox,minimum width=1.15cm] at (2.02,-1.15) {$S_2$};
\draw[overviewwire] (2.60,1.15)--(4.95,1.15);
\draw[overviewwire] (2.60,-1.15)--(4.95,-1.15);
\node[overviewbox,minimum width=2.10cm] (overviewdispone) at (6.0,1.15) {$D_p(g_1z)$};
\node[overviewbox,minimum width=2.10cm] (overviewdisptwo) at (6.0,-1.15) {$D_p(g_2z)$};
\node[draw,rounded corners=2pt,fill=white,minimum height=.55cm,inner xsep=8pt]
  (overviewsampler) at (6.0,0) {shared Gaussian draw $z$};
\draw[overviewcontrol] (overviewsampler.north)--(overviewdispone.south);
\draw[overviewcontrol] (overviewsampler.south)--(overviewdisptwo.north);
\draw[overviewwire] (7.05,1.15)--(11.55,1.15);
\draw[overviewwire] (7.05,-1.15)--(11.55,-1.15);
\draw[decorate,decoration={brace,amplitude=4pt},line width=.55pt]
  (8.9,1.08)--(8.9,-1.08);
\node[align=center] at (10.60,0) {Gaussian\\and separable};
\node[above] at (9.45,1.15) {$A_1$};
\node[below] at (9.45,-1.15) {$A_2$};
\node[overviewbox,minimum width=1.40cm] at (12.25,1.15) {$\Phi_{\eta,\nu}$};
\node[overviewbox,minimum width=1.40cm] at (12.25,-1.15) {$\Phi_{\eta,\nu}$};
\draw[overviewwire] (12.95,1.15)--(14.15,1.15) node[right] {$B_1$};
\draw[overviewwire] (12.95,-1.15)--(14.15,-1.15) node[right] {$B_2$};

\begin{scope}[yshift=-.20cm]
\node[anchor=west,font=\fontsize{14}{16}\selectfont] at (0,-2.30) {\textbf{(ii)} One-mode binary mixture};
\begin{scope}[shift={(3.05,-4.55)},x=.40cm,y=.40cm]
  \pgfmathsetmacro{\overviewcentre}{5/(2*sqrt(2))}
  \draw[-{Latex[length=1.2mm]},black!45] (-4.2,0)--(4.25,0) node[right,black] {$q$};
  \draw[-{Latex[length=1.2mm]},black!45] (0,-4.0)--(0,4.1) node[right,black] {$p$};
  \draw[densely dashed,black!35] (-3.3,-3.3)--(3.3,3.3);
  \filldraw[fill=black!10,draw=black,line width=.65pt]
    (-\overviewcentre,-\overviewcentre) ellipse[x radius=.5,y radius=2];
  \filldraw[fill=black!22,draw=black,line width=.65pt]
    (\overviewcentre,\overviewcentre) ellipse[x radius=2,y radius=.5];
  \fill (-\overviewcentre,-\overviewcentre) circle[radius=1.0pt];
  \fill (\overviewcentre,\overviewcentre) circle[radius=1.0pt];
  \node[left] at (-2.6,-1.55) {$\lvert\varphi_0\rangle$};
  \node[above] at (4.2,2.6) {$\lvert\varphi_1\rangle$};
\end{scope}
\draw[overviewwire] (5.8,-4.45)--(8.4,-4.45)
  node[midway,above=5pt] {choose either state}
  node[midway,below=5pt] {probability $0.5$ each};
\node at (11.55,-4.45)
  {$\rho_{s,d}=\frac12\lvert\varphi_0\rangle\!\langle\varphi_0\rvert
                 +\frac12\lvert\varphi_1\rangle\!\langle\varphi_1\rvert$};

\end{scope}

\begin{scope}[yshift=-1.40cm]
\node[anchor=west,font=\fontsize{14}{16}\selectfont] at (0,-5.85) {\textbf{(iii)} One-mode Fock-comb mixture};
\node[anchor=east] at (1.65,-6.73) {$\lvert\psi_0\rangle$};
\node[anchor=east] at (1.65,-8.33) {$\lvert\psi_1\rangle$};
\foreach \base/\res in {-7.18/0,-8.78/1}{
  \draw[-{Latex[length=1.3mm]},black!45] (1.95,\base)--(6.85,\base)
    node[right,black] {$n$};
  \foreach \j/\height in {0/.90,1/.90,2/.45}{
    \pgfmathsetmacro{\level}{4*\j+\res}
    \pgfmathsetmacro{\stemx}{2.15+.46*\level}
    \draw[line width=.9pt] (\stemx,\base)--++(0,\height);
    \fill (\stemx,\base+\height) circle[radius=1.2pt];
    \node[below=2pt] at (\stemx,\base) {\pgfmathprintnumber[precision=0]{\level}};
  }
}
\draw[decorate,decoration={brace,amplitude=4pt},line width=.55pt]
  (7.35,-6.22)--(7.35,-8.83);
\draw[overviewwire] (7.6,-7.53)--(8.65,-7.53) node[midway,above=5pt] {mix};
\node at (11.65,-7.53)
  {$\rho_\star=\frac13\lvert\psi_0\rangle\!\langle\psi_0\rvert
                  +\frac23\lvert\psi_1\rangle\!\langle\psi_1\rvert$};
\end{scope}
\end{tikzpicture}%
}
  \end{minipage}
 \end{minipage}\par
 \vspace{3.71pt}
 \caption{\textbf{(a)} Non-Gaussian inputs extend the region of positive quantum capacity beyond the thermal threshold. Dark gray has zero capacity (anti-degradability region); light gray has positive thermal-input rates (Holevo--Werner lower bound~\cite{HolevoWerner2001}). Theorem~\ref{analytic_thm} proves positive quantum capacity in the light-blue region. Dark blue shows the additional region supported by numerical searches using Fock-comb inputs. For $\nu\geq\nu_{\rm c}$, the optimum over all multimode Gaussian inputs vanishes in both regions. Our bounds leave the white region unresolved. For $0<\nu<\nu_{\rm c}$, the highlighted thermal threshold shows Gaussian superadditivity: a correlated (yet separable) Gaussian pair beats the single-mode Gaussian optimum. \textbf{(b)} Three simple input preparations at the heart of our results.
(i) Displace two squeezed vacua using the same Gaussian random number, then discard that number. This creates correlations
without entanglement and, for suitable parameters, gives the Gaussian
superadditivity of Theorem~\ref{gaussian_two_uses_thm}.
(ii) Choose with equal probability either of the two oppositely displaced
states with perpendicular squeezing axes in Eq.~\eqref{squeezed_inputs}.
Their mixture $\rho_{s,d}$ gives positive coherent information at and
below the thermal threshold, as proved in Theorem~\ref{analytic_thm}.
(iii) Mix the two Fock combs in Eq.~\eqref{codewords} to obtain
$\rho_\star$ in Eq.~\eqref{witness}. Theorem~\ref{main_thm} shows that
this single-mode input achieves a positive rate where no multimode
Gaussian input does. Each Fock-comb state is a superposition of equally spaced Fock states.}
 \label{highground_figure}
 \label{state_preparations_figure}
\end{figure*}

\medskip

\paragraph{Explicit non-Gaussian rates.---}
To obtain larger rates and extend the positive-capacity region of Theorem~\ref{analytic_thm}, we turn to mixtures with finite Fock support.

Let us start with a simple example. At $\eta=\eta_{\rm G}(1)=0.8$ and $\nu=1$, all Gaussian inputs have nonpositive coherent information. Consider the two orthogonal states
\bb
 \ket{\psi_j}=\frac{2\ket{j}+2\ket{j+4}+\ket{j+8}}{3},
 \quad j\in\{0,1\},
 \label{codewords}
\ee
and the rank-two mixture
\bb
 \rho_\star=\frac13\proj{\psi_0}+\frac23\proj{\psi_1}\,,
 \label{witness}
\ee
shown in Fig.~\ref{state_preparations_figure}(b)(iii).

To prove the improved bound, we use the physical cutoff-erasure channel $\T_M$.  If $P_M=\sum_{n=0}^{M}\proj n$ and $\ket\perp$ is orthogonal to the retained Fock space, then
\bb
 \T_M(X)=P_MXP_M+\Tr[(\mathds{1}-P_M)X]\proj\perp,
 \label{cutoff}
\ee
where the flag records that the photon number exceeded the cutoff.

\begin{samepage}
\begin{theorem}[(An explicit non-Gaussian rate)]
\label{main_thm}
At $(\eta,\nu)=(0.8,1)$, the state $\rho_\star$ in Eq.~\eqref{witness} satisfies
\bb
 \Ic(\rho_\star,\Phi_{0.8,1})&\geq R_\star\coloneqq5\times10^{-4},\\
 \text{while}\qquad \Ic^{\G,\infty}(\Phi_{0.8,1})&=0.
 \label{mainineq}
\ee
Consequently, the non-Gaussian state $\rho_\star$ yields $Q(\Phi_{0.8,1})\geq5\times10^{-4}$, while all multimode Gaussian states give zero rate.
\end{theorem}
\end{samepage}

\begin{proof}
The Gaussian benchmark vanishes by Theorem~\ref{gaussian_two_uses_thm}. The cutoff-erasure map is a physical channel, so data processing gives
\bb
 \Ic(\rho_\star,\Phi_{0.8,1})\geq
 \Ic(\rho_\star,\T_{14}\circ\Phi_{0.8,1}).
 \label{dataprocessing}
\ee
The cutoff output and its joint state with a purifying reference are finite matrices. In Sec.~\ref{finite_witness_proof} of the SM we derive these matrices, and in Sec.~\ref{certification_appendix} of the SM we prove that their entropy difference is at least $R_\star$, using outward-rounded interval arithmetic. Together with Eq.~\eqref{dataprocessing}, this proves Eq.~\eqref{mainineq}; the finite-energy coding theorem then gives $Q\geq R_\star$~\cite[Theorem~5]{WildeQi2018}.
\end{proof}

\begingroup
\paragraph{Mechanism and applications.---}
To generalise this input, we use the same information-theoretic intuition as in Theorem~\ref{analytic_thm}. For a mixture of two orthogonal pure states, $\rho=p\proj{\psi_0}+(1-p)\proj{\psi_1}$, we have

\bb
 \Ic(\rho,\Phi_{\eta,\nu})=\chi_B-\chi_E,
 \label{holevo_balance}
\ee
where $\chi_B$ and $\chi_E$ are the Holevo quantities of the output ensembles at Bob and at the purified environment~\cite{SchumacherWestmoreland1998,WildeBook2017}. They quantify the information retained about Alice's choice $X$ between the two inputs. For Bob this is $\chi_B=H(X)-H(X|B)$, and for the environment it is $\chi_E=H(X)-H(X|E)$. Their difference gives the entropy balance in Eq.~\eqref{analytic_balance}. The coherent information is positive precisely when Bob retains more information about Alice's choice than the environment does. We therefore seek orthogonal inputs that leave similar environmental outputs.

To see the idea, consider pure loss first. For a superposition of $|0\rangle$ and $|m\rangle$, losing zero photons from the first component or $m$ from the second leaves the transmitted mode in the vacuum. Because both processes leave the same transmitted state, the input coherence generates coherence between environmental photon numbers $0$ and $m$. For a superposition of $|1\rangle$ and $|m+1\rangle$, the same two loss events leave one transmitted photon, while generating a similar pattern of environmental coherence. The intuition here is that in both cases the environment gets $0$ and $m$ photons, so the environment looks similar, while the two states left in Bob's system by these loss events are $\ket{0}$ and $\ket{1}$, so they are orthogonal. %

This suggests two superpositions with similar amplitude profiles, shifted by one photon. By spacing their occupied Fock levels by an integer $m\geq2$, we also make them orthogonal. We thus arrive at the \emph{Fock-comb} inputs

\bb
|\psi_r\rangle&=\sum_{j=0}^{J-1}c_{rj}|mj+r\rangle,
 \quad r=0,1,\\[-1mm]
 \rho&=p\proj{\psi_0}+(1-p)\proj{\psi_1},
\label{comb_ansatz}
\ee
with $\sum_j|c_{rj}|^2=1$ and $0<p<1$. Incidentally, related superpositions of equally spaced Fock states do appear in the bosonic error-correction literature~\cite{Michael2016,Grimsmo2020}.

For identical amplitude profiles, Bob's outputs have mean photon numbers differing by $\eta$, so they are not identical. Similar environmental coherence patterns do not, however, imply identical outputs: loss and gain probabilities depend on the initial Fock number. We therefore optimise the amplitudes and mixing probabilities.

Theorem~\ref{main_thm}\hypertarget{main_phase_flip}{} uses $m=4$, $J=3$ and the amplitudes and probabilities in Eqs.~\eqref{codewords}--\eqref{witness}, giving positive coherent information at $(\eta,\nu)=(0.8,1)$. Changing only the sign of the $\ket5$ component of $\ket{\psi_1}$ leaves all photon-number probabilities unchanged, but numerically makes the coherent information negative. Thus the relative phases matter as well as the photon-number distribution.

Optimising longer Fock-comb mixtures gives higher rates and numerically extends the positive-capacity region (Fig.~\ref{highground_figure}(a)). The inset of Fig.~\ref{wl_comparison_figure} shows the gain over all multimode Gaussian inputs at $\nu=1$. A finite-support input proves positive capacity down to $\eta=0.7841$, below the analytic threshold $\eta_*(1)\approx0.7917$; see Sec.~\ref{threshold_appendix} of the SM for the proof.

Using the same Fock-comb construction, we also obtain positive coherent information beyond the thermal threshold for thermal amplification and additive Gaussian noise (Secs.~\ref{thermal_amplifier_section} and~\ref{additive_noise_section} of the SM). For both channels, these rates outperform all multimode Gaussian inputs: Gaussian superadditivity is absent (Sec.~\ref{gr_thermal_evaluation_subsection} of the SM). 
\endgroup

\medskip

\paragraph{Capacity bounds and weak loss.---}
To complement these achievable rates, we derive an upper bound on the quantum capacity
using the technique of \emph{degradable extensions}~\cite{SmithSmolin2008Extensions,Fanizza2021}. The idea is to give the receiver an
additional bosonic mode: discarding it recovers the original channel,
so retaining it can only increase the capacity. We suitably choose the extension so that the receiver can reconstruct the full complementary output from the received modes. This makes the extension degradable and allows us to evaluate its capacity exactly. Minimizing this
capacity over the allowed extension parameters gives our upper bound, which improves on the earlier bounds as shown in Fig.~\ref{wl_comparison_figure}.

\begin{figure}[t]
 \centering
 \includegraphics[width=\columnwidth]{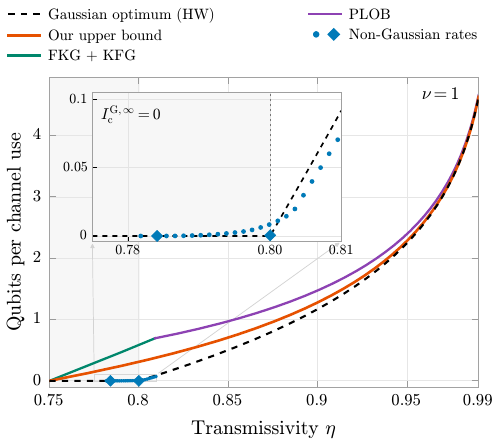}
 \caption{Our upper bound (orange) improves the earlier upper bounds at bath occupation $\nu=1$. The green bound is the one of FKG~\cite{Fanizza2021} with the KFG improvement~\cite{Kianvash2024}, while the purple one is the PLOB bound~\cite{Pirandola2017}. The dashed curve is the Holevo--Werner lower bound~\cite{HolevoWerner2001}, which we prove  to be optimal over all multimode Gaussian inputs (Theorem~\ref{gaussian_two_uses_thm}). The inset highlights positive non-Gaussian rates in the gray region where the multimode Gaussian optimum vanishes. Specifically, blue circles show rates obtained by numerically optimising the amplitudes and mixing probabilities of Fock-comb inputs [Eq.~\eqref{comb_ansatz} and its
extension to mixtures of more combs]. Diamonds represent the rigorously proved positive lower
bounds of Theorem~\ref{main_thm} and Eq.~\eqref{eta07841}.}
 \label{wl_comparison_figure}
\end{figure}

At small loss $\varepsilon\coloneqq1-\eta$ and fixed bath occupation
$\nu$, our upper bound and the thermal-input lower bound~\cite{HolevoWerner2001} differ by
only $O(\varepsilon^2\log(1/\varepsilon))$. Together, they determine
the quantum capacity through the term linear in the loss, which
earlier bounds leave undetermined~\cite{HolevoWerner2001,Pirandola2017,Fanizza2021,Kianvash2024}.%

\begin{theorem}[(Weak-loss expansion)]
\label{weak_loss_thm}
For every fixed \mbox{$\nu>0$}, the quantum capacity of the thermal attenuator $\Phi_{1 - \varepsilon,\nu}$ with transmissivity $1-\varepsilon$ satisfies, as $\varepsilon\to0$,
\bb
 Q(\Phi_{1 - \varepsilon,\nu})
 = \log_2\frac1\varepsilon - g(\nu) - \frac\varepsilon{\ln2}
 + O\!\left(\varepsilon^2\log\frac1\varepsilon\right),
 \label{weak_loss_main}
\ee
where the implicit constant may depend on $\nu$.
\end{theorem}
The proof is given in Section~\ref{weak_loss_section} of the SM.
Let us consider a practical example of application. 

Take a mode\hypertarget{main_weak_loss_example}{} with wavelength $9\,\mu\mathrm{m}$, as used in free-space communication experiments~\cite{Grillot2022ThermalInfrared}. For a thermal bath at room temperature ($300\,\mathrm{K}$), Planck's law gives $\nu\approx0.0049$. Taking $\eta=0.99$, our upper bound exceeds the thermal-input lower bound by about $10^{-5}$ qubits per channel use, compared with $10^{-3}$ for FKG~\cite{Fanizza2021}. Moreover, our linearised formula for the quantum capacity $Q(\Phi_{1-\varepsilon,\nu})\approx\log_2(1/\varepsilon)-g(\nu)-\varepsilon/\ln 2$ has an absolute error below $10^{-4}$ qubits per channel use.

\medskip
\paragraph{Discussion.---} 
In this work we advanced the understanding of quantum communication over the thermal attenuator, one of the most widely adopted models of optical links. 

Before our work, the quantum capacity was known to be positive only for transmissivities larger than the thermal threshold, as shown by the Holevo--Werner lower bound using thermal states~\cite{HolevoWerner2001}. Here, we show that the quantum capacity is still positive below this threshold, for every bath occupation $\nu>0$.

In doing so, we determine the Gaussian optimum of the thermal attenuator over arbitrarily many channel uses: two uses suffice. For $0<\nu<\nu_{\rm c}\approx 0.296$, a correlated yet unentangled two-mode Gaussian
input can outperform the single-use Gaussian optimum (corresponding to the Holevo--Werner lower bound) at and even
below the thermal threshold. We dub this phenomenon \emph{Gaussian
superadditivity}. In contrast, for $\nu\geq\nu_{\rm c}$, we prove that the thermal
threshold is the exact boundary for positive Gaussian coherent
information, even with arbitrarily many modes.

Our single-mode non-Gaussian inputs give positive coherent information in an explicit region of transmissivities smaller than the thermal threshold. Hence, for $\nu\geq\nu_{\rm c}$,
they outperform every multimode Gaussian input. Optimising
mixtures of Fock-comb states gives larger rates and extends the region
of positive capacity further.

We also derive an upper bound on the
quantum capacity. At weak loss and fixed bath occupation, this bound agrees with the
thermal-state lower bound through the term linear in the loss,
determining the full capacity to that order. Thus, while non-Gaussian
states provably help at larger loss, Gaussian states---and in particular
thermal states---are approximately optimal in the weak-loss regime.

Together, these results give both achievable-rate benchmarks and
limits for bosonic quantum error correction
\cite{NohAlbertJiang2019,LinNoh2024}. Our inputs have finite energy,
while the upper bound applies even without an energy constraint. The main open question is to determine the quantum capacity of the
thermal attenuator; and which input states and practical quantum error-correcting codes can approach it. More accessible questions may
be to determine the exact parameter region where the capacity is positive and whether Gaussian inputs are \emph{exactly} optimal at
sufficiently small loss.

\par\bigskip

\begingroup

\paragraph{AI statement.—}

We used AI (OpenAI GPT-5.5, GPT-5.6 and GPT-6; Anthropic Claude Opus~5, Claude Opus~5.5 and Claude Fable~5.1) to search for and simplify the non-Gaussian inputs, explore approaches, and draft arguments, after some of us had tried to solve this problem for several years before modern generative AI tools became available. We independently verified everything and rewrote the text for human readers, and we take full responsibility for the content of this paper.

\par\medskip
\paragraph{Data availability.—}
We provide the numerical data and Python programs supporting this work in the accompanying supplementary data files. We explain how to reproduce the numerical results in Sec.~\ref{numerics_guide_section} of the Supplemental Material~\cite{SupplementalMaterial}.

\par
\endgroup

\bibliography{references}

\clearpage
\onecolumngrid
\makeatletter
\do@output@MVL{\setcounter{page}{1}}
\makeatother
\setcounter{equation}{0}
\renewcommand{\theequation}{S\arabic{equation}}
\setcounter{thm}{0}
\renewcommand{\thethm}{S\arabic{thm}}
\renewcommand{\theHequation}{supp.\arabic{equation}}

 \newenvironment{smaddition}{\ignorespaces}{}

\begin{center}
{\large\bfseries Supplemental Material for\\[0.3em] ``Bosonic Quantum Communication beyond the Thermal Threshold''\par}
\vspace{0.75em}
Francesco Anna Mele, Giuseppe Catalano, Marco Fanizza, Vittorio Giovannetti, and Ludovico Lami
\end{center}

\noindent\textbf{Overview.}
We first ask whether the Holevo--Werner thermal-input lower bound remains optimal when optimising over arbitrary multimode Gaussian states. Surprisingly, this is not always the case. Section~\ref{gaussian_regularization_section} proves that one use gives the Gaussian optimum for every single-mode Gaussian channel outside the thermal-attenuator class, up to Gaussian unitaries before and after the channel. For thermal attenuation, two uses suffice. For thermal attenuation, Sec.~\ref{gaussian_noise_transition_section} identifies the phenomenon of \emph{Gaussian superadditivity}: for $0<\nu<\nu_{\rm c}$, where $\nu_{\rm c}\approx0.296$, a correlated (yet unentangled) two-mode Gaussian input gives a strictly higher rate per use than the single-use Gaussian optimum at and just below the thermal threshold. In contrast, for $\nu\geq\nu_{\rm c}$, the Holevo--Werner lower bound is optimal even over arbitrarily many modes. In this regime, the thermal threshold is the exact boundary at and below which no multimode Gaussian input can yield positive coherent information. These two sections therefore determine when a single Gaussian input suffices and when correlations between two inputs improve the rate. We then turn to inputs that can outperform this multimode Gaussian optimum.

What about non-Gaussian states? Section~\ref{analytic_construction_subsection} proves analytically that an equal mixture of two displaced squeezed states gives positive coherent information for every $\nu>0$, both at the thermal threshold and throughout an explicit interval below it. For $\nu\geq\nu_{\rm c}$, these single-mode inputs therefore outperform all Gaussian states over arbitrarily many modes. Sections~\ref{finite_witness_proof}--\ref{threshold_appendix} then give explicit Fock-comb inputs and rigorous lower bounds on their rates. Numerically optimising this construction gives higher rates in a larger parameter region. Sections~\ref{thermal_amplifier_section} and~\ref{additive_noise_section} extend the Fock-comb construction to thermal amplifiers and additive Gaussian-noise channels, \space with improvements over the optimum for all multimode Gaussian inputs in both cases.

To complement these achievable rates, Sec.~\ref{weak_loss_section} establishes an improved upper bound on the quantum capacity of the thermal attenuator. At weak loss $\varepsilon=1-\eta$ and fixed $\nu>0$, this bound agrees with the Holevo--Werner lower bound through the term linear in the loss. We therefore obtain an explicit approximate formula for the full capacity, with an absolute error of order $\varepsilon^2\log(1/\varepsilon)$. This complements the non-Gaussian advantage at larger loss with a weak-loss regime where thermal states are approximately optimal. Section~\ref{numerics_guide_section} explains how to reproduce the numerical results.

\raggedbottom
\setcounter{secnumdepth}{3}
\setcounter{section}{0}
\setcounter{table}{0}
\setcounter{figure}{0}
\renewcommand{\thefigure}{S\arabic{figure}}
\renewcommand{\theHfigure}{supp.figure.\arabic{figure}}
\section{Gaussian optimisation: one use or two}
\label{gaussian_optimizer_supplement_section}
\label{gaussian_regularization_section}

Our aim is to show that one use suffices for every single-mode Gaussian channel outside the thermal-attenuator class, while two uses suffice for thermal attenuation. We first prove a common two-use reduction by comparing arbitrary multimode Gaussian inputs with separable pairs. Section~\ref{gaussian_evaluation_section} evaluates their rates, and Sec.~\ref{gr_thermal_evaluation_subsection} shows that thermal inputs in one use already approach the same optimum for amplification and additive noise. Together with the Gaussian normal forms below, this leaves thermal attenuation as the only class that can require two uses. \space Section~\ref{gaussian_noise_transition_section} then classifies which single-mode Gaussian channels can exhibit Gaussian superadditivity and determines the critical bath occupation for thermal attenuation.

For a one-mode Gaussian channel $\mathcal{N}$, we let $\mathcal{G}_m$ denote the $m$-mode Gaussian states with finite covariance matrices and define
\bb
 \Ic^\G(\mathcal{N}^{\otimes m})\coloneqq\sup_{\rho\in\mathcal{G}_m}\Ic(\rho,\mathcal{N}^{\otimes m}),\qquad
 \Ic^{\G,\infty}(\mathcal{N})\coloneqq\sup_{m\geq1}\frac{\Ic^\G(\mathcal{N}^{\otimes m})}{m}.
 \label{gr_definition}
\ee
We use vacuum covariance $\mathds{1}$, base-two entropies, and $s(w)\coloneqq g((w-1)/2)$ for the entropy of a one-mode Gaussian state with symplectic eigenvalue $w$. In grouped position--momentum coordinates the symplectic form is $\Omega_m=\left(\begin{smallmatrix}0&\mathds{1}_m\\-\mathds{1}_m&0\end{smallmatrix}\right)$, and physical covariance matrices satisfy $V+i\Omega_m\geq0$. We write $\Ic(V,\mathcal{N})$ for the coherent information of the Gaussian state with covariance $V$; displacements do not affect it. There is no energy constraint on these suprema, and the complementary channel always comes from a pure dilation. Pure inputs give zero, so the quantities in \eqref{gr_definition} are nonnegative extended real numbers; infinite values are allowed even though each input has finite covariance.

\begin{grproposition}[(Two-use saturation)]
\label{gaussian_two_uses_prop}
Every one-input-mode, one-output-mode Gaussian channel $\mathcal{N}$ satisfies
\bb
 \Ic^{\G,\infty}(\mathcal{N})=\frac12\Ic^\G(\mathcal{N}^{\otimes2}),\qquad
 \Ic^\G(\mathcal{N}^{\otimes2k})=2k\Ic^{\G,\infty}(\mathcal{N})\quad(k\in\mathbb{N},\ k\geq1).
 \label{gr_two_use_saturation}
\ee
Products of separable two-mode Gaussian inputs approach these suprema; when they are infinite, finite-covariance product inputs achieve arbitrarily large rates.
\end{grproposition}

Sections~\ref{gaussian_normal_forms_subsection} and~\ref{gr_endpoint_subsection} give the preliminary reductions and the quasiconvexity lemma. The proof of Proposition~\ref{gaussian_two_uses_prop} is then given in Sec.~\ref{gr_two_use_proof_subsection}. Its main argument concerns the phase-insensitive channel $\mathcal{N}_{\gamma,c}:V\mapsto\gamma V+c\mathds{1}$ with $\gamma>0$ and $c>|1-\gamma|$. For inputs without position--momentum correlations, we first write the covariance as a pure covariance plus classical displacement noise. Quasiconvexity allows us to bound each noise strength by its zero- or infinite-noise endpoint. A matrix comparison then bounds each surviving noisy direction by a rate approached by a separable two-mode input. The crucial point is the mode count: a correlated noisy direction needs a second mode, which we call its helper. Doubling the input removes the restriction on position--momentum correlations. Gaussian unitary normal forms and the boundary cases below complete the reduction to a general one-mode channel.

\subsection{One-mode normal forms and boundary cases}
\label{gaussian_normal_forms_subsection}
The one-mode canonical classification of Holevo~\cite[Sec.~1]{Holevo2007} reduces the problem to a strictly noisy phase-insensitive channel and a few boundary cases. We give the reductions explicitly, because the same argument will also be needed for the rank-one quasiconvexity lemma.

We write the covariance action of $\mathcal{N}$ as $V\mapsto XVX^{\mathsf{T}}+Y$. Gaussian input unitaries map the finite-covariance Gaussian states bijectively onto themselves, and an output unitary leaves coherent information unchanged. Applying these unitaries independently to each use therefore preserves the optimum for every number of channel uses. This argument uses the absence of an energy constraint. Since in two dimensions $X\Omega X^{\mathsf{T}}=(\det X)\Omega$, we can express complete positivity as
\bb
 Y\geq0,\qquad \det Y\geq(1-\det X)^2.
 \label{gaussian_one_mode_cp}
\ee
When $\det X=\gamma>0$ and $Y>0$, the matrices
\bb
 S_{\rm out}\coloneqq\sqrt{c}\,Y^{-1/2},\qquad
 S_{\rm in}\coloneqq\sqrt{\gamma}\,(S_{\rm out}X)^{-1},\qquad c\coloneqq\sqrt{\det Y},
 \label{gaussian_normalising_matrices}
\ee
have determinant one and hence are symplectic. They transform the channel to $X'=\sqrt\gamma\,\mathds{1}_2$, $Y'=c\mathds{1}_2$.

For $\det X\leq0$, we prove antidegradability by constructing a channel with two identical output marginals: in a pure dilation of this extension, the complement of one output retains the other. If $\det X=0$, we stack two copies of $X$ and use noise $Y\oplus Y$, obtaining the complete-positivity matrix $(Y+i\Omega)\oplus(Y+i\Omega)\geq0$ because $X\Omega X^{\mathsf{T}}=0$. If instead $\det X=-\kappa<0$, we first whiten $Y$ as above and then apply the input symplectic matrix $(S_{\rm out}X)^{-1}\sqrt\kappa\,Z$, where $Z\coloneqq\operatorname{diag}(1,-1)$, to obtain $X'=\sqrt\kappa\,Z$, $Y'=c\mathds{1}_2$ with $c\geq1+\kappa$. For this channel, we take
\bb
 X_{\rm ext}\coloneqq\begin{pmatrix}\sqrt\kappa\,Z\\ \sqrt\kappa\,Z\end{pmatrix},\qquad
 Y_{\rm ext}\coloneqq\begin{pmatrix}c\mathds{1}_2&\kappa\mathds{1}_2\\ \kappa\mathds{1}_2&c\mathds{1}_2\end{pmatrix}.
 \label{gaussian_conjugate_extension}
\ee
A balanced beam splitter transforms its complete-positivity matrix into the direct sum of $(c+\kappa)\mathds{1}_2+i(1+2\kappa)\Omega$ and $(c-\kappa)\mathds{1}_2+i\Omega$. Both are positive semidefinite, and both output marginals reproduce the original channel. For an antidegradable channel, data processing on a pure reference--receiver--environment state gives $\Ic(R\rangle B)\leq\Ic(R\rangle E)=-\Ic(R\rangle B)$. The coherent information is therefore nonpositive. Tensor products of the antidegrading map give the same bound for correlated inputs across any number of uses, while pure Gaussian inputs attain zero.

If $Y$ is singular, we deduce from Eq.~\eqref{gaussian_one_mode_cp} that $\det X=1$. By removing the symplectic $X$ on the input and rotating the noise, with the inverse rotation on the input, we obtain either the identity channel or $X=\mathds{1}_2$, $Y=y\operatorname{diag}(1,0)$, where $y>0$. In the latter case, we use the pure-vacuum dilation
\bb
 Q_B=Q_A+\sqrt y\,Q_0,\quad P_B=P_A,\qquad
 Q_E=Q_0,\quad P_E=P_0-\sqrt y\,P_A.
 \label{gaussian_singular_dilation}
\ee
We reproduce its full complement by homodyning $P_B$ and preparing a vacuum displaced in momentum by $-\sqrt y$ times the outcome, thereby proving degradability. To determine the optimum, we take a thermal input $t\mathds{1}_2$, for which the receiver and complementary covariances are $\operatorname{diag}(t+y,t)$ and $\operatorname{diag}(1,1+yt)$, respectively, and obtain
\bb
 \Ic(t\mathds{1}_2,\mathcal{N})
 =s\bigl(\sqrt{t(t+y)}\bigr)-s\bigl(\sqrt{1+yt}\bigr)
 =\frac12\log_2\frac{t}{y}+o(1)\longrightarrow+\infty.
 \label{gaussian_singular_divergence}
\ee
Here $s(w)=\log_2w+\log_2(e/2)+o(1)$. For the identity channel, the thermal input instead gives $s(t)\to+\infty$. Products of these finite-covariance inputs give unbounded rates for every number of channel uses.

We will also need a monotonicity property of these channels. Adding a positive covariance matrix is equivalent to averaging over Gaussian displacements. Displacement covariance makes the conditional receiver and full-complement entropies independent of the displacement, so the change in coherent information is $\chi_B-\chi_E$. Relative-entropy data processing gives $\chi_E\leq\chi_B$ for a degrading map and the reverse inequality for an antidegrading map. Thus a covariance increment increases coherent information for degradable channels and decreases it for antidegradable channels. All average and conditional entropies here are finite, so the ensemble identity also applies to singular displacement distributions. Tensor-product degrading or antidegrading maps give the same conclusion without any restriction on correlations between the inputs.

Finally, quantum-limited attenuation is antidegradable for $\gamma\leq0.5$ and degradable for $0.5\leq\gamma<1$: we simulate the other beam-splitter output by attenuation with transmissivity $\gamma/(1-\gamma)$ or $(1-\gamma)/\gamma$, respectively. For quantum-limited amplification, $\gamma>1$, the full complement has linear part $\sqrt{\gamma-1}\,Z$ and noise $\gamma\mathds{1}_2$; we obtain it from the receiver by applying the completely positive Gaussian map with gain $\sqrt{(\gamma-1)/\gamma}\,Z$ and noise $(2\gamma-1)\mathds{1}_2/\gamma$. In the degradable cases, for every finite $V$ we choose $t\geq\max\{1,\lambda_{\max}(V)\}$ and use monotonicity to obtain
\bb
 \Ic(V,\mathcal{N}^{\otimes m})
 \leq m\Ic(t\mathds{1}_2,\mathcal{N})
 \longrightarrow m\log_2\frac{\gamma}{|1-\gamma|}.
\ee
We evaluate the last limit using the thermal receiver and complementary eigenvalues $\gamma t+|1-\gamma|$ and $|1-\gamma|t+\gamma$. With thermal products we approach the matching supremum, whereas in the antidegradable cases we obtain its zero branch, including $\gamma=0.5$. For $\gamma\neq1$, we therefore have
\bb
 \Ic^\G(\mathcal{N}_{\gamma,|1-\gamma|}^{\otimes m})=m\left[\log_2\frac{\gamma}{|1-\gamma|}\right]_+.
 \label{gaussian_quantum_limited_optimum}
\ee

Equation~\eqref{gaussian_quantum_limited_optimum} recovers the quantum-limited capacity formula of Ref.~\cite{Wolf2007}. Product Gaussian inputs therefore suffice in all boundary cases of Proposition~\ref{gaussian_two_uses_prop}. For the remaining channels, $\mathcal{N}_{\gamma,c}$ with $\gamma>0$ and $c>|1-\gamma|$, we next prove that each classical-noise strength can be replaced by zero or infinity in an upper bound.

\subsection{Rank-one quasiconvexity for Gaussian channels}\label{gr_endpoint_subsection}
The next lemma lets us replace a finite classical-noise strength by zero or infinity when seeking an upper bound. Quasiconvexity holds for every one-mode Gaussian channel. The additional statement about finite, commuting limits is needed only when $\det X>0$ and $Y>0$. After the normal-form reduction, its proof compares the receiver and full-complement responses to a covariance increment; their order turns out to be independent of the input covariance.

We use the Gaussian entropy representation of Ref.~\cite[Eq.~(339)]{RubboliRegulaBertaLamiTomamichel2026}, which in bits reads
\bb
 S(A)=n(\log_2 e-1)+\frac1{2\ln2}\int_0^1\ln\det(A+it\Omega_n)\,\dd t.
 \label{gr_entropy_integral}
\ee
Williamson's theorem gives $\det(A+it\Omega_n)=\prod_j(\sigma_j^2-t^2)$. Each scalar logarithm integrates to $2\ln2\,[s(\sigma_j)-\log_2 e+1]$, including at $\sigma_j=1$, where the logarithmic singularity is integrable.

For a general finite-mode Gaussian channel and a Gaussian complementary channel, we write their covariance maps as $V_\bullet(V)=X_\bullet VX_\bullet^{\mathsf{T}}+Y_\bullet$, where $\bullet\in\{B,E\}$, and define their entropy responses to a change in the input covariance matrix by
\bb
 D_\bullet(V,t)\coloneqq X_\bullet^{\mathsf{T}}\bigl(V_\bullet(V)+it\Omega_\bullet\bigr)^{-1}X_\bullet,\qquad 0\leq t<1.
 \label{gaussian_entropy_responses}
\ee
The determinant lemma expresses the entropy change along $V+auu^{\mathsf T}$ through the real scalars $u^{\mathsf{T}}D_\bullet(V,t)u$. Each response is nonnegative and nondecreasing in $t^2$. The more restrictive property needed below is the order between the receiver and complement responses.

\begin{grlemma}[(Rank-one quasiconvexity)]
\label{classical_noise_endpoints_lem}
Let $\mathcal{N}$ be any one-input-mode, one-output-mode Gaussian channel. For every finite physical $m$-mode covariance matrix $V$ and real vector $u$, the function $a\mapsto\Ic(V+auu^{\mathsf{T}},\mathcal{N}^{\otimes m})$ is quasiconvex on $a\geq0$: its value on every finite interval is at most the larger endpoint value. In particular,
\bb
 \Ic(V+auu^{\mathsf{T}},\mathcal{N}^{\otimes m})
 \leq\max\left\{\Ic(V,\mathcal{N}^{\otimes m}),\lim_{b\to\infty}\Ic(V+buu^{\mathsf{T}},\mathcal{N}^{\otimes m})\right\},\qquad a\geq0,
 \label{gr_endpoint}
\ee
where the limit may equal $+\infty$ or $-\infty$. If the covariance action of $\mathcal{N}$ is $V\mapsto XVX^{\mathsf{T}}+Y$ with $\det X>0$ and $Y>0$, the coherent-information limits obtained by sending finitely many independent classical-noise strengths to infinity are finite, commute with the entropy integral, and are independent of their relative rates of divergence.
\end{grlemma}
\begin{proof}
When making the normal-form reductions above, we preserve positive rank-one covariance increments, so we can use these reductions also for the present assertion. In the nonpositive-determinant and singular-noise cases, which are respectively antidegradable and degradable, we have already proved monotonicity by the displacement-ensemble argument, and the one-ray limits therefore exist in the extended reals. It remains to treat the canonical case: for $\gamma\neq1$, we use the beam-splitter or two-mode-squeezer dilation with bath covariance $c/|1-\gamma|$, retaining the bath purifier, and compute its Schur complement to obtain
\bb
 D_B^{-1}=V+\frac c\gamma\mathds{1}+\frac{it}{\gamma}\Omega_m,\qquad
 D_E^{-1}=V+\frac{\gamma c(1-t^2)}{c^2-(1-\gamma)^2t^2}\mathds{1}
 +it\left(1+\frac{\gamma(1-\gamma)(1-t^2)}{c^2-(1-\gamma)^2t^2}\right)\Omega_m.
 \label{phase_insensitive_responses}
\ee
Both matrices are positive definite for $t<1$. For additive noise, $\gamma=1$, we can obtain the response directly, without taking a limit of a thermal bath. For each input mode, introduce two independent vacuum modes, labelled $1,2$, and use the symplectic transformation
\bb
 Q_B&=Q_A+\sqrt c\,Q_1, & P_B&=P_A+\sqrt c\,P_2,\\
 Q_{E_1}&=Q_1, & P_{E_1}&=P_1-\sqrt c\,P_A,\\
 Q_{E_2}&=Q_2-\sqrt c\,Q_A-cQ_1, & P_{E_2}&=P_2.
 \label{gr_additive_dilation}
\ee
The receiver covariance is $V+c\mathds{1}$, as required. Retaining both environmental outputs gives a pure dilation; a Schur-complement calculation then yields $D_E^{-1}=V+(1-t^2)\mathds{1}/c+it\Omega_m$, exactly Eq.~\eqref{phase_insensitive_responses} at $\gamma=1$.

The input covariance cancels from the difference of the inverse responses:
\bb
 D_E^{-1}-D_B^{-1}
 =\left(\frac{\gamma(1-t^2)}{c^2-(1-\gamma)^2t^2}-\frac1\gamma\right)
 \bigl(c\mathds{1}+it(1-\gamma)\Omega_m\bigr).
 \label{gr_response_difference}
\ee
The final matrix is positive definite for $t<1$, whereas the scalar factor is nonincreasing in $t^2$, since its derivative is $\gamma[(1-\gamma)^2-c^2]/[c^2-(1-\gamma)^2t^2]^2\leq0$. Because inversion reverses order, we conclude that, for $u\neq0$, the difference of the positive responses $b(z)=u^{\mathsf{T}}D_B(V,\sqrt z)u$ and $e(z)=u^{\mathsf{T}}D_E(V,\sqrt z)u$ either has constant sign or changes sign once, from positive to negative.

Both responses also increase strictly with $z$. To verify this general Gaussian-channel property, for any positive covariance matrix $A$ we solve for the real and imaginary parts of its inverse and obtain, with $K=-\Omega A^{-1}\Omega>0$,
\bb
 \operatorname{Re}(A+it\Omega)^{-1}=(A-zK)^{-1},\qquad
 \frac{\dd}{\dd z}(A-zK)^{-1}=(A-zK)^{-1}K(A-zK)^{-1}>0.
 \label{gr_resolvent_monotonicity}
\ee
The strict inequality survives compression to a real vector whenever that signal direction is not in the channel's kernel. Here both response matrices in \eqref{phase_insensitive_responses} are invertible. The determinant lemma and Eq.~\eqref{gr_entropy_integral} therefore give
\bb
 \Ic(V+auu^{\mathsf{T}},\mathcal{N}_{\gamma,c}^{\otimes m})-\Ic(V,\mathcal{N}_{\gamma,c}^{\otimes m})&=\frac1{2\ln2}\int_0^1\ln\frac{1+ab(t^2)}{1+ae(t^2)}\,\dd t,\\
 \frac{\dd}{\dd a}\Ic(V+auu^{\mathsf{T}},\mathcal{N}_{\gamma,c}^{\otimes m})&=\frac1{2\ln2}\int_0^1\frac{b(t^2)-e(t^2)}{[1+ab(t^2)][1+ae(t^2)]}\,\dd t\qquad(a>0).
 \label{gr_ray_derivative}
\ee
We justify differentiation using $0<b/(1+ab),e/(1+ae)\leq1/a$ locally away from $a=0$. If the sign is constant, the coherent information is monotone along the ray; otherwise, denoting the crossing by $t_0$ and taking $a_1>a_0>0$, we compare the derivative integrands at $a_1$ and $a_0$ by multiplying the latter by
\bb
 w(t)=\frac{[1+a_0b(t^2)][1+a_0e(t^2)]}{[1+a_1b(t^2)][1+a_1e(t^2)]},
\ee
which decreases strictly with $t$. Subtract $w(t_0)$ times the derivative at $a_0$ from the derivative at $a_1$. The integrand contains the factors $w(t)-w(t_0)$ and $b(t^2)-e(t^2)$, which have the same sign on either side of $t_0$. The integral is therefore positive. In particular, once the derivative is nonnegative it cannot become negative, which excludes an interior maximum above both endpoints. This proves quasiconvexity. The case $u=0$ is constant. To obtain the infinite-endpoint bound, it remains to justify the limits.

To justify those limits, we collect independent noise vectors in $U=(u_1,\ldots,u_r)$, set $A_\bullet(V,t)=V_\bullet(V)+it\Omega_\bullet$, and expand by principal minors to obtain
\bb
 \frac{\det A_\bullet(V+\sum_j a_ju_ju_j^{\mathsf{T}},t)}{\det A_\bullet(V,t)}
 =\sum_{S\subseteq\{1,\ldots,r\}}\det(U_S^{\mathsf{T}}D_\bullet U_S)\prod_{j\in S}a_j.
 \label{gr_noise_polynomial}
\ee
For a general Gaussian channel these coefficients are strictly positive if both $X_BU$ and $X_EU$ have full column rank. For the present channel the condition follows from Eq.~\eqref{phase_insensitive_responses}; the empty determinant is one. The ratio of the two polynomials is therefore a convex combination of their coefficient ratios. %
These finitely many ratios are positive rational functions of $t$, with no zeros or poles on $[0,1)$. Near $t=1$, each has the form $(1-t)^k h(t)$, where $k$ is an integer and $h$ has a finite positive limit. Its logarithm is therefore bounded in absolute value by a constant times $1+|\ln(1-t^2)|$. Taking the largest constant over the coefficient ratios gives a bound $C_{V,U}(1+|\ln(1-t^2)|)$, independent of all nonnegative noise amplitudes. By applying dominated convergence with this integrable bound, we justify all successive noise limits; when every amplitude diverges, we retain the common highest-degree term and obtain
\bb
 \lim_{a_1,\ldots,a_r\to\infty}\frac{\det A_B(V+\sum_j a_ju_ju_j^{\mathsf{T}},t)/\det A_B(V,t)}{\det A_E(V+\sum_j a_ju_ju_j^{\mathsf{T}},t)/\det A_E(V,t)}
 =\frac{\det(U^{\mathsf{T}}D_BU)}{\det(U^{\mathsf{T}}D_EU)},
 \label{gr_noise_limit}
\ee
independently of their relative rates. By the same reasoning, we obtain the corresponding limits for any subset of diverging amplitudes. The endpoint inequality also survives these limits. To see this, keep the previous noise strengths finite, apply quasiconvexity on a finite interval for the next strength, and then take the previous limits in the inequality. The maximum of two convergent quantities also converges. We can then send the remaining upper endpoint to infinity, with the same value in either order. Rank-one quasiconvexity and the commuting noise limits reduce the coherent-information upper bound to zero-noise inputs and infinite-noise limits. We next compare each remaining noisy direction with the coherent information of a separable two-mode Gaussian input.
\end{proof}

\subsection{Proof of Proposition~\ref{gaussian_two_uses_prop}}
\label{gr_two_use_proof_subsection}
\begin{proof}[Proof of Proposition~\ref{gaussian_two_uses_prop}]
Section~\ref{gaussian_normal_forms_subsection} settles the boundary cases and reduces the remaining channels to $\mathcal N_{\gamma,c}$ with $\gamma>0$ and $c>|1-\gamma|$. We first treat inputs without position--momentum correlations, then remove this restriction by doubling the input.

\subsubsection{A common determinant representation}
\label{gr_master_subsection}
A common determinant formula will let us treat attenuation, amplification and additive noise together. It is useful both for taking the noise limits and for counting helper modes. We may assume that the two-use supremum is finite: if it is infinite, Eq.~\eqref{gr_definition} gives saturation immediately, and products of two-mode inputs give the assertion for every even number of channel uses.

For a real symmetric matrix $C$ and a real antisymmetric matrix $K$ of the same size, we define the affine matrix pencil
\bb
 H_s(C,K;t)\coloneqq C+cs\mathds{1}+it\bigl[1+(1-\gamma)s\bigr]K.
 \label{gr_entropy_pencil}
\ee
We recover the two inverse responses in \eqref{phase_insensitive_responses} by evaluating this pencil at $s=1/\gamma$ and $s=\gamma(1-t^2)/[c^2-(1-\gamma)^2t^2]$; by eliminating the bath purifier, we then obtain
\bb
 \det(B+it\Omega_m)&=\gamma^{2m}\det H_{1/\gamma}(V,\Omega_m;t),\\
 \det(E+it\Omega_E)&=\bigl[c^2-(1-\gamma)^2t^2\bigr]^m
 \det H_{\gamma(1-t^2)/[c^2-(1-\gamma)^2t^2]}(V,\Omega_m;t).
 \label{gr_master_determinants}
\ee
For $\gamma\neq1$, the purifier has covariance $c\mathds{1}/|1-\gamma|$, and after eliminating it we obtain the Schur complement $|1-\gamma|H_{\gamma(1-t^2)/[c^2-(1-\gamma)^2t^2]}(V,\Omega_m;t)$; in the amplification case, we first conjugate that block by $\grdiag(\mathds{1},-\mathds{1})$ to change the sign of its symplectic term. Multiplying by the purifier determinant, we recover the second prefactor in \eqref{gr_master_determinants}. For additive noise, the two-vacuum dilation in Eq.~\eqref{gr_additive_dilation} gives the same determinant identity directly, with prefactor $c^{2m}$.

Lemma~\ref{classical_noise_endpoints_lem} guarantees that the one-mode thermal limit exists. The entropy integral gives
\bb
 \ell_{\gamma,c}&\coloneqq\lim_{a\to\infty}\Ic(a\mathds{1}_2,\mathcal{N}_{\gamma,c})
 =\frac1{2\ln2}\int_0^1\ln\frac{\gamma^2(1-t^2)}{c^2-(1-\gamma)^2t^2}\,\dd t,\\
 \Ic(V,\mathcal{N}_{\gamma,c}^{\otimes m})&=m\ell_{\gamma,c}+\frac1{2\ln2}\int_0^1
 \ln\frac{\det H_{1/\gamma}(V,\Omega_m;t)}{\det H_{\gamma(1-t^2)/[c^2-(1-\gamma)^2t^2]}(V,\Omega_m;t)}\,\dd t.
 \label{gr_master_integral}
\ee
Indeed, the receiver has $m$ modes and we can take the full complement to have $2m$, so their entropy constants differ by $-m(\log_2 e-1)$. Using the identity $\int_0^1\ln(1-t^2)\,\dd t=2\ln2-2$, we combine this constant with \eqref{gr_master_determinants}, and then send $V=a\mathds{1}$ to infinity to obtain the first integral in \eqref{gr_master_integral}. Both determinants, and hence coherent information, are invariant under passive symplectic transformations and full time reversal; these symmetries will justify the coordinate changes and the doubling argument below. Equation~\eqref{gr_master_integral} expresses the coherent information of attenuation, amplification, and additive noise through the same entropy integral. We next use this integral to bound the contribution of each classical-noise direction by a rate approached by a separable two-mode input.

\subsubsection{A physical two-mode bound for each noisy direction}
\label{gr_rank_subsection}
Each surviving noise direction will contribute at most the two-use Gaussian optimum. The proof first bounds its rate by a scalar comparison curve and then constructs separable two-mode states that approach that curve.

We first take $V=Q\oplus P$. Since physicality is equivalent to $P\geq Q^{-1}$, we can write $V$ as the pure covariance $Q\oplus Q^{-1}$ plus positive momentum noise. By diagonalising that noise and applying Lemma~\ref{classical_noise_endpoints_lem} successively, we reduce the upper bound to endpoints with $n$ independent momentum directions at infinite noise, whose limits commute. We then orthogonally diagonalise $Q$ with eigenvalues $q_j>0$, collect the selected directions in a full-column-rank real matrix $X\in\mathbb{R}^{m\times n}$, and define
\bb
 \mathsf{A}(z)\coloneqq X^{\mathsf{T}}\grdiag\left(\frac{q_j}{\beta_j^2-z}\right)_jX,\qquad
 \mathsf{B}(z)\coloneqq X^{\mathsf{T}}\grdiag\left(\frac1{\beta_j^2-z}\right)_jX,\qquad
 \beta_j^2\coloneqq\gamma^2+c^2+\gamma c(q_j+q_j^{-1}).
 \label{gr_stieltjes}
\ee
The $\beta_j$ are the receiver's symplectic eigenvalues for the pure baseline, so in particular $\beta_j\geq\gamma+c>1$. If $\lambda_1(z)\leq\cdots\leq\lambda_n(z)$ are the generalised eigenvalues of $(\mathsf{A}(z),\mathsf{B}(z))$, we apply the determinant lemma to obtain
\bb
 \lim_{a\to\infty}\Ic\bigl(Q\oplus(Q^{-1}+aXX^{\mathsf{T}}),\mathcal{N}_{\gamma,c}^{\otimes m}\bigr)
 &=\frac1{2\ln2}\int_0^1\sum_{i=1}^n\phi_{t^2}(\lambda_i(t^2))\,\dd t,\\
 \phi_z(q)&\coloneqq\ln\frac{\gamma(1-z)(\gamma q+c)}{[c^2-(1-\gamma)^2z]q+\gamma c(1-z)}.
 \label{gr_matrix_rate}
\ee
For completeness, the two momentum responses of a pure diagonal mode are $\gamma(\gamma q+c)/(\beta^2-z)$ and $\bigl([c^2-(1-\gamma)^2z]q+\gamma c(1-z)\bigr)/[(\beta^2-z)(1-z)]$, whose compressions to $X$ are linear combinations of $\mathsf{A}(z)$ and $\mathsf{B}(z)$. By factoring their determinant ratio through the generalised eigenvalues, we obtain \eqref{gr_matrix_rate}.

The sign of $\partial_q\phi_z(q)$ is the sign of $\gamma^2-c^2+(1-2\gamma)z$. If $c\geq\gamma$, it is nonpositive on $[0,1]$, so using $\phi_z(q)\leq\phi_z(0)=0$ we bound every endpoint by zero. Otherwise, under the strict-noise assumption we have $\gamma>0.5$, and the sign changes exactly once, from positive to negative, at $z_0=(\gamma^2-c^2)/(2\gamma-1)\in(0,1)$.

For $0<x<1$, we introduce the comparison curve
\bb
 P_x(z)\coloneqq\frac{2\gamma cx+(\gamma^2+c^2-z)x^2}{\gamma c(1-x^2)}
 \label{gr_envelope}
\ee
because it satisfies the square identity
\bb
 q+\frac{x^2[\gamma^2+c^2+\gamma c(q+q^{-1})-z]}{\gamma c(1-x^2)}-P_x(z)
 =\frac{(q-x)^2}{q(1-x^2)}.
 \label{gr_square}
\ee
Divide this identity by $\beta_j^2-z$. The right-hand side increases with $z$, while the extra term on the left is constant. Compressing the resulting diagonal matrix to $X$ shows that $\mathsf{A}(z)-P_x(z)\mathsf{B}(z)$ increases in the matrix order. We also see that $P_x(z)$ increases continuously from zero to infinity as $x$ runs from zero to one, since its derivative has positive numerator proportional to $\gamma c(1-x)^2+x[(\gamma+c)^2-z]$.

We choose $x_i\in(0,1)$ so that $P_{x_i}(z_0)=\lambda_i(z_0)$. On the span of the first $i$ generalised eigenvectors at $z_0$, the quadratic form of $\mathsf{A}(z)-P_{x_i}(z)\mathsf{B}(z)$ is nonpositive for $z\leq z_0$, whereas on the span of the last $n-i+1$ it is nonnegative for $z\geq z_0$. By applying the two min--max formulas on these subspaces, we obtain $\lambda_i(z)\leq P_{x_i}(z)$ before the crossing and the reverse inequality afterwards, and hence in both intervals we have
\bb
 \phi_z(\lambda_i(z))\leq\phi_z(P_{x_i}(z)).
 \label{gr_scalar_comparison}
\ee
The subspaces are fixed at $z_0$; no simultaneous diagonalisation for different $z$ is required.

To realise the comparison curve physically, we take two product squeezed vacua with position covariance $\grdiag(x,L)$ and add momentum noise of strength $a$ in the direction $w=(\sqrt{1-x^2},x)^{\mathsf{T}}$, obtaining
\bb
 V_{a,L,x}\coloneqq\grdiag(x,L)\oplus\left[\grdiag(x^{-1},L^{-1})+aww^{\mathsf{T}}\right],\qquad a\geq0,\quad L>0.
 \label{gr_separable}
\ee
Every such covariance describes a separable state, since it is a product pure covariance plus classical displacement noise. To compare its rate with the desired curve, we compute its single generalised eigenvalue in \eqref{gr_matrix_rate}:
\bb
 \frac{\displaystyle\frac{(1-x^2)x}{\gamma^2+c^2+\gamma c(x+x^{-1})-z}+\frac{x^2L}{\gamma^2+c^2+\gamma c(L+L^{-1})-z}}
 {\displaystyle\frac{1-x^2}{\gamma^2+c^2+\gamma c(x+x^{-1})-z}+\frac{x^2}{\gamma^2+c^2+\gamma c(L+L^{-1})-z}}
 \ \xrightarrow[L\to\infty]{}\ x+\frac{x^2[\gamma^2+c^2+\gamma c(x+x^{-1})-z]}{\gamma c(1-x^2)}=P_x(z).
 \label{gr_pair_envelope}
\ee
The convergence is uniform for $z\in[0,1]$, because the first term in the denominator is bounded away from zero. For fixed $x$, these eigenvalues therefore remain in a compact positive interval, on which we can bound $\phi_z$ minus $\ln(1-z)$ uniformly using strict noise. After taking the first infinite-noise limit as in Lemma~\ref{classical_noise_endpoints_lem}, we apply dominated convergence to conclude that
\bb
 \frac1{2\ln2}\int_0^1\phi_{t^2}(P_x(t^2))\,\dd t
 =\lim_{L\to\infty}\lim_{a\to\infty}\Ic(V_{a,L,x},\mathcal{N}_{\gamma,c}^{\otimes2})
 \leq\Ic^\G(\mathcal{N}_{\gamma,c}^{\otimes2}).
 \label{gr_physical_envelope}
\ee
For each approximation, we first choose $L$ finite and then $a$ sufficiently large but finite, so that we never use an infinite covariance as an input. For $n$ classical-noise directions sent to infinity, Eqs.~\eqref{gr_matrix_rate}, \eqref{gr_scalar_comparison} and \eqref{gr_physical_envelope} bound the limiting coherent information by $n\Ic^\G(\mathcal{N}_{\gamma,c}^{\otimes2})$. To turn this into a bound per channel use, we next count how many modes are needed for these correlated contributions.

\subsubsection{Counting helper modes}
\label{gr_count_subsection}
After the endpoint reductions below, a rank bound ensures that distinct correlated directions have distinct helpers; unpaired directions contribute only the one-mode thermal rate.

We rotate the $n$ noisy directions to $p_1,\ldots,p_n$, let $d=m-n$, and parametrise the pure position covariance and its inverse as
\bb
 Q_{\rm tot}&=\begin{pmatrix}A&AZ_q\\Z_q^{\mathsf{T}}A&W_q+Z_q^{\mathsf{T}}AZ_q\end{pmatrix},\qquad A>0,\quad W_q>0,\\
 Q_{\rm tot}^{-1}&=\begin{pmatrix}A^{-1}+Z_qW_q^{-1}Z_q^{\mathsf{T}}&-Z_qW_q^{-1}\\-W_q^{-1}Z_q^{\mathsf{T}}&W_q^{-1}\end{pmatrix}.
 \label{gr_pure_completion}
\ee
By deleting the noisy momenta, we retain
\bb
 C(A)&\coloneqq\begin{pmatrix}A&AZ\\Z^{\mathsf{T}}A&W+Z^{\mathsf{T}}AZ\end{pmatrix},\qquad
 W\coloneqq W_q\oplus W_q^{-1},\quad Z\coloneqq(Z_q\ \ 0),\quad K\coloneqq0_n\oplus\Omega_d.
 \label{gr_compressed_real}
\ee
For any finite physical baseline, the coefficient of $T^n$ in a determinant with noise $T$ on these momenta is the complementary principal minor. Using this identity in \eqref{gr_master_integral}, we express the infinite-noise rate as
\bb
 m\ell_{\gamma,c}+\frac1{2\ln2}\int_0^1
 \ln\frac{\det H_{1/\gamma}(C(A),K;t)}{\det H_{\gamma(1-t^2)/[c^2-(1-\gamma)^2t^2]}(C(A),K;t)}\,\dd t.
 \label{gr_dual}
\ee
Finite changes in deleted rows and columns therefore do not affect this limit, which we can pass through the integral by Lemma~\ref{classical_noise_endpoints_lem}. Indeed, the leading power $T^n$ selects the $T$ term in every noisy momentum direction. Its coefficient contains only the retained principal minor; entries in the deleted rows and columns contribute only to lower powers of $T$.

We rotate the noisy modes to diagonalise $A$; changing its $i$th eigenvalue $a_i$ then changes $C(A)$ by a rank-one positive term along $(e_i,Z_{q,i:}^{\mathsf{T}},0_d)$. To apply the lemma to a genuine physical noise ray, we fix a pure completion \eqref{gr_pure_completion} at $a_i=\varepsilon>0$, add $(a_i-\varepsilon)$ times this rank-one term, and add noise $T$ on the selected momenta. %
The added vector has position part $(e_i,Z_{q,i:}^{\mathsf T})$ and zero momentum part. Thus the covariance remains physical for $a_i\geq\varepsilon$, and its retained block is exactly $C(A)$, with momentum block $W_q^{-1}$. We apply quasiconvexity with $\varepsilon,T$ finite and then take $T\to\infty$. Lemma~\ref{classical_noise_endpoints_lem} allows us to commute this limit with the infinite endpoint in $a_i$. Only then do we let $\varepsilon\to0$ in the retained expression. This replaces each $a_i$ by zero or infinity in the upper bound.

We control these further limits by the same coefficient-ratio argument as in the lemma. Each determinant in \eqref{gr_dual} is a multiaffine polynomial in the nonnegative $a_i$ with positive coefficients, because its constant matrix, $(0_n\oplus W)+cs\mathds{1}+it[1+(1-\gamma)s]K$, is positive definite by $W+it\Omega_d>0$ and $c>|1-\gamma|$. The coefficients depend on $W,Z$ and $t$, but not on the $a_i$, so the coefficient-ratio bound above applies uniformly as these eigenvalues vary. The finitely many coefficient ratios are positive rational functions of $t$, whose logarithms are bounded by a constant times $1+|\ln(1-t^2)|$; using this integrable bound, we apply dominated convergence at zero and infinity, independently of the relative divergence rates.

At $a_i=0$, the corresponding retained row contributes the ratio $[c^2-(1-\gamma)^2t^2]/[\gamma^2(1-t^2)]$; upon integrating its logarithm, we cancel one copy of $\ell_{\gamma,c}$ in \eqref{gr_dual}. We delete these modes, leaving $r$ noisy modes and $d$ helpers with $r+d\leq m$, and send the remaining eigenvalues of $A$ to infinity to obtain the rate
\bb
 (r+d)\ell_{\gamma,c}+\frac1{2\ln2}\int_0^1
 \ln\frac{\det\left[W+\frac c\gamma(\mathds{1}+Z^{\mathsf{T}}Z)+\frac{it}\gamma\Omega_d\right]}
 {\det\left[W+\frac{\gamma c(1-t^2)}{c^2-(1-\gamma)^2t^2}(\mathds{1}+Z^{\mathsf{T}}Z)+it\left(1+\frac{\gamma(1-\gamma)(1-t^2)}{c^2-(1-\gamma)^2t^2}\right)\Omega_d\right]}\,\dd t.
 \label{gr_hot_limit}
\ee
To verify this expression, we compute the Schur term from \eqref{gr_entropy_pencil}, namely $csZ^{\mathsf{T}}A(A+cs\mathds{1})^{-1}Z\to csZ^{\mathsf{T}}Z$, while the removed determinant ratio tends to one. Since the limit is independent of rates, we may take $A=a\mathds{1}_r$ and rotate its coordinates so that only $j=\grrank Z_q\leq\min(r,d)$ rows of $Z_q$ are nonzero. The pure completion then splits into $r-j$ independent modes and a block with $j$ noisy directions and $d$ helpers. For each finite $a$, Eq.~\eqref{gr_physical_envelope} bounds the latter contribution by $j\Ic^\G(\mathcal{N}_{\gamma,c}^{\otimes2})$. This bound is independent of $a$, so it remains valid in the large-energy limit, while the independent modes contribute $(r-j)\ell_{\gamma,c}$.

By using thermal inputs for the first inequality and products of arbitrary one-mode Gaussian inputs for the second, we obtain $\ell_{\gamma,c}\leq\Ic^\G(\mathcal{N}_{\gamma,c})\leq\tfrac12\Ic^\G(\mathcal{N}_{\gamma,c}^{\otimes2})$. %
\begin{smaddition}
The contribution of each endpoint is therefore bounded by
\[
 j\Ic^\G(\mathcal{N}_{\gamma,c}^{\otimes2})+(r-j)\ell_{\gamma,c}
 \leq\frac{r+j}{2}\Ic^\G(\mathcal{N}_{\gamma,c}^{\otimes2})
 \leq\frac m2\Ic^\G(\mathcal{N}_{\gamma,c}^{\otimes2}),
\]
because $j\leq d$, $r+d\leq m$, and the two-use supremum is nonnegative. Taking the largest endpoint value gives
\end{smaddition}
\bb
 \Ic(Q\oplus P,\mathcal{N}_{\gamma,c}^{\otimes m})
 &\leq\frac m2\Ic^\G(\mathcal{N}_{\gamma,c}^{\otimes2}).
 \label{gr_real_optimum}
\ee
The endpoint with no noisy directions is pure and has rate zero. %
This sharper endpoint count will also give the one-use formula below. The coherent information per use of any Gaussian input without position--momentum correlations is therefore at most the two-use Gaussian optimum per use. We next combine an arbitrary input with its time reverse to obtain the same bound without this restriction.

\subsubsection{Arbitrary covariance matrices and saturation}
\label{gr_doubling_subsection}
To remove the restriction on position--momentum correlations, combine the input with its time reverse. Full time reversal preserves physicality and, by Eq.~\eqref{gr_master_integral}, coherent information. For $V=\left(\begin{smallmatrix}Q&C\\C^{\mathsf{T}}&P\end{smallmatrix}\right)$, apply to these two copies the passive canonical transformation
\bb
 q'_1=(q_A+q_B)/\sqrt2,\quad p'_1=(p_A+p_B)/\sqrt2,\qquad
 q'_2=(p_A-p_B)/\sqrt2,\quad p'_2=(-q_A+q_B)/\sqrt2.
 \label{gr_doubling}
\ee
In grouped coordinates, the resulting covariance is $V\oplus\Omega_m^{\mathsf{T}}V\Omega_m$, with vanishing position--momentum block, and its coherent information is twice that of $V$. We can therefore apply \eqref{gr_real_optimum} to these $2m$ modes and divide by two to obtain
\bb
 \Ic(V,\mathcal{N}_{\gamma,c}^{\otimes m})\leq\frac m2\Ic^\G(\mathcal{N}_{\gamma,c}^{\otimes2}).
 \label{gr_arbitrary_upper}
\ee
By including $m=2$ in \eqref{gr_definition}, we obtain the reverse bound on the regularised quantity; taking products of two-use approximants, we then establish equality for every even number of channel uses. We can repeat the proof with the two-use supremum restricted to separable inputs, because both the comparison states \eqref{gr_separable} and products of the one-mode thermal approximants are separable. We thus conclude that separable two-mode states suffice to approach the unrestricted supremum, including when it is infinite: if the separable supremum were finite, we could repeat the upper bound at $m=2$ and contradict that infinity.

Together with the normal-form reductions and boundary cases, this proves Proposition~\ref{gaussian_two_uses_prop}: separable two-mode inputs approach the regularised Gaussian rate of every one-mode Gaussian channel. We next evaluate this rate by comparing thermal inputs with correlated pairs.
\end{proof}

\subsection{Evaluating the Gaussian optimum}
\label{gaussian_evaluation_section}

Comparing the thermal and correlated-pair rates reduces the remaining optimisation to one real parameter.

We work with the phase-insensitive channel $\mathcal N_{\gamma,c}:V\mapsto\gamma V+c\mathds{1}$, where $\gamma>0$ and $c>|1-\gamma|$. The boundary and singular cases were evaluated in Sec.~\ref{gaussian_normal_forms_subsection}. Throughout this subsection the channel parameters $\gamma,c$ are fixed. The variable $\kappa$ below describes the input; it is not a channel gain.

\emph{The two rates to compare.} A thermal input with covariance $a\mathds{1}_2$ approaches the signed coherent information $\ell_{\gamma,c}$ as $a\to\infty$, as defined in Eq.~\eqref{gr_master_integral}. If that limit is negative, a pure input gives the better value zero. These inputs therefore approach the one-use lower bound $[\ell_{\gamma,c}]_+$. The large-energy thermal limit was evaluated in Ref.~\cite{HolevoWerner2001}; optimisation over the thermal input energy was established in Ref.~\cite{Bradler2015}.

The other candidate comes from the separable two-mode family $V_{a,L,x}$ in Eq.~\eqref{gr_separable}, prepared as shown in Fig.~\ref{gaussian_pair_preparation_figure}. It starts with product squeezed vacua, with position covariance $\operatorname{diag}(x,L)$ and momentum covariance $\operatorname{diag}(x^{-1},L^{-1})$. The preparation is the correlated Gaussian random-displacement channel
\[
 \mathcal R_{a,x}(\rho)=\mathbb E_z[U_z\rho U_z^\dagger],\qquad
 U_z=D_p(\sqrt{1-x^2}\,z)\otimes D_p(xz),\qquad
 z\sim\mathcal N(0,a/2),
\]
where $D_p(\delta)=e^{i\delta\hat q}$ shifts the momentum by $\delta$. In covariance form its signal matrix is $\mathds{1}_4$ and its added noise is $0_2\oplus aww^{\mathsf T}$, with $w=(\sqrt{1-x^2},x)^{\mathsf T}$. This matrix is positive semidefinite, so the map is a valid Gaussian channel; at $a=0$ it is the identity. Acting on the product squeezed vacua gives exactly $V_{a,L,x}$. Each realisation is a product state, so the average is separable. The draw $z$ is discarded before transmission. Put $x=1/\kappa$. After the ordered limits $a\to\infty$ and then $L\to\infty$, only the shape parameter $\kappa>1$ remains.

\begin{figure}[htbp]
\centering
\begin{tikzpicture}[x=1cm,y=1cm,font=\normalsize,
    gqwire/.style={line width=.65pt,-{Latex[length=1.7mm,width=1.2mm]}},
    gclassical/.style={densely dashed,line width=.65pt,-{Latex[length=1.7mm,width=1.2mm]}}]
  \draw[densely dotted,line width=.7pt,rounded corners=3pt,fill=black!2]
    (3.75,-2.4) rectangle (8.05,2.4);
  \node[font=\small,align=center] at (5.9,2.95)
    {Correlated random-displacement channel\\$\mathcal R_{a,x}$};
  \node at (1.85,2.6) {Local squeezing};
  \node at (11.35,2.6) {Two channel uses};
  \node at (0,1.7) {$\lvert0\rangle$};
  \node at (0,-1.7) {$\lvert0\rangle$};
  \draw[gqwire] (.35,1.7) -- (1.1,1.7);
  \draw[gqwire] (.35,-1.7) -- (1.1,-1.7);
  \draw[line width=.65pt,fill=black!4] (1.1,1.25) rectangle (2.6,2.15);
  \draw[line width=.65pt,fill=black!4] (1.1,-2.15) rectangle (2.6,-1.25);
  \node at (1.85,1.7) {$S_x$};
  \node at (1.85,-1.7) {$S_L$};
  \draw[gqwire] (2.6,1.7) -- (4.4,1.7);
  \draw[gqwire] (2.6,-1.7) -- (4.4,-1.7);
  \draw[line width=.65pt,fill=black!4] (4.4,1.25) rectangle (7.4,2.15);
  \draw[line width=.65pt,fill=black!4] (4.4,-2.15) rectangle (7.4,-1.25);
  \node at (5.9,1.7) {$D_p(\sqrt{1-x^2}\,z)$};
  \node at (5.9,-1.7) {$D_p(xz)$};
  \draw[line width=.65pt,rounded corners=2pt] (4.05,-.5) rectangle (7.75,.5);
  \node[font=\small] at (5.9,.22) {Shared Gaussian draw};
  \node at (5.9,-.20) {$z\sim\mathcal N(0,a/2)$};
  \draw[gclassical] (5.9,.5) -- (5.9,1.25);
  \draw[gclassical] (5.9,-.5) -- (5.9,-1.25);
  \draw[gqwire] (7.4,1.7) -- (10.5,1.7);
  \draw[gqwire] (7.4,-1.7) -- (10.5,-1.7);
  \node[above] at (8.4,1.7) {$A_1$};
  \node[below] at (8.4,-1.7) {$A_2$};
  \draw[decorate,decoration={brace,amplitude=5pt},line width=.55pt]
    (8.55,1.65) -- (8.55,-1.65);
  \node[align=center,font=\small,text width=1.4cm] at (9.6,0) {Gaussian\\input\\[.2em]$V_{a,L,x}$};
  \draw[line width=.65pt,fill=black!4] (10.5,1.25) rectangle (12.2,2.15);
  \draw[line width=.65pt,fill=black!4] (10.5,-2.15) rectangle (12.2,-1.25);
  \node at (11.35,1.7) {$\mathcal N_{\gamma,c}$};
  \node at (11.35,-1.7) {$\mathcal N_{\gamma,c}$};
  \draw[gqwire] (12.2,1.7) -- (13,1.7) node[right] {$B_1$};
  \draw[gqwire] (12.2,-1.7) -- (13,-1.7) node[right] {$B_2$};
\end{tikzpicture}
\caption{Shared classical randomness creates the correlations in this two-mode Gaussian input. Two independent squeezed vacua receive momentum displacements controlled by the same Gaussian random number $z$. Dashed arrows show this classical control; the dotted box is the correlated random-displacement channel. Averaging over $z$ gives the covariance $V_{a,L,x}$, and the random number is discarded. Each realisation is a product state, so the averaged input is separable. The two modes then pass through independent uses of the same Gaussian channel. For thermal attenuation with $0<\nu<\nu_{\rm c}$, suitable choices give Gaussian superadditivity at the thermal threshold: a higher rate per use than any single-mode Gaussian input.}
\label{gaussian_pair_preparation_figure}
\end{figure}

The limiting \emph{total} coherent information of this pair is the scalar integral
\bb
 F_{\gamma,c}(\kappa)\coloneqq\frac1{2\ln2}\int_0^1\phi_{t^2}\bigl(P_{1/\kappa}(t^2)\bigr)\,\dd t,\qquad \kappa>1,
 \label{gr_canonical_envelope}
\ee
where $\phi$ is the logarithm of the receiver-to-complement response ratio in Eq.~\eqref{gr_matrix_rate}, and $P$ is the comparison curve in Eq.~\eqref{gr_envelope}. Equation~\eqref{gr_physical_envelope} identifies this integral with the limiting coherent information of $V_{a,L,1/\kappa}$. Thus the rate per channel use is $F_{\gamma,c}(\kappa)/2$. The factor $1/2$ counts the second mode, whose correlations make this family useful.

The mode count in Sec.~\ref{gr_count_subsection} explains why these two families suffice. Each useful correlated noise direction requires a helper mode, and its comparison integral bounds the pair's contribution by $F_{\gamma,c}(\kappa)$ for some $\kappa$. Unpaired noisy modes give the thermal contribution; pure modes give zero. Keeping the comparison integral in that proof therefore gives an explicit upper bound, which the two input families approach.

The endpoints of $F$ will also matter. For each $t<1$, $\phi_{t^2}(q)$ is monotone in $q>0$, with endpoint values zero and $\ln[\gamma^2(1-t^2)/(c^2-(1-\gamma)^2t^2)]$. This gives an integrable bound independent of $\kappa$. Dominated convergence therefore yields the continuous extension
\bb
 F_{\gamma,c}(1)=\ell_{\gamma,c},\qquad \lim_{\kappa\to\infty}F_{\gamma,c}(\kappa)=0.
 \label{gr_canonical_endpoints}
\ee
The endpoint $\kappa=1$ corresponds to one hot mode and a pure helper, hence to a total value $\ell_{\gamma,c}$; two independent hot thermal inputs instead approach $2\ell_{\gamma,c}$. The supremum of $F$ is finite and nonnegative, even when some of its values are negative. All limits here are approached with finite covariances; they need not be attained by a finite-energy input.

\begin{grproposition}[(Explicit canonical Gaussian optimum)]
\label{gaussian_explicit_formula_prop}
For $\gamma>0$ and $c>|1-\gamma|$,
\bb
 \Ic^\G(\mathcal{N}_{\gamma,c})=[\ell_{\gamma,c}]_+,\qquad
 \Ic^{\G,\infty}(\mathcal{N}_{\gamma,c})=\max\left\{[\ell_{\gamma,c}]_+,\frac12\sup_{\kappa\geq1}F_{\gamma,c}(\kappa)\right\}.
 \label{gr_canonical_formula}
\ee
The thermal rate is
\bb
 \ell_{\gamma,c}=
 \begin{cases}
 \displaystyle\log_2\frac\gamma{|1-\gamma|}-s\!\left(\frac c{|1-\gamma|}\right),&\gamma\neq1,\\
 \displaystyle\log_2\frac2{ec},&\gamma=1.
 \end{cases}
 \label{gr_canonical_thermal_rate}
\ee
For $\gamma\neq1$, the two-use comparison rate is
\bb
 F_{\gamma,c}(\kappa)=\log_2\frac\gamma{|1-\gamma|}+s(\gamma+c\kappa)-s(p_\kappa)-s(q_\kappa),
 \label{gr_canonical_F}
\ee
where the positive roots are specified by
\bb
 p_\kappa q_\kappa&=\frac{c^2+\gamma c\kappa}{|1-\gamma|},\\
 p_\kappa-q_\kappa&=\gamma\left(\frac{c\kappa}{1-\gamma}-1\right)\quad(\gamma<1),\qquad
 p_\kappa+q_\kappa=\gamma\left(\frac{c\kappa}{\gamma-1}+1\right)\quad(\gamma>1).
 \label{gr_canonical_roots}
\ee
Both roots exceed one throughout the strict-noise region. For amplification their order is immaterial. For additive noise,
\bb
 F_{1,c}(\kappa)=\log_2\frac2{ec\kappa}+s(1+c\kappa)-s(1+c/\kappa).
 \label{gr_additive_F}
\ee
\end{grproposition}

Equation~\eqref{gr_canonical_formula} compares the best thermal or pure input with the best correlated pair, in both cases per channel use.

\begin{proof}[Proof of Proposition~\ref{gaussian_explicit_formula_prop}]
We first prove that these rates are optimal, then evaluate the integral defining $F$.

\subsubsection{Why these families give the exact optimum}
\label{gr_evaluation_optimum_subsection}

\emph{Achievable values.} Pure inputs give zero, and increasingly hot thermal inputs approach $\ell_{\gamma,c}$. The separable family gives
\bb
 \lim_{L\to\infty}\lim_{a\to\infty}\Ic(V_{a,L,1/\kappa},\mathcal{N}_{\gamma,c}^{\otimes2})=F_{\gamma,c}(\kappa),\qquad\kappa>1.
 \label{gr_pair_limit}
\ee
For any desired accuracy, choose $L$ finite first and then $a$ finite. Products of these finite-covariance inputs approach the corresponding rate over more channel uses. Thus both branches on the right-hand side of Eq.~\eqref{gr_canonical_formula} are lower bounds. The remaining question is whether any Gaussian input can do better.

\emph{The upper bound and the cost of a helper.} Suppose first that $|1-\gamma|<c<\gamma$, the region where the receiver and complement responses cross. Write $F_*\coloneqq\sup_{\kappa\geq1}F_{\gamma,c}(\kappa)$ and, for brevity, $\ell=\ell_{\gamma,c}$. For an input $Q\oplus P$, the endpoint reduction in Sec.~\ref{gr_rank_subsection} leaves a pure covariance plus a set of classical-noise directions sent to infinity. Equations~\eqref{gr_matrix_rate} and~\eqref{gr_scalar_comparison} bound each surviving direction by a comparison rate no larger than $F_*$.

To count those rates correctly, recall the final block in Sec.~\ref{gr_count_subsection}: it has $r$ noisy modes and $d$ helpers, with $r+d\leq m$. After the rotation used there, only $j\leq\min(r,d)$ noisy modes couple to helpers; the other $r-j$ give the thermal contribution. Its coherent information is therefore bounded by
\bb
 jF_*+(r-j)\ell_{\gamma,c}
 \leq jF_*+(m-2j)[\ell_{\gamma,c}]_+,
 \qquad 0\leq j\leq\left\lfloor\frac m2\right\rfloor.
 \label{gr_evaluation_mode_budget}
\ee
Here $r-j\leq m-2j$ follows from $j\leq d$. Replacing $\ell$ by its positive part also allows us to account for any unused modes by the upper bound $[\ell]_+$ per mode.

The last expression is $m[\ell]_++j(F_*-2[\ell]_+)$. If a pair gives no more than two independent one-mode inputs, its coefficient of $j$ is nonpositive; otherwise we use as many pairs as the mode count allows. Maximising over the integer $j$ gives
\bb
 \Ic(Q\oplus P,\mathcal{N}_{\gamma,c}^{\otimes m})
 \leq\left\lfloor\frac m2\right\rfloor\max\left\{2[\ell_{\gamma,c}]_+,\sup_{\kappa\geq1}F_{\gamma,c}(\kappa)\right\}
 +(m\bmod2)[\ell_{\gamma,c}]_+.
 \label{gr_evaluated_mode_count}
\ee
For one mode, a phase rotation removes the position--momentum covariance, so this bound immediately gives $\Ic^\G(\mathcal N_{\gamma,c})\leq[\ell_{\gamma,c}]_+$. For a general $m$-mode covariance, use the doubling construction of Sec.~\ref{gr_doubling_subsection}. It produces a $2m$-mode input with no position--momentum correlations and twice the original coherent information. Applying Eq.~\eqref{gr_evaluated_mode_count} to those $2m$ modes and dividing by two gives a total upper bound $m\max\{[\ell_{\gamma,c}]_+,F_*/2\}$. Dividing by $m$ proves the regularised upper bound in Eq.~\eqref{gr_canonical_formula}.

\emph{The zero branch.} If $c\geq\gamma$, every response integral in Eq.~\eqref{gr_matrix_rate} is nonpositive. The endpoint reduction and doubling therefore bound every Gaussian coherent information by zero, which pure inputs attain. Equations~\eqref{gr_canonical_envelope} and~\eqref{gr_canonical_endpoints} also give $F_*=0$ and $\ell_{\gamma,c}\leq0$, so the same formula applies.

For the phase-insensitive channel $\mathcal N_{\gamma,c}$, thermal or pure inputs approach the one-use Gaussian optimum, and separable pairs give the only possible improvement per use. We next evaluate their limiting coherent information as a function of the pair parameter $\kappa$.

\subsubsection{Turning the comparison integral into entropies}
\label{gr_evaluation_integral_subsection}

To evaluate the integral, use the one-mode identity from Eq.~\eqref{gr_entropy_integral}:
\bb
 \frac1{2\ln2}\int_0^1\ln(w^2-t^2)\,\dd t=s(w)-\log_2(e/2),\qquad w\geq1.
 \label{gr_scalar_entropy_integral}
\ee
Thus every factor $w^2-t^2$ contributes an ordinary one-mode entropy and a known constant. We only need to factor the numerator and denominator of the response ratio.

For $\gamma\neq1$, substitute $P_{1/\kappa}$ into $\phi$ using Eqs.~\eqref{gr_envelope} and~\eqref{gr_matrix_rate}. With $z=t^2$, the result is
\bb
 \exp\bigl(\phi_z(P_{1/\kappa}(z))\bigr)
 =\frac{\gamma^2(1-z)[(\gamma+c\kappa)^2-z]}{(1-\gamma)^2(p_\kappa^2-z)(q_\kappa^2-z)}.
 \label{gr_envelope_factorisation}
\ee
The relations in Eq.~\eqref{gr_canonical_roots} are precisely what this factorisation needs: $p_\kappa q_\kappa$ fixes the constant term of $(p_\kappa^2-z)(q_\kappa^2-z)$, while the sum or difference fixes $p_\kappa^2+q_\kappa^2$ and hence its linear term; the quadratic coefficient is one. Multiplying the denominators verifies the identity directly.

We also need the roots to be real and greater than one, so that they are valid arguments of $s$. Put $v=c/|1-\gamma|>1$. For attenuation, the positive product and difference determine positive roots with $p_\kappa>q_\kappa$, and
\bb
 (p_\kappa+1)(q_\kappa-1)=(1-\gamma)(v^2-1)>0.
 \label{gr_attenuator_root_physicality}
\ee
Since $p_\kappa+1>0$, we have $q_\kappa>1$ and hence $p_\kappa>1$. For amplification, the discriminant at $\kappa=1$ is $[\gamma+(\gamma-2)v]^2$. Its derivative in $\kappa$ is $2\gamma v[\gamma(v\kappa+1)-2]>0$, so both roots are real throughout $\kappa\geq1$. Their positive sum and product make them positive. Finally, the sum exceeds two and $(p_\kappa-1)(q_\kappa-1)=(\gamma-1)(v^2-1)>0$; they must therefore both exceed one.

We can now integrate the logarithm of Eq.~\eqref{gr_envelope_factorisation} term by term using Eq.~\eqref{gr_scalar_entropy_integral}. The factor $1-z$ corresponds to $w=1$, with entropy $s(1)=0$. The two numerator factors and two denominator factors cancel all entropy constants. The remaining constant prefactor contributes $\log_2(\gamma/|1-\gamma|)$, giving Eq.~\eqref{gr_canonical_F}.

For additive noise, the factorisation instead reads
\bb
 \exp\bigl(\phi_z(P_{1/\kappa}(z))\bigr)
 =\frac{(1-z)[(1+c\kappa)^2-z]}{c^2\kappa^2[(1+c/\kappa)^2-z]}.
 \label{gr_additive_factorisation}
\ee
There is one fewer denominator factor, so one entropy constant remains. Together with the prefactor $1/(c^2\kappa^2)$ it gives $\log_2[2/(ec\kappa)]$, yielding Eq.~\eqref{gr_additive_F}. Evaluating the thermal integral in Eq.~\eqref{gr_master_integral} in the same way gives Eq.~\eqref{gr_canonical_thermal_rate}: for $\gamma\neq1$ its two factors contribute $s(1)-s(c/|1-\gamma|)$, while for $\gamma=1$ the pure-mode constant gives $\log_2[2/(ec)]$. The explicit entropy formulas complete Proposition~\ref{gaussian_explicit_formula_prop}, reducing the regularised Gaussian optimum to a one-parameter comparison between thermal inputs and separable pairs. We next express these rates in the usual channel parameters for attenuation, amplification, and additive noise.
\end{proof}

\subsection{ Gaussian additivity and the standard channel formulas}
\label{gr_thermal_evaluation_subsection}

For attenuation and amplification, $\ell$ is the thermal-input limit and $F(\kappa)/2$ the rate per use of a correlated pair. \space We first show that this pair cannot improve the Gaussian rate for amplification or additive noise. The corollary then gives the explicit formulas for all three channels.

\begin{gaaddition}
\begin{grlemma}[(Gaussian additivity for amplification and additive noise)]
\label{gr_nonattenuating_additivity_lem}
For $\gamma\geq1$ and $c>\gamma-1$, the functions in Proposition~\ref{gaussian_explicit_formula_prop} satisfy
\bb
 \sup_{\kappa\geq1}F_{\gamma,c}(\kappa)=[\ell_{\gamma,c}]_+.
 \label{gr_nonattenuating_envelope}
\ee
Consequently, for every integer $m\geq1$,
\bb
 \frac1m\Ic^\G(\mathcal N_{\gamma,c}^{\otimes m})
 =\Ic^{\G,\infty}(\mathcal N_{\gamma,c})
 =\Ic^\G(\mathcal N_{\gamma,c})=[\ell_{\gamma,c}]_+.
 \label{gr_amplifier_additive_gaussian_optimum}
\ee
\end{grlemma}
\begin{proof}
We compare the pair and thermal rates before integrating. Put $h=\gamma-1\geq0$, fix $x=1/\kappa\in(0,1)$, and write
\bb
 r(z)=\frac{\gamma^2(1-z)}{c^2-h^2z},\qquad
 Q(z)=\frac\gamma c P_x(z)=\frac{2\gamma cx+(\gamma^2+c^2-z)x^2}{c^2(1-x^2)},\qquad 0\leq z<1.
 \label{gr_nonattenuating_response}
\ee
Equations~\eqref{gr_matrix_rate}, \eqref{gr_physical_envelope} and~\eqref{gr_master_integral} give
\bb
 2\ln2\,F_{\gamma,c}(1/x)=\int_0^1\ln\frac{r(t^2)[Q(t^2)+1]}{Q(t^2)+r(t^2)}\,\dd t,
 \qquad
 2\ln2\,\ell_{\gamma,c}=\int_0^1\ln r(t^2)\,\dd t.
 \label{gr_nonattenuating_integrals}
\ee
If $c\geq\gamma$, then $r\leq1$ and both integrands are nonpositive, so $\sup F=0$. Suppose now $h<c<\gamma$. We will bound the first integrand by a fixed fraction of the second, valid on both sides of its sign change.

Set $M=\gamma^2-h^2=2h+1$, $d=c^2-h^2$ and $z_*=(\gamma^2-c^2)/M\in(0,1)$, so $r(z_*)=1$. With $Q_*=Q(z_*)$ and $\theta=Q_* /(Q_*+1)\in(0,1)$, define
\bb
 H(z)=\ln\frac{r(z)[Q(z)+1]}{Q(z)+r(z)}-\theta\ln r(z).
\ee
This choice makes $H(z_*)=H'(z_*)=0$. We show that $H\leq0$ throughout the interval.

Write $Q(z)=q+\beta(1-z)$, where $q=Q(1)>0$ and $\beta=x^2/[c^2(1-x^2)]>0$, and put $C(z)=c^2-h^2z$. Differentiation gives
\bb
 H'(z)=\frac{(z-z_*)K(1-z)}{(Q+r)(Q_*+1)(1-z)C^2(Q+1)},
\ee
where
\bb
 K(y)=\beta h^2[M(q+1)+2\beta d]y^2+d\beta[h^2q+\gamma^2+\beta d]y-dMq(q+1).
\ee
All denominator factors are positive. The polynomial $K$ is strictly increasing on $y\geq0$, with a negative constant term; it is linear when $\gamma=1$. Its unique positive root lies beyond $1-z_*=d/M$, since
\bb
 K(d/M)=-\frac{d(Q_*+1)}M[Q_*M^2-2\beta\gamma^2d]<0,
 \qquad
 Q_*M^2-2\beta\gamma^2d=2\gamma\beta\left[\frac{cM^2}{x}+hc^2+h\gamma(3h+1)\right]>0.
\ee
Thus $H$ decreases after $z_*$. Before $z_*$ it either increases throughout, or first decreases and then increases. Its maximum is therefore at $0$ or $z_*$, where $H(z_*)=0$.

It remains to check $H(0)\leq0$. Put $T=\gamma/c>1$. Direct substitution gives
\bb
 \frac{H(0)}2=\frac{\ln T}{Q_*+1}-\ln\frac{T+x}{1+Tx},\qquad
 Q_*+1=\frac{1+2Tx+[T^2-(T^2-1)/M]x^2}{1-x^2}
 \geq A\coloneqq\frac{1+2Tx+x^2}{1-x^2},
\ee
where we used $M\geq1$. For $1\leq s\leq T$,
\bb
 A\frac{1-x^2}{(s+x)(1+sx)}-\frac1s
 =\frac{x(2Ts-s^2-1)}{s(s+x)(1+sx)}\geq0,
\ee
because $2Ts-s^2-1=(s^2-1)+2s(T-s)\geq0$. Integrating from $1$ to $T$ gives
$A\ln[(T+x)/(1+Tx)]\geq\ln T$, and hence $H(0)\leq0$. We have proved the desired pointwise comparison.

Equation~\eqref{gr_nonattenuating_integrals} now yields $F(1/x)\leq\theta\ell\leq[\ell]_+$. The endpoints $F(1)=\ell$ and $F(\infty)=0$ give equality for the supremum. These limits also follow directly from the integrals: the first integrand lies between $0$ and $\ln r$, whose absolute value is integrable. Finally, Proposition~\ref{gaussian_explicit_formula_prop} gives $\Ic^{\G,\infty}=[\ell]_+$. Products of thermal inputs approach this rate when $\ell>0$, and pure product inputs give zero otherwise, proving the same equality for every $m$.
\end{proof}
\end{gaaddition}

\begin{grcorollary}[(Thermal attenuation, amplification, and additive noise)]
\label{gr_standard_channels_cor}
\label{additive_gaussian_cor}
For either thermal channel below, with bath occupation $\nu>0$, put $v=2\nu+1$. Then
\bb
 \Ic^\G(\mathcal N)=[\ell]_+,\qquad
 \Ic^{\G,\infty}(\mathcal N)=\frac12\Ic^\G(\mathcal N^{\otimes2})
 =\max\left\{[\ell]_+,\frac12\sup_{\kappa\geq1}F(\kappa)\right\}.
 \label{gr_main}
\ee
Here $\kappa\geq1$ parametrises the input pair, while the channel parameters are fixed; in the amplifier formula below, $\gamma$ denotes the gain.

\emph{Thermal attenuation.} For $\mathcal N=\Phi_{\eta,\nu}$, $0<\eta<1$,
\bb
 \ell(\eta,\nu)&=\log_2\frac\eta{1-\eta}-g(\nu),\\
 F(\kappa;\eta,\nu)&=\log_2\frac\eta{1-\eta}+s\bigl(\eta+(1-\eta)v\kappa\bigr)-s(p_\kappa)-s(q_\kappa),
 \label{gr_F}
\ee
where the positive roots satisfy
\bb
 p_\kappa-q_\kappa=\eta(v\kappa-1),\qquad p_\kappa q_\kappa=\eta v\kappa+(1-\eta)v^2.
 \label{gr_symplectic_roots}
\ee
\emph{Thermal amplification.} For $\mathcal N=\mathcal A_{\gamma,\nu}:V\mapsto\gamma V+(\gamma-1)v\mathds{1}$, with gain $\gamma>1$,
\bb
 \ell&=\log_2\frac\gamma{\gamma-1}-g(\nu),\\
 F(\kappa)&=\log_2\frac\gamma{\gamma-1}+s\bigl(\gamma+(\gamma-1)v\kappa\bigr)-s(p_\kappa)-s(q_\kappa),
 \label{gr_amplifier_rates}
\ee
where now
\bb
 p_\kappa+q_\kappa=\gamma(v\kappa+1),\qquad p_\kappa q_\kappa=\gamma v\kappa+(\gamma-1)v^2.
 \label{gr_amplifier_roots}
\ee
In both cases $p_\kappa,q_\kappa>1$, and
\bb
 F(1)=\ell,\qquad\lim_{\kappa\to\infty}F(\kappa)=0.
 \label{gr_F_endpoints}
\ee
For the thermal amplifier, the pair gives no improvement: for every $m\geq1$,
\bb
 \frac1m\Ic^\G(\mathcal A_{\gamma,\nu}^{\otimes m})
 =\Ic^{\G,\infty}(\mathcal A_{\gamma,\nu})
 =\left[\log_2\frac\gamma{\gamma-1}-g(\nu)\right]_+.
 \label{gr_amplifier_additivity}
\ee
\emph{Additive noise.} For $\mathcal N_\xi:V\mapsto V+2\xi\mathds{1}$ with $\xi>0$, and every integer $m\geq1$,
\bb
 \frac1m\Ic^\G(\mathcal N_\xi^{\otimes m})
 =\Ic^{\G,\infty}(\mathcal N_\xi)
 =[-\log_2(e\xi)]_+.
 \label{additive_benchmark}
\ee
\end{grcorollary}
\begin{proof}
\emph{Thermal attenuation and amplification.} Substitute $(\gamma,c)=(\eta,(1-\eta)v)$ and $c=(\gamma-1)v$, respectively, into Proposition~\ref{gaussian_explicit_formula_prop}, using $s(v)=g(\nu)$. Equation~\eqref{gr_canonical_endpoints} gives the endpoint values.

Lemma~\ref{gr_nonattenuating_additivity_lem} then gives the amplifier equality~\eqref{gr_amplifier_additivity}. For additive noise, apply the same lemma with $(\gamma,c)=(1,2\xi)$ and $\ell_{1,2\xi}=-\log_2(e\xi)$. This proves Eq.~\eqref{additive_benchmark} directly.
\end{proof}

 Thermal inputs suffice for amplification and additive noise, regardless of the number of channel uses. For attenuation, a correlated pair can do better; the next section determines when this Gaussian superadditivity occurs.

\clearpage
\section{ Gaussian superadditivity of single-mode Gaussian channels}\label{gaussian_noise_transition_section}\label{gr_threshold_subsection}
We now ask which single-mode Gaussian channels can benefit from correlated Gaussian inputs. The normal-form reduction in Sec.~\ref{gaussian_normal_forms_subsection} preserves the optimum for every number of uses. Channels with $\det X\leq0$ have zero coherent information after Gaussian optimisation; those with singular $Y$ already have an unbounded single-use optimum. Quantum-limited channels are additive, and Lemma~\ref{gr_nonattenuating_additivity_lem} proves the same for noisy amplification and additive noise. Thus only strictly noisy thermal attenuation remains. The result below completes this classification by identifying the critical bath occupation and an exact scalar test for superadditivity.

We fix $\eta,\nu$ and write $F(\kappa)$, $\ell$, and $\ell_+=[\ell]_+$ for the rates in Corollary~\ref{gr_standard_channels_cor}. The one-use Gaussian optimum is the thermal-state bound $\ell_+$. The question here is whether a correlated two-mode Gaussian input can give more than $2\ell_+$ in total, and hence a larger rate per use. We call this phenomenon \emph{Gaussian superadditivity}. By Proposition~\ref{gaussian_two_uses_prop}, two uses already suffice for the regularised Gaussian optimum, so resolving this comparison also resolves the optimisation over any number of channel uses.

Gaussian superadditivity is possible because input correlations affect the receiver's and environment's joint entropies differently. This does not require entanglement between the inputs. Let $\rho_1,\rho_2$ be the marginals of a finite-energy two-mode Gaussian input $\rho_{12}$, and let $B_j,E_j$ be the receiver and full complementary outputs of use $j$. With $\mathcal N=\Phi_{\eta,\nu}$, the entropy identity
\bb
 \Ic(\rho_{12},\mathcal N^{\otimes2})=\Ic(\rho_1,\mathcal N)+\Ic(\rho_2,\mathcal N)+I(E_1:E_2)-I(B_1:B_2)
 \label{gr_correlation_balance}
\ee
follows by expanding $I(U:V)=S(U)+S(V)-S(UV)$. Correlations reduce a joint entropy relative to the sum of its marginal entropies. If this reduction is larger for the environment than for the receiver, it increases the joint coherent information. At the thermal threshold the two marginal coherent informations are nonpositive; a positive joint rate requires the correlation term to outweigh these losses.

The shared-displacement family of Sec.~\ref{gaussian_evaluation_section}, shown in Fig.~\ref{gaussian_pair_preparation_figure}, lets us test whether this correlation gain is large enough. It remains separable at every finite value of its parameters. Its total coherent information approaches
\bb
 F(\kappa;\eta,\nu)
 =\lim_{L\to\infty}\lim_{a\to\infty}
 \Ic(V_{a,L,1/\kappa},\Phi_{\eta,\nu}^{\otimes2}),\qquad \kappa>1.
 \label{gr_transition_pair_rate}
\ee
We therefore compare $F/2$ with $\ell_+$. The problem is to decide whether the scalar curve $F(\kappa)$ ever exceeds $2\ell_+$.

At the thermal threshold, this comparison is especially direct: $\ell=0$, and both endpoints $F(1)$ and $F(\infty)$ vanish. If $F'(1)>0$, the pair has positive coherent information for some $\kappa>1$. The condition $F'(1)\leq0$ alone does not exclude a positive value at larger $\kappa$. We first construct a finite-energy input with positive coherent information, then bound $F$ on its full domain to prove the converse and extend it to every transmissivity.

We retain the threshold notation
\bb
 R_{\mathrm G}(\nu) \coloneqq 2^{g(\nu)},\qquad
 \eta_{\mathrm G}(\nu) \coloneqq \frac{R_{\mathrm G}(\nu)}{1 + R_{\mathrm G}(\nu)},
 \label{gr_thermal_threshold}
\ee
and define the function that determines the sign of the derivative $F'(1)$:
\bb
 d(\nu) \coloneqq \frac{2}{1 + R_{\mathrm G}(\nu)}
 \left[(2\nu + 1)\ln\!\left(1 + \frac{1 + R_{\mathrm G}(\nu)}{\nu}\right) - R_{\mathrm G}(\nu)(\nu + 1)\ln\!\left(1 + \frac1\nu\right)\right].
 \label{gr_threshold_function}
\ee
Equation~\eqref{gr_kappa_boundary_slope} below relates $d(\nu)$ to $F'(1)$. The unique positive zero of $d$ is the critical bath occupation $\nu_{\rm c}$ in Theorem~\ref{gaussian_two_uses_thm}; Sec.~\ref{gr_certificate_subsection} proves its uniqueness and the required signs on the full parameter range.

\begin{grproposition}[( Gaussian superadditivity of single-mode channels)]
\label{gaussian_noise_transition_prop}
A single-mode Gaussian channel can exhibit Gaussian superadditivity only if it is equivalent, up to Gaussian input and output unitaries, to a thermal attenuator with $0<\nu<\nu_{\rm c}$. Every other single-mode Gaussian channel satisfies $\Ic^\G(\mathcal N^{\otimes m})/m=\Ic^{\G,\infty}(\mathcal N)=\Ic^\G(\mathcal N)$ for all $m\geq1$, with infinite values allowed.
The function $d$ has exactly one positive zero $\nu_{\rm c}$, with $\nu_{\rm c}\approx0.296$ and the exact enclosure in Eq.~\eqref{gr_critical_bracket}. For every $\nu\geq\nu_{\rm c}$ and every $0<\eta<1$,
\bb
 \Ic^\G(\Phi_{\eta,\nu}^{\otimes m}) = m\ell_+\quad(m\geq1),\qquad
 \Ic^{\G,\infty}(\Phi_{\eta,\nu}) = \ell_+.
 \label{gr_additive_regime}
\ee
For every $0<\nu<\nu_{\rm c}$, a finite-energy separable two-mode Gaussian input has positive coherent information at $\eta = \eta_{\mathrm G}(\nu)$ and on an interval immediately below it, where $\Ic^\G(\Phi_{\eta,\nu}) = 0$. For arbitrary $0<\eta<1$ and $\nu>0$, Gaussian superadditivity is characterised by
\bb
 \Ic^{\G,\infty}(\Phi_{\eta,\nu})>\Ic^\G(\Phi_{\eta,\nu})
 \quad\Longleftrightarrow\quad \sup_{\kappa\geq1}F(\kappa)>2\ell_+.
 \label{gr_exact_gap_test}
\ee
\end{grproposition}
The normal forms and Lemma~\ref{gr_nonattenuating_additivity_lem} already exclude all channels except noisy thermal attenuators. For these, the following argument excludes $\nu\geq\nu_{\rm c}$ and proves superadditivity near the thermal threshold for every $0<\nu<\nu_{\rm c}$. The proof first constructs an input with positive coherent information when $F'(1)>0$, then bounds $F$ everywhere when $F'(1)\leq0$. The exact-arithmetic sign proof at the end determines the bath occupations in each case; the universal two-use reduction does not depend on it.

\subsection{Positive two-use coherent information at the thermal threshold}\label{gr_scalar_subsection}
We first show that the sign of $d(\nu)$ has a direct meaning: it is the sign of $F'(1)$ at the thermal threshold. This needs only the two relations defining $p_\kappa,q_\kappa$ in Eq.~\eqref{gr_symplectic_roots}. Put $v=2\nu+1$ and $b=\eta+(1-\eta)v$. At $\kappa=1$, their positive roots and derivatives are
\bb
 p_1=v,\qquad q_1=b,\qquad
 p'_1=\frac{\eta v(v+1)}{v+b},\qquad
 q'_1=\frac{\eta v(1-b)}{v+b}.
 \label{gr_direct_root_derivatives}
\ee
Indeed, differentiate $p_\kappa-q_\kappa$ and $p_\kappa q_\kappa$ and solve the resulting two linear equations. Since $s'(z)=\ln[(z+1)/(z-1)]/(2\ln2)$, differentiating Eq.~\eqref{gr_F} then gives
\bb
 F'(1)=\frac{v}{2\ln2\,(v+b)}
 \left[2(1-\eta)v\ln\!\left(1+\frac1{(1-\eta)\nu}\right)
 -\eta(v+1)\ln\!\left(1+\frac1\nu\right)\right].
 \label{gr_direct_rate_derivative}
\ee
At $\eta=\eta_{\mathrm G}(\nu)$, the bracket is exactly $d(\nu)$. Therefore
\bb
 F(1)=0,\qquad F'(1)=\frac{v}{2\ln2\,(v+b)}d(\nu)
 \qquad\bigl(\eta=\eta_{\mathrm G}(\nu)\bigr).
 \label{gr_kappa_boundary_slope}
\ee
Suppose first that $d(\nu)>0$. Since $F(1)=0$, a positive derivative gives $F(\kappa)>0$ for some fixed $\kappa>1$. The exact-arithmetic sign proof in Sec.~\ref{gr_certificate_subsection} shows that this happens for every $0<\nu<\nu_{\rm c}$. In symbols,
\bb
 d(\nu)>0
 \ \Longrightarrow\ F'(1)>0
 \ \Longrightarrow\ \exists\,\kappa>1:\ F(\kappa)>0
 \ \Longrightarrow\ \exists\,\rho\in\Gcal_2:\ \Ic(\rho,\Phi_{\eta_{\mathrm G}(\nu),\nu}^{\otimes2})>0.
\ee
It remains to turn the positive limit into an actual input. Fix the selected $\kappa$, so that $F(\kappa)>0$. By Eq.~\eqref{gr_pair_limit}, we can first choose $L$ finite and large enough that the inner limit is positive, and then choose $a$ finite and large enough that the coherent information itself is positive. The resulting state $V_{a,L,1/\kappa}$ is separable and has finite energy. Keeping this same input fixed, continuity in $\eta$ preserves positivity on an interval immediately below $\eta_{\mathrm G}(\nu)$. Throughout that interval the one-use Gaussian optimum remains zero. This proves the claimed Gaussian superadditivity without an infinite-energy input.

As a concrete example, interval arithmetic gives, at $\nu = 0.1$ and $\eta = \eta_{\mathrm G}(0.1)$,
\bb
 3.85\times10^{-3}<F(1.2)<3.86\times10^{-3}
 \qquad\text{bits per two uses}.
 \label{gr_example}
\ee
Dividing by two gives the per-use comparison: $\Ic^{\G,\infty}(\Phi_{\eta,\nu})>1.92\times10^{-3}$ while $\Ic^\G(\Phi_{\eta,\nu}) = 0$, without finding the maximising $\kappa$. 

Noh, Pirandola and Jiang~\cite{NohPirandolaJiang2020} showed that, under sufficiently low input-energy constraints, correlated multimode Gaussian states can outperform a single-mode thermal input at the same mean energy per mode. They obtain these states by passively mixing thermal inputs of different energies. Their work does not prove that this thermal benchmark is optimal among all single-mode Gaussian inputs. The resulting rate per use is an average of single-mode thermal rates, so it cannot exceed the best thermal rate when arbitrary energies are allowed. Here, two-mode Gaussian inputs beat the single-use Gaussian optimum even without an energy constraint.

 For every $0<\nu<\nu_{\rm c}$, separable two-mode Gaussian inputs of finite energy therefore give a positive rate at and just below the thermal threshold, where every single-mode Gaussian rate is nonpositive. We next prove that the pair rate stays nonpositive when $F'(1)\leq0$, which will establish the opposite regime.

\subsection{Ruling out positive two-use coherent information}\label{gr_shape_subsection}
Fix $\eta=\eta_{\mathrm G}(\nu)$. We show that $F'(1)\leq0$ forces $F(\kappa)\leq0$ for every $\kappa\geq1$. It suffices to consider $\nu\geq1/4$: by Lemma~\ref{gr_signs_lem}(a) below, every $\nu$ with $d(\nu)\leq0$ satisfies $\nu\geq\nu_{\rm c}>1/4$.

\emph{Change of variable.} The following relations hold at every transmissivity. Write $R\coloneqq\eta/(1-\eta)$ and $C\coloneqq\nu(\nu+1)/(R+1)$. The two roots in Eq.~\eqref{gr_symplectic_roots} obey $(p_\kappa+1)(q_\kappa-1)=4C$, so that $\kappa$ can be traded for
\bb
 y\coloneqq\frac{p_\kappa+1}{2},\qquad q_\kappa=1+\frac{2C}{y},\qquad
 \kappa=\frac{1}{2\nu+1}\left[1+\frac{2(R+1)}{R}\,u(y)\right],\qquad
 u(y)\coloneqq y-1-\frac Cy.
 \label{gr_scalar_variables}
\ee
Since $u'(y)=1+C/y^2>0$, $\kappa$ increases with $y$. The endpoint $\kappa=1$ is $y=\nu+1$, where $u=R\nu/(R+1)$, and the domain is $y\geq\nu+1$. In natural units, $\mathcal F\coloneqq(\ln2)F$ reads
\bb
 \mathcal F(y;R,\nu)&=\ln R+(\ln2)\left[g\!\left(\frac{u(y)}{R}\right)-g(y-1)-g\!\left(\frac Cy\right)\right],\\
 \mathcal F(\nu+1;R,\nu)&=\ln R-(\ln2)g(\nu),\qquad\lim_{y\to\infty}\mathcal F(y;R,\nu)=0,
 \label{gr_scalar_y}
\ee
and, since $(\ln2)g'(x)=\ln(1+1/x)$,
\bb
 \mathcal F_y=\frac{u'(y)}{R}\ln\!\left(1+\frac{R}{u(y)}\right)-\ln\!\left(1+\frac1{y-1}\right)+\frac C{y^2}\ln\!\left(1+\frac yC\right).
 \label{gr_first_derivative}
\ee
At the thermal threshold $R=R_{\mathrm G}(\nu)$ both endpoint values of $\mathcal F$ vanish and
\bb
 \mathcal F_y(\nu+1;R_{\mathrm G},\nu)=\frac{d(\nu)}{2\eta_{\mathrm G}(\nu)(\nu+1)},
 \label{gr_boundary_slope}
\ee
so that $\mathcal F_y(\nu+1)$, $F'(1)$ and $d(\nu)$ have the same sign.

\begin{grlemma}[(Curvature of the pair rate)]
\label{gr_shape_lem}
Let $\nu\geq1/4$ and $R=R_{\mathrm G}(\nu)$. On $y>\nu+1$, the function $\mathcal H(y)\coloneqq y^3\mathcal F_{yy}(y)$ first increases and then decreases, and $\mathcal H(y)\to-\infty$ as $y\to\infty$. Consequently $\mathcal F_{yy}$ changes sign at most twice, and its signs follow the order $-,+,-$, where the first one or two segments may be absent.
\end{grlemma}
\begin{proof}
Each of the three arguments of $g$ in Eq.~\eqref{gr_scalar_y} is a linear combination $f$ of $y$, $1$ and $1/y$, and therefore satisfies $f'''+3f''/y=0$. In $\mathcal H'(y)/y^3=\mathcal F_{yyy}+3\mathcal F_{yy}/y$ the terms proportional to $g'(f)$ then cancel, and only the rational functions $(\ln2)g''(x)=-1/[x(x+1)]$ and $(\ln2)g'''(x)=(2x+1)/[x(x+1)]^2$ remain. Collecting them gives
\bb
 \frac{\mathcal H'(y)}{y^3}=-\frac{(y^2+C)\,T(u)}{y^4u^2(u+C)^2(u+R)^2},\qquad u=u(y),
 \label{gr_J}
\ee
with the quartic
\bb
 T(u)\coloneqq{}&2Cu^4+(C^2+4CR-10C-R^2+3R-2)u^3+(2C^2R-17C^2-2CR-4C+R^2-R)u^2\\
 &-C(8C^2+7CR+2C+2R)u-C^2R(4C+1).
 \label{gr_shape_polynomial}
\ee
All factors multiplying $-T(u)$ in Eq.~\eqref{gr_J} are positive. Expand $T$ around the endpoint,
\bb
 T\!\left(\frac{R\nu}{R+1}+z\right)=\sum_{j=0}^4c_jz^j,\qquad c_4=2C>0.
 \label{gr_quartic_coefficients}
\ee
By Lemma~\ref{gr_signs_lem}(b), $c_0<0$, $c_1<0$ and $c_3^3>27c_4^2c_1$. For $z>0$, the arithmetic--geometric mean inequality applied to $c_4z$, $c_4z$ and $-c_1/z^2$ gives
\bb
 \frac{\dd}{\dd z}\left[\frac1{z^2}\,T\!\left(\frac{R\nu}{R+1}+z\right)\right]
 =2c_4z+c_3-\frac{c_1}{z^2}-\frac{2c_0}{z^3}
 \geq c_3+3\bigl(c_4^2|c_1|\bigr)^{1/3}-\frac{2c_0}{z^3}>0,
 \label{gr_quartic_monotone}
\ee
where $c_3+3(c_4^2|c_1|)^{1/3}>0$ is equivalent to $c_3^3>27c_4^2c_1$. The quotient increases from $-\infty$ to $+\infty$, so $T$ changes sign exactly once on $u>R\nu/(R+1)$, from negative to positive. Since $u$ increases with $y$, $\mathcal H'$ changes sign exactly once, from positive to negative. Moreover Eq.~\eqref{gr_J} gives $y\mathcal H'(y)\to-2C<0$, so $\mathcal H'(y)\leq-C/y$ for all large $y$ and $\mathcal H\to-\infty$. A function that increases and then decreases to $-\infty$ vanishes at most twice, and $\mathcal F_{yy}=\mathcal H/y^3$ has the stated sign pattern.
\end{proof}

Suppose now that $d(\nu)\leq0$, that is, $\mathcal F_y(\nu+1)\leq0$. By Eq.~\eqref{gr_first_derivative}, $\mathcal F_y\to0$ as $y\to\infty$; on the last segment of Lemma~\ref{gr_shape_lem}, where it decreases, $\mathcal F_y$ is therefore positive. Starting from a nonpositive value, $\mathcal F_y$ can then cross zero only once, on the segment where it increases (this segment is present: otherwise $\mathcal F_y$ would decrease towards zero on the whole half-line and $\mathcal F_y(\nu+1)$ would be positive): there is $y_0\geq\nu+1$ such that $\mathcal F_y\leq0$ on $[\nu+1,y_0]$ and $\mathcal F_y\geq0$ on $[y_0,\infty)$. As $\mathcal F$ vanishes at both ends of its domain,
\bb
 \mathcal F(y)=\begin{cases}
 \displaystyle\int_{\nu+1}^y\mathcal F_y(w)\,\dd w\leq0,&\nu+1\leq y\leq y_0,\\[5pt]
 \displaystyle-\int_y^\infty\mathcal F_y(w)\,\dd w\leq0,&y\geq y_0.
 \end{cases}
 \label{gr_global_no_positive}
\ee
Hence, at the thermal threshold, $d(\nu)\leq0$ implies $F(\kappa)\leq0$ for every $\kappa\geq1$.

\subsection{Extending the comparison to every transmissivity}\label{gr_all_eta_subsection}
Let $\nu\geq\nu_{\rm c}$. Then $d(\nu)\leq0$ by Lemma~\ref{gr_signs_lem}(a), and Eq.~\eqref{gr_global_no_positive} gives $F(\kappa)\leq0$ at $\eta=\eta_{\mathrm G}(\nu)$. Since $F(1)=0$ there, the supremum over $\kappa$ is zero, and Eq.~\eqref{gr_main} gives a vanishing regularised Gaussian optimum at the thermal threshold.

\emph{Below the thermal threshold.} A smaller transmissivity is obtained by composing with a second thermal attenuator with the same bath occupation, and data processing carries the bound over: for every Gaussian input $\rho$ and every $m$,
\bb
 \Phi_{\eta,\nu}=\Phi_{\eta/\eta_{\mathrm G}(\nu),\nu}\circ\Phi_{\eta_{\mathrm G}(\nu),\nu}
 \quad\Longrightarrow\quad
 \Ic(\rho,\Phi_{\eta,\nu}^{\otimes m})\leq\Ic(\rho,\Phi_{\eta_{\mathrm G}(\nu),\nu}^{\otimes m})\leq0,
 \qquad\eta\leq\eta_{\mathrm G}(\nu).
 \label{gr_composition}
\ee
Pure Gaussian inputs attain zero.

\emph{Above the thermal threshold.} Here the thermal rate $\ell$ is positive and data processing is not enough. We show instead that $F$ cannot grow with the transmissivity faster than $\ell$, which in natural units grows as $\ln R$. Keep $y\geq\nu+1$ and $\nu$ fixed, differentiate Eq.~\eqref{gr_scalar_y} with respect to $R$ and use $\ln(1+1/x)=\int_0^1(x+s)^{-1}\dd s$. Combining the two logarithmic terms gives
\bb
 R\,\partial_R\mathcal F(y;R,\nu)=1-(y-\nu-1)(y+\nu)\int_0^1
 \frac{ys(R+1)^2+\nu(\nu+1)}{\bigl[y(R+1)(y-1+Rs)-\nu(\nu+1)\bigr]\bigl[ys(R+1)+\nu(\nu+1)\bigr]}\,\dd s\leq1,
\ee
because the prefactor is nonnegative and the two factors in the denominator are bounded below by $R\nu(\nu+1)$ and $\nu(\nu+1)$, respectively. Integrating in $R$ from the threshold,
\bb
 \mathcal F(y;R,\nu)\leq\mathcal F(y;R_{\mathrm G}(\nu),\nu)+\ln\frac{R}{R_{\mathrm G}(\nu)}\leq\ln\frac{R}{R_{\mathrm G}(\nu)}=(\ln2)\,\ell.
 \label{gr_R_integration}
\ee
Thus $\sup_{\kappa\geq1}F(\kappa)\leq\ell\leq2\ell$, and Eq.~\eqref{gr_main} yields $\Ic^{\G,\infty}(\Phi_{\eta,\nu})=\ell$; product thermal inputs approach $m\ell$ on $m$ uses, and the regularised bound forbids more. Equation~\eqref{gr_main} also gives the test~\eqref{gr_exact_gap_test}. This proves Eq.~\eqref{gr_additive_regime} for every $\nu\geq\nu_{\rm c}$ and every transmissivity.

\subsection{An exact-arithmetic proof of the scalar signs}\label{gr_certificate_subsection}
The two scalar facts used above are the following.

\begin{grlemma}[(Scalar signs)]
\label{gr_signs_lem}
(a) The function $d$ in Eq.~\eqref{gr_threshold_function} has exactly one zero $\nu_{\rm c}$ in $(0,\infty)$. It is positive on $(0,\nu_{\rm c})$ and negative on $(\nu_{\rm c},\infty)$, and
\bb
 0.2961<\nu_{\rm c}<0.2962.
 \label{gr_critical_bracket}
\ee
(b) For every $\nu\geq1/4$ and $R=R_{\mathrm G}(\nu)$, the coefficients in Eq.~\eqref{gr_quartic_coefficients} satisfy
\bb
 c_0<0,\qquad c_1<0,\qquad c_3^3-27c_4^2c_1>0.
 \label{gr_quartic_signs}
\ee
\end{grlemma}
\begin{proof}
Analytic estimates cover the two unbounded ends of the parameter range, and outward-rounded interval arithmetic covers the compact remainder. All interval evaluations and symbolic identities are carried out by \texttt{Supplementary\_data/numerics/result\_gaussian\_regularization.py}.

\emph{Part (a).} Put $x\coloneqq(\nu+1)^{-1}\in(0,1)$, so that large occupations correspond to $x\to0$, and
\bb
 \rho(x)\coloneqq(1-x)^{-\frac{1-x}{x}}=xR_{\mathrm G}\!\left(\frac{1-x}{x}\right),\qquad
 q(x)\coloneqq(2-x)\ln\frac{1+\rho(x)}{1-x}+\frac{\rho(x)}x\ln(1-x),\qquad
 d\!\left(\frac{1-x}{x}\right)=\frac{2q(x)}{\rho(x)+x}.
 \label{gr_Q}
\ee
Thus $d$ and $q$ have the same sign. For $0<x\leq1/100$, the bounds $1\leq-\ln(1-x)/x\leq(1-x)^{-1}$ give $e^{99/100}\leq\rho(x)\leq e$ and $-\ln(1-x)\leq\ln(100/99)$, hence
\bb
 q(x)\leq2\left[\ln(1+e)+\ln\frac{100}{99}\right]-e^{99/100}<-0.04.
 \label{gr_large_tail}
\ee
For $0<\nu\leq1/100$, write $\Lambda\coloneqq\ln(1/\nu)\geq\ln100$ and $a\coloneqq\exp[(\ln100+101/100)/100]<3$. From $(\ln2)g(\nu)=(1+\nu)\ln(1+\nu)+\nu\Lambda\leq\nu(\Lambda+101/100)$, where the right-hand side increases with $\nu$ on this interval, and from $e^z-1\leq ze^z$, we get $0\leq R_{\mathrm G}(\nu)-1\leq a\,\nu(\Lambda+101/100)$ and $R_{\mathrm G}(\nu)\leq a$. Splitting $\ln[1+(1+R_{\mathrm G})/\nu]=\ln(1+1/\nu)+\ln[1+R_{\mathrm G}/(1+\nu)]$ in Eq.~\eqref{gr_threshold_function},
\bb
 d(\nu)=\frac{2}{1+R_{\mathrm G}}\left\{\nu\ln\!\left(1+\frac1\nu\right)+(2\nu+1)\ln\!\left(1+\frac{R_{\mathrm G}}{1+\nu}\right)-(R_{\mathrm G}-1)(\nu+1)\ln\!\left(1+\frac1\nu\right)\right\}.
\ee
We drop the first term and use $2\nu+1\geq1$, $1-(R_{\mathrm G}-1)/2\leq2/(1+R_{\mathrm G})\leq1$, $\ln[1+R_{\mathrm G}/(1+\nu)]\geq\ln[(2+\nu)/(1+\nu)]\geq\ln2-\nu/2$ and $\ln(1+1/\nu)\leq\Lambda+1/100$. The product $\nu(\Lambda+101/100)(\Lambda+1/100)$ increases with $\nu$ on the interval, its derivative being $\Lambda^2-\tfrac{98}{100}\Lambda-\tfrac{10099}{10000}>0$. Evaluating all bounds at $\nu=1/100$,
\bb
 d(\nu)\geq\frac{3-a}{2}\left(\ln2-\frac1{200}\right)-\frac{101}{10000}\,a\left(\ln100+\frac{101}{100}\right)\left(\ln100+\frac1{100}\right)>0.39.
 \label{gr_small_tail}
\ee
On the remaining compact ranges, interval evaluation on uniform rational panels gives
\bb
 \begin{array}{c|c}
 \text{range}&\text{bound}\\ \hline
 1/100\leq x\leq3/4&q(x)<-0.012\\
 3/4\leq x\leq4/5&q'(x)>0.55\\
 1/100\leq\nu\leq1/4&d(\nu)>0.016
 \end{array}
 \label{gr_compact_signs}
\ee
together with $q(3/4)<0<q(4/5)$; here $q'$ is evaluated from the closed form obtained by differentiating Eq.~\eqref{gr_Q}, with $\rho'(x)=\rho(x)[\ln(1-x)+x]/x^2$, as implemented in the verifier. Hence $q<0$ for $x\leq3/4$, $q$ is strictly increasing with exactly one zero on $[3/4,4/5]$, and $d>0$ for $\nu\leq1/4$, that is, for $x\geq4/5$. This gives the unique zero and the signs on either side. A sign test at two rational points in $(0.2961,0.2962)$, at which $d$ takes opposite signs, gives Eq.~\eqref{gr_critical_bracket}.

\emph{Part (b).} In the same coordinate, $\nu\geq1/4$ becomes $0<x\leq4/5$ and $R=\rho(x)/x$. The rescaled coefficients
\bb
 \widehat c_j\coloneqq x^{5-j}[\rho(x)+x]^5c_j,\qquad j=0,\dots,4,
 \label{gr_compact_polynomial}
\ee
are polynomials in $x$ and $\rho(x)$ with integer coefficients, and
$\widehat c_3^{\,3}-27\widehat c_4^{\,2}\widehat c_1=x^6[\rho(x)+x]^{15}(c_3^3-27c_4^2c_1)$, so that Eq.~\eqref{gr_quartic_signs} is equivalent to
\bb
 \widehat c_0<0,\qquad\widehat c_1<0,\qquad\widehat c_3^{\,3}-27\widehat c_4^{\,2}\widehat c_1>0.
 \label{gr_compact_quartic_signs}
\ee
Since $\frac{\dd}{\dd x}\ln\rho(x)=[\ln(1-x)+x]/x^2<0$, on a panel $[x_1,x_2]$ one has $\rho(x)\in[\rho(x_2),\rho(x_1)]$, with $\rho(0)=e$ by continuity. The verifier checks Eq.~\eqref{gr_compact_quartic_signs} as interval inequalities on a partition of $[0,4/5]$ into rational panels obtained by bisection.

\emph{Arithmetic.} Interval endpoints are integer multiples of $2^{-224}$ and every operation is rounded outwards in integer arithmetic. Logarithms and exponentials are enclosed by truncated Taylor series with explicit remainder bounds, and the square roots in Eq.~\eqref{gr_example} by integer square roots. The polynomial identities, including Eq.~\eqref{gr_J}, are checked symbolically. No floating-point value enters a sign test.
\end{proof}

Lemmas~\ref{gr_shape_lem} and~\ref{gr_signs_lem} complete the proof of Proposition~\ref{gaussian_noise_transition_prop}. For thermal attenuation, thermal inputs are Gaussian-optimal at every transmissivity when $\nu\geq\nu_{\rm c}$; for $0<\nu<\nu_{\rm c}$, separable pairs improve the rate at and just below the thermal threshold. The normal forms and Lemma~\ref{gr_nonattenuating_additivity_lem} rule out Gaussian superadditivity for all other single-mode Gaussian channels. We next construct non-Gaussian inputs with positive rates below the thermal threshold, even where Gaussian inputs with arbitrarily many modes cannot give a positive rate.

\begingroup
\renewcommand{\thetable}{S\arabic{table}}

\clearpage
\section{Positive quantum capacity beyond the thermal threshold}
\label{analytic_positivity_section}

\subsection{Analytic construction}
\label{analytic_construction_subsection}

We prove Theorem~\ref{analytic_thm} by constructing finite-energy inputs with positive coherent information at the thermal threshold and below it, for every nonzero bath occupation. The construction uses an equal mixture of two displaced squeezed states.

The useful quantity is the uncertainty about which state was sent. Since each component is pure, the coherent information of the mixture equals the environment's conditional uncertainty minus Bob's. At large separation, both conditional uncertainties tend to zero. We choose the squeezing so that Bob's tends to zero faster; their difference is then positive at a finite separation.

We achieve this with perpendicular squeezing axes and opposite displacements along $q=p$, as shown in Fig.~\ref{squeezed_pair_figure}. The symmetry of the pair lets us compare Gaussian overlaps instead of evaluating the entropies of the non-Gaussian output mixtures. The proof has two steps: choose a finite squeezing for which Bob has the larger decay exponent, then choose a finite separation large enough to make his conditional uncertainty smaller. This works whenever $\eta_*(\nu)<\eta<1$, with $\eta_*(\nu)<\eta_{\rm G}(\nu)$ as shown below. Both choices remain finite, so the input has finite mean photon number.

We use $[\hat q,\hat p]=i$ and define the covariance matrix by $V_{jk}=\langle\{\Delta\hat R_j,\Delta\hat R_k\}\rangle$, where $\hat{\boldsymbol R}=(\hat q,\hat p)^{\mathsf T}$ and $\Delta\hat R_j=\hat R_j-\langle\hat R_j\rangle$. The standard Gaussian-state formulas used below are reviewed in Refs.~\cite{Weedbrook2012,Serafini2017}.

\begin{figure}[htbp]
\centering
\begin{tikzpicture}[x=0.68cm,y=0.68cm,>=Latex,font=\small]
  \pgfmathsetmacro{\paircentre}{5/(2*sqrt(2))}
  \draw[->,black!55] (-4.1,0) -- (4.25,0) node[right,black] {$q$};
  \draw[->,black!55] (0,-4.1) -- (0,4.25) node[above,black] {$p$};
  \draw[densely dashed,black!45] (-3.4,-3.4) -- (3.55,3.55)
    node[above right,black] {$q=p$};
  \filldraw[fill=black!10,draw=black,line width=0.75pt]
    (-\paircentre,-\paircentre) ellipse [x radius=0.5,y radius=2];
  \filldraw[fill=black!22,draw=black,line width=0.75pt]
    (\paircentre,\paircentre) ellipse [x radius=2,y radius=0.5];
  \draw[<->,line width=0.65pt,shorten <=1.5pt,shorten >=1.5pt]
    (-\paircentre,-\paircentre) -- (\paircentre,\paircentre);
  \node[fill=white,inner sep=1.5pt] at (0.60,0.60) {$d$};
  \fill (-\paircentre,-\paircentre) circle (1.5pt);
  \fill (\paircentre,\paircentre) circle (1.5pt);
  \node[left] at (-2.4,-1.1) {$\ket{\varphi_0}$};
  \node[above] at (2.8,2.4) {$\ket{\varphi_1}$};
\end{tikzpicture}
\caption{Mixing two displaced squeezed states with equal probabilities gives a single-mode non-Gaussian input. The ellipses show the two states in phase space: their squeezing axes are perpendicular, and their centres (black dots) are separated by $d$ along $q=p$. For the channels in Theorem~\ref{analytic_thm}, suitable squeezing and separation leave the receiver less uncertain than the environment about which state was sent. This gives positive coherent information where thermal inputs give zero rate.}
\label{squeezed_pair_figure}
\end{figure}

We first establish the two auxiliary facts used in the construction: an overlap bound for binary uncertainty and the ordering of the channel thresholds.

\begin{lemma}[(Conditional entropy of a binary ensemble)]
\label{binary_entropy_bound_lemma}
Let $X$ label two equiprobable states $\sigma_0,\sigma_1$ on $B$, both with finite entropy. Suppose that $C(t)=\Tr\big[\sigma_0^t\sigma_1^{1-t}\big]$ is positive and analytic near $t=1/2$, with $C'(1/2)=0$. Write $C=C(1/2)$ and $\mathcal V=\left.\partial_t^2\ln C(t)\right|_{t=1/2}$. Then
\bb
\frac{C}{4\ln2}e^{-\sqrt{\mathcal V}/2}\leq H(X|B)\leq h_2(C/2),
\label{supp_binary_entropy_bound}
\ee
where $h_2$ denotes the binary entropy.
\end{lemma}

\begin{proof}
Let $P_{\rm e}$ be the minimum probability of error in guessing $X$ from $B$. Fano's inequality and the conditional min-entropy bound~\cite{KoenigRennerSchaffner2009,FurrerAbergRenner2011,WildeBook2017} give
\bb
-\log_2(1-P_{\rm e})\leq H(X|B)\leq h_2(P_{\rm e}).
\label{supp_entropy_error}
\ee
The upper bound in Eq.~\eqref{supp_binary_entropy_bound} follows from the quantum Chernoff bound $P_{\rm e}\leq C/2$~\cite{Audenaert2007Chernoff}, since $h_2$ is increasing on $[0,1/2]$.

For the lower bound, write $\sigma_0=\sum_i\alpha_i\proj{u_i}$ and $\sigma_1=\sum_j\beta_j\proj{v_j}$ in complete eigenbases, including their kernels. The inequality of Nussbaum and Szko\l a~\cite[Eq.~(10)]{NussbaumSzkola2009} gives
\bb
P_{\rm e}\geq\frac14\sum_{ij}\min(\alpha_i,\beta_j)|\langle u_i|v_j\rangle|^2.
\label{supp_spectral_error}
\ee
Termwise differentiation of the spectral overlap series is justified near $t=1/2$ because powers of $|\ln(\alpha_i/\beta_j)|$ are bounded by constant multiples of $(\alpha_i/\beta_j)^\delta+(\alpha_i/\beta_j)^{-\delta}$ for any fixed $\delta>0$, and H\"older's inequality bounds the resulting overlap series at nearby exponents inside $(0,1)$. These bounds also hold on separable Hilbert spaces. To see this, truncate each state spectrally and apply the finite-dimensional bounds on the span of both truncated supports; this also covers the min-entropy bound, stated in finite dimension in Ref.~\cite{KoenigRennerSchaffner2009}. The normalised truncations converge in trace norm, hence so do their optimal discrimination errors. The spectral sums converge by dominated convergence, and the overlaps converge because the square roots converge in Hilbert--Schmidt norm. The classical register $X$ remains two-dimensional, so the conditional-entropy continuity bound~\cite[Lemma~2]{Winter2016Continuity} gives convergence of $H(X|B)$ independently of the dimension of $B$. Equation~\eqref{supp_entropy_error} therefore passes to the limit.

To express the right-hand side of Eq.~\eqref{supp_spectral_error} in terms of $C$ and $\mathcal V$, define $r_{ij}=\sqrt{\alpha_i\beta_j}|\langle u_i|v_j\rangle|^2/C$, omitting terms of zero weight; this is a probability distribution because $C=\Tr\big[\sqrt{\sigma_0}\sqrt{\sigma_1}\big]$. Differentiating $C(t)=\sum_{ij}\alpha_i^t\beta_j^{1-t}|\langle u_i|v_j\rangle|^2$ shows that the random variable $L_{ij}=\ln(\alpha_i/\beta_j)$ has mean $C'(1/2)/C=0$ and variance $\mathcal V$ under $r$.

Since $\min(\alpha_i,\beta_j)=\sqrt{\alpha_i\beta_j}\,e^{-|L_{ij}|/2}$ whenever $\alpha_i\beta_j>0$, Eq.~\eqref{supp_spectral_error} gives $P_{\rm e}\geq(C/4)\mathbb E_r[e^{-|L|/2}]$. Jensen's inequality gives $\mathbb E_r[e^{-|L|/2}]\geq e^{-\mathbb E_r|L|/2}$, while Cauchy--Schwarz gives $\mathbb E_r|L|\leq\sqrt{\mathcal V}$, so that $P_{\rm e}\geq(C/4)e^{-\sqrt{\mathcal V}/2}$. Combining this estimate with Eq.~\eqref{supp_entropy_error} and $-\log_2(1-x)\geq x/\ln2$ proves the lower bound in Eq.~\eqref{supp_binary_entropy_bound}.
\end{proof}

\begin{lemma}[(Comparison with the thermal threshold)]
\label{threshold_order_lemma}
For every $\nu>0$, $\eta_{\rm AD}(\nu)<\eta_*(\nu)<\eta_{\rm G}(\nu)$.
\end{lemma}
\begin{proof}
The first inequality follows from $\eta_*(\nu)-\eta_{\rm AD}(\nu)=\nu/[4(2\nu+1)(\nu+1)]>0$.

For the second, we set $F(\nu)=2^{g(\nu)}$; since $\eta_{\rm G}(\nu)=F(\nu)/(1+F(\nu))$, we need to show that
\bb
 F(\nu)>\frac{8\nu^2+9\nu+2}{3\nu+2}.
 \label{supp_threshold_comparison}
\ee
The function $F$ is strictly concave:
\bb
 \frac{F''(\nu)}{F(\nu)}=\ln^2(1+1/\nu)-\frac{1}{\nu(\nu+1)}<0,
\ee
The last inequality follows from $\sinh x>x$ for $x=\tfrac12\ln(1+1/\nu)>0$. On $0<\nu\leq1$, concavity puts $F$ above the chord joining $F(0)=1$ and $F(1)=4$. For $\nu\geq1$, its derivative decreases to $e>8/3$, so integrating from $\nu=1$ gives the second bound below. Thus
\bb
 F(\nu)\geq\begin{cases}1+3\nu,&0<\nu\leq1,\\(8\nu+4)/3,&\nu\geq1.\end{cases}
\ee
Each bound is strictly larger than the right-hand side of Eq.~\eqref{supp_threshold_comparison}:
\bb
 1+3\nu-\frac{8\nu^2+9\nu+2}{3\nu+2}=\frac{\nu^2}{3\nu+2}>0,\qquad
 \frac{8\nu+4}{3}-\frac{8\nu^2+9\nu+2}{3\nu+2}=\frac{\nu+2}{3(3\nu+2)}>0.
\ee
This proves the second inequality.
\end{proof}

\begin{samepage}
We can now prove the main statement.

\begin{manualthm}{\ref{analytic_thm}}
For every $\nu>0$, the quantum capacity is strictly positive at the thermal threshold,
\bb
 Q(\Phi_{\eta_{\rm G}(\nu),\nu})>0,
 \label{supp_capacity_at_threshold}
\ee
and remains positive below it:
\bb
 Q(\Phi_{\eta,\nu})>0
 \qquad\text{for }\eta_*(\nu)<\eta<\eta_{\rm G}(\nu),
 \label{supp_capacity_below_threshold}
\ee
where
\bb
 \eta_*(\nu)=\frac{8\nu^2+9\nu+2}{4(2\nu+1)(\nu+1)},
 \qquad
 \eta_{\rm AD}(\nu)<\eta_*(\nu)<\eta_{\rm G}(\nu).
 \label{supp_threshold_order}
\ee
For $\nu\geq\nu_{\rm c}$, the finite-energy single-mode non-Gaussian inputs constructed below outperform Gaussian inputs of every blocklength throughout $\eta_*(\nu)<\eta\leq\eta_{\rm G}(\nu)$, where $\Ic^{\G,\infty}(\Phi_{\eta,\nu})=0$.
\end{manualthm}
\end{samepage}

\begin{proof}
We fix $\nu>0$ and $\eta_*(\nu)<\eta<1$, and construct a single-mode state $\rho_{s,d}$, with finite parameters $s,d$, such that
\bb
 \Ic(\rho_{s,d},\Phi_{\eta,\nu})>0.
 \label{supp_positive_input}
\ee
This suffices because $Q\geq\Ic$ for every input of finite energy~\cite[Theorem~5]{WildeQi2018}. The hypotheses of that coding theorem hold: the photon-number observable obeys the Gibbs condition, and a bound on the input energy bounds both output energies and hence both output entropies. Since Lemma~\ref{threshold_order_lemma} gives $\eta_*(\nu)<\eta_{\rm G}(\nu)$, this construction covers both positivity claims.

\paragraph{1. The input and the entropy balance.---}
We consider two pure Gaussian states $\ket{\varphi_0}$ and $\ket{\varphi_1}$ with covariance matrices
\bb
 V_0=\begin{pmatrix}s^{-1}&0\\0&s\end{pmatrix},
 \qquad
 V_1=\begin{pmatrix}s&0\\0&s^{-1}\end{pmatrix},
\ee
and first-moment vectors
\bb
 \boldsymbol\mu_0=-\frac{d}{2\sqrt2}\binom11,
 \qquad
 \boldsymbol\mu_1=\frac{d}{2\sqrt2}\binom11.
 \label{supp_squeezed_pair}
\ee
Here $s>1$ controls the squeezing, while $d>0$ is the distance between the two means. We take their equal mixture,
\bb
 \rho_{s,d}\coloneqq
 \tfrac12\proj{\varphi_0}+\tfrac12\proj{\varphi_1}.
 \label{supp_squeezed_mixture}
\ee
Its mean photon number is
\bb
 \Tr[\rho_{s,d}\hat n]
 =\frac{s+s^{-1}-2}{4}+\frac{d^2}{8},
\ee
which is finite whenever $s,d$ are finite. The Wigner function of this mixture is positive, since both components are Gaussian. The non-Gaussian advantage proved here therefore does not require Wigner negativity.

Let $X$ record which state was chosen. For $j=0,1$, we denote by $\sigma_{B,j}$ and $\sigma_{E,j}$ the corresponding outputs at Bob and the environment, respectively. Here $E$ includes the mode purifying the thermal bath.

For either $Y=B,E$, the joint classical--quantum state is $\frac12\sum_{j=0}^1\proj{j}_{X}\otimes\sigma_{Y,j}$. The two classical blocks are orthogonal, so
\bb
 H(X|Y)
 =1+\frac12\sum_{j=0}^1S(\sigma_{Y,j})
 -S\!\left(\frac{\sigma_{Y,0}+\sigma_{Y,1}}2\right).
 \label{supp_cq_entropy}
\ee
Each input component is pure, so its joint output on $BE$ is pure as well and $S(\sigma_{B,j})=S(\sigma_{E,j})$ for both $j$. These component entropies cancel when we subtract Eq.~\eqref{supp_cq_entropy} for $B$ and $E$, giving the pure-ensemble information balance~\cite{SchumacherWestmoreland1998}:
\bb
H(X|E)-H(X|B)
 &=S(\Phi_{\eta,\nu}(\rho_{s,d})) -S(\Phi_{\eta,\nu}^{\rm c}(\rho_{s,d})) =\Ic(\rho_{s,d},\Phi_{\eta,\nu}).
\label{supp_binary_balance}
\ee
For the binary squeezed-state mixture, positive coherent information is therefore equivalent to Bob having less conditional uncertainty than the purified environment about Alice's choice. We next compare these two conditional entropies as the displacement increases.

\paragraph{2. Comparing the two uncertainties.---}
We first fix the squeezing $s$ and study what happens as the separation $d$ increases.

For $Y=B,E$, we consider the overlap
\bb
 C_Y(t,d)=
 \Tr[\sigma_{Y,0}(d)^t\sigma_{Y,1}(d)^{1-t}],
 \qquad 0<t<1.
 \label{supp_gaussian_overlap}
\ee

For fixed squeezing $s$, the Gaussian overlap formula~\cite[Theorem~2]{PirandolaLloyd2008} makes $\ln C_Y(t,d)$ affine in $d^2$, with coefficients analytic near $t=1/2$. In particular,
\bb
 C_Y(1/2,d)=C_Y(1/2,0)e^{-\kappa_Yd^2},
 \label{supp_overlap_decay}
\ee
for some $\kappa_Y>0$ that depends on the squeezing and the channel parameters, but not on $d$.

We now show that
\bb
 C_Y(t,d)=C_Y(1-t,d),\qquad Y=B,E,
 \label{supp_overlap_symmetry}
\ee
which implies $\partial_t C_Y(1/2,d)=0$. This identity follows from a symmetry of the input pair: the antisymplectic reflection $(q,p)\mapsto(-p,-q)$ exchanges $\ket{\varphi_0}$ and $\ket{\varphi_1}$.

The map $(q,p)\mapsto(-p,-q)$, applied to the input and to the bath mode, commutes with the beam splitter and leaves the thermal state of the bath invariant; combined with $(q_{E'},p_{E'})\mapsto(p_{E'},q_{E'})$ on the purifying mode, it also leaves the two-mode squeezed vacuum invariant. The conditional outputs inherit this symmetry: for $Y=B,E$, their covariance matrices and first moments are related by an antisymplectic reflection $R_Y$, namely $V_{Y,1}=R_YV_{Y,0}R_Y^{\mathsf T}$ and $\boldsymbol\mu_{Y,1}=R_Y\boldsymbol\mu_{Y,0}$. Since the antiunitaries implementing these reflections preserve positive powers and the trace of products of positive operators, Eq.~\eqref{supp_overlap_symmetry} follows.

Lemma~\ref{binary_entropy_bound_lemma} now applies. Since $\partial_t^2\ln C_Y(t,d)|_{t=1/2}=O(1+d^2)$, we obtain
\bb
H(X|B)\leq O(d^2)e^{-\kappa_Bd^2},\qquad
 H(X|E)\geq e^{-\kappa_Ed^2-O(d)}
  \qquad(d\to\infty).
 \label{supp_entropy_bounds}
\ee
Here and below, the implicit constants may depend on $s,\eta,\nu$, but not on $d$.

A strict inequality between the decay exponents is enough to determine the sign for large separation. If $\kappa_B>\kappa_E$, then
\bb
 \frac{H(X|B)}{H(X|E)}
 \leq O(d^2)e^{-(\kappa_B-\kappa_E)d^2+O(d)}
 \longrightarrow0.
\ee
For sufficiently large $d$, Bob is then strictly less uncertain about Alice's choice. Equation~\eqref{supp_binary_balance} gives
\bb
 \kappa_B>\kappa_E
 \quad\Longrightarrow\quad
 \Ic(\rho_{s,d},\Phi_{\eta,\nu})>0
 \quad\text{for sufficiently large finite }d.
 \label{supp_exponent_criterion}
\ee
If Bob's overlap decays faster than the environment's, $\kappa_B>\kappa_E$, the binary squeezed-state mixture has positive coherent information at sufficiently large finite separation. We now choose a finite squeezing $s$ that ensures this strict inequality.

\paragraph{3. Choosing the squeezing.---}
We compute the two exponents and compare their large-$s$ expansions at fixed $\eta$ and $\nu$. The calculation uses the Gaussian power and overlap formulas of Ref.~\cite[Lemma~1 and Theorem~2]{PirandolaLloyd2008}.

\emph{Output covariances.}
Set $w=2\nu+1$. For an input with covariance $\operatorname{diag}(z,z^{-1})$, the beam-splitter dilation gives Bob's covariance matrix~\cite{Serafini2017}:
\bb
 B(z)=\operatorname{diag}\bigl(\eta z+(1-\eta)w,\eta/z+(1-\eta)w\bigr).
 \label{supp_complement_covariance}
\ee
The environment consists of the other beam-splitter output $F$ and the mode $E'$ purifying the bath. In the order $(q_F,q_{E'},p_F,p_{E'})$, its covariance is $Q(z)\oplus P(z)$, where
\bb
 Q(z)=\begin{pmatrix}(1-\eta)z+\eta w&\sqrt{\eta(w^2-1)}\\ \sqrt{\eta(w^2-1)}&w\end{pmatrix},\qquad
 P(z)=\begin{pmatrix}(1-\eta)/z+\eta w&-\sqrt{\eta(w^2-1)}\\ -\sqrt{\eta(w^2-1)}&w\end{pmatrix}.
 \label{supp_complement_blocks}
\ee
The differences between the output means are $\Delta\boldsymbol\mu_B=\sqrt\eta\,d(1,1)^{\mathsf T}/\sqrt2$ and, up to an overall sign, $\Delta\boldsymbol\mu_E=\sqrt{1-\eta}\,d(1,0,1,0)^{\mathsf T}/\sqrt2$.

For $z=s,s^{-1}$, Bob's two covariance matrices have the same determinant. Their common symplectic eigenvalue $v_s$ satisfies
\bb
 v_s^2=(\eta s+(1-\eta)w)(\eta/s+(1-\eta)w)=\eta^2+(1-\eta)^2w^2+\eta(1-\eta)w(s+s^{-1}).
 \label{supp_output_symplectic_value}
\ee
Since $s+s^{-1}\geq2$, we have $v_s\geq\eta+(1-\eta)w>1$. Each joint output on $BE$ is pure, so its two subsystems have the same nontrivial symplectic eigenvalues~\cite{Serafini2017}. Bob has one mode and the environment has two; the latter therefore has symplectic eigenvalues $(v_s,1)$. The eigenvalue $1$ refers to a pure normal mode, not to the physical mode $E'$ alone.

\emph{Covariances of the square roots.}
The normalised square root $\sqrt\sigma/\Tr\sqrt\sigma$ of a Gaussian state has the same mean and symplectic eigenvalues obtained by replacing each $v$ with $v+\sqrt{v^2-1}$~\cite[Lemma~1]{PirandolaLloyd2008}; see also \cite[Corollary~4]{SeshadreesanLamiWilde2018} for faithful states. Denoting the square-root covariance by a hat, we obtain
\bb
 \widehat B(z)=\bigl(1+\sqrt{1-v_s^{-2}}\bigr)B(z),\qquad
 \widehat Q(z)=Q(z)+\frac{Q(z)-P(z)^{-1}}{\sqrt{1-v_s^{-2}}},\qquad z=s,s^{-1}.
 \label{supp_square_root_covariances}
\ee
The first identity is the one-mode rescaling. For the second, choose Williamson coordinates that preserve the $q,p$ blocks, so that both blocks equal $\operatorname{diag}(v_s,1)$ and
\bb
 \frac{Q-P^{-1}}{\sqrt{1-v_s^{-2}}}=\operatorname{diag}\bigl(\sqrt{v_s^2-1},0\bigr)
 \qquad\text{in these coordinates}.
\ee
The added term therefore changes only the mixed normal mode. Since $Q$ and $P^{-1}$ transform by the same congruence, the identity holds in the original coordinates as well. The pure mode needs no limiting argument: its projector is its own square root. Interchanging $Q$ and $P$ gives the formula for $\widehat P(z)$.

\emph{The two exponents.}
At fixed $s$, the covariance-dependent prefactors cancel in the ratio of overlaps. In our quadrature convention, the Gaussian overlap formula reads~\cite[Theorem~2]{PirandolaLloyd2008}
\bb
 \frac{C_Y(1/2,d)}{C_Y(1/2,0)}=\exp\!\left[-\Delta\boldsymbol\mu_Y^{\mathsf T}(\widehat V_{Y,0}+\widehat V_{Y,1})^{-1}\Delta\boldsymbol\mu_Y\right],\qquad Y=B,E,
 \label{supp_overlap_quadratic_form}
\ee
where $\widehat V_{Y,j}$ is the covariance of the normalised square root of $\sigma_{Y,j}$. For Bob,
\bb
 \widehat B(s)+\widehat B(s^{-1})=\bigl(1+\sqrt{1-v_s^{-2}}\bigr)[\eta(s+s^{-1})+2(1-\eta)w]\mathds{1}_2.
\ee
Since $\|\Delta\boldsymbol\mu_B\|^2=\eta d^2$, we can compute $\kappa_B$ directly from the inverse of this matrix.

For the environment, put $M=\widehat Q(s)+\widehat Q(s^{-1})$. The sum of the $p$-blocks has the same diagonal entries as $M$ and opposite off-diagonal entries, so
\bb
 \widehat V_{E,0}+\widehat V_{E,1}=M\oplus\begin{pmatrix}M_{11}&-M_{12}\\-M_{12}&M_{22}\end{pmatrix}.
\ee
The upper-left entries of the inverses of these two blocks are equal, so each displaced quadrature contributes $(1-\eta)d^2(M^{-1})_{11}/2$ to Eq.~\eqref{supp_overlap_quadratic_form}. Comparing with Eq.~\eqref{supp_overlap_decay} gives
\bb
 \kappa_B=\frac{\eta}{(1+\sqrt{1-v_s^{-2}})[\eta(s+s^{-1})+2(1-\eta)w]},\qquad
 \kappa_E=(1-\eta)(M^{-1})_{11}.
 \label{supp_explicit_exponents}
\ee
It is simpler to compare the reciprocals of these positive exponents: their leading terms will cancel. Inverting the positive $2\times2$ matrix $M$ gives
\bb
 \kappa_E^{-1}=\frac{M_{11}M_{22}-M_{12}^2}{(1-\eta)M_{22}}=\frac{1}{1-\eta}\left(M_{11}-\frac{M_{12}^2}{M_{22}}\right).
 \label{supp_environment_reciprocal}
\ee

\emph{Expansion for large squeezing.}
We need the entries of $M$ through their constant terms. Equation~\eqref{supp_complement_blocks} gives $\det P(z)=\eta+(1-\eta)w/z$ and hence
\bb
 P(z)^{-1}=\frac{1}{\eta+(1-\eta)w/z}\begin{pmatrix}w&\sqrt{\eta(w^2-1)}\\ \sqrt{\eta(w^2-1)}&(1-\eta)/z+\eta w\end{pmatrix}.
 \label{supp_inverse_P}
\ee
For the two values of $z$,
\bb
 P(s)^{-1}=\begin{pmatrix}w/\eta&\sqrt{\eta(w^2-1)}/\eta\\ \sqrt{\eta(w^2-1)}/\eta&w\end{pmatrix}+O(s^{-1}),\qquad
 P(s^{-1})^{-1}=\begin{pmatrix}0&0\\0&w^{-1}\end{pmatrix}+O(s^{-1}).
\ee
The other matrix sum is exact:
\bb
 Q(s)+Q(s^{-1})=\begin{pmatrix}(1-\eta)(s+s^{-1})+2\eta w&2\sqrt{\eta(w^2-1)}\\2\sqrt{\eta(w^2-1)}&2w\end{pmatrix}.
\ee
Using Eq.~\eqref{supp_output_symplectic_value}, we have $v_s^{-2}=[\eta(1-\eta)ws]^{-1}+O(s^{-2})$. Hence
\bb
 \sqrt{1-v_s^{-2}}=1-\frac{1}{2\eta(1-\eta)ws}+O(s^{-2}),\qquad
 \frac{1}{\sqrt{1-v_s^{-2}}}=1+\frac{1}{2\eta(1-\eta)ws}+O(s^{-2}).
 \label{supp_root_expansions}
\ee
We must retain the correction of order $s^{-1}$ because the upper-left entry of $Q(s)$ grows as $(1-\eta)s$. Its contribution is
\bb
 \frac{Q(s)+Q(s^{-1})}{2\eta(1-\eta)ws}=\operatorname{diag}\!\left(\frac{1}{2\eta w},0\right)+O(s^{-1}).
\ee
The inverse matrices $P(s)^{-1}$ and $P(s^{-1})^{-1}$ remain bounded, so their products with the same $s^{-1}$ correction enter only the remainder. Summing Eq.~\eqref{supp_square_root_covariances} for $z=s,s^{-1}$ gives
\bb
 M=2[Q(s)+Q(s^{-1})]-P(s)^{-1}-P(s^{-1})^{-1}+\operatorname{diag}\!\left(\frac{1}{2\eta w},0\right)+O(s^{-1}).
 \label{supp_M_expansion}
\ee
All matrix remainders are entrywise. Substitution gives the three required entries:
\begin{align}
 M_{11}&=2(1-\eta)s+4\eta w-\frac{w}{\eta}+\frac{1}{2\eta w}+O(s^{-1}),\notag\\
 M_{12}&=\sqrt{\eta(w^2-1)}(4-\eta^{-1})+O(s^{-1}),\qquad
 M_{22}=3w-w^{-1}+O(s^{-1}).
 \label{supp_asymptotic_blocks}
\end{align}
For Bob, Eqs.~\eqref{supp_explicit_exponents} and~\eqref{supp_root_expansions} give
\bb
 \kappa_B^{-1}
 =\left(2-\frac{1}{2\eta(1-\eta)ws}+O(s^{-2})\right)\left(s+s^{-1}+\frac{2(1-\eta)w}{\eta}\right)
 =2s+\frac{4(1-\eta)w}{\eta}-\frac{1}{2\eta(1-\eta)w}+O(s^{-1}).
\label{supp_bob_reciprocal_expansion}
\ee
For the environment, we insert Eq.~\eqref{supp_asymptotic_blocks} into Eq.~\eqref{supp_environment_reciprocal}; since $3w-w^{-1}>0$, the denominator stays away from zero, and we obtain
\bb
 \kappa_E^{-1}=2s+\frac{1}{1-\eta}\left[4\eta w-\frac{w}{\eta}+\frac{1}{2\eta w}-\frac{\eta(w^2-1)(4-\eta^{-1})^2}{3w-w^{-1}}\right]+O(s^{-1}).
 \label{supp_environment_reciprocal_expansion}
\ee

\emph{Comparison.}
The common term $2s$ cancels when we subtract Eq.~\eqref{supp_bob_reciprocal_expansion} from Eq.~\eqref{supp_environment_reciprocal_expansion}:
\bb
 \lim_{s\to\infty}(\kappa_E^{-1}-\kappa_B^{-1})=\frac{1}{\eta(1-\eta)}\left[(8\eta-5)w+\frac1w-\frac{w(w^2-1)(4\eta-1)^2}{3w^2-1}\right].
 \label{supp_reciprocal_difference_unfactored}
\ee
To determine the sign, we put the terms in brackets over the common denominator $w(3w^2-1)$ and factor the numerator as
\bb
 [(8\eta-5)w^2+1](3w^2-1)-w^2(w^2-1)(4\eta-1)^2
 =[4\eta w(w+1)-(4w^2+w-1)][4(1-\eta)w(w-1)+3w-1].
\ee
Using $\eta_*(\nu)=(4w^2+w-1)/[4w(w+1)]$ in the first factor, we conclude that
\bb
 \lim_{s\to\infty}(\kappa_E^{-1}-\kappa_B^{-1})=\frac{4(w+1)[4(1-\eta)w(w-1)+3w-1]}{\eta(1-\eta)(3w^2-1)}[\eta-\eta_*(\nu)].
 \label{supp_exponent_difference}
\ee
The prefactor is strictly positive for $w>1$ and $0<\eta<1$, so the limit is positive because $\eta>\eta_*(\nu)$. Thus $\kappa_E^{-1}>\kappa_B^{-1}$, and hence $\kappa_B>\kappa_E$, for all sufficiently large $s$.

Choose one such finite $s$ and keep it fixed. Equation~\eqref{supp_exponent_criterion} then gives positive coherent information for all sufficiently large finite $d$. The input has finite energy, so $Q\geq\Ic$~\cite[Theorem~5]{WildeQi2018} proves Eq.~\eqref{supp_capacity_below_threshold}. The inequality $\eta_*(\nu)<\eta_{\rm G}(\nu)$ from Lemma~\ref{threshold_order_lemma} also puts the thermal threshold within this range, proving Eq.~\eqref{supp_capacity_at_threshold}.

\medskip

\paragraph{Comparison with multimode Gaussian inputs.}
For $\nu\geq\nu_{\rm c}$ and $\eta\leq\eta_{\rm G}(\nu)$, Theorem~\ref{gaussian_two_uses_thm} gives
\bb
 \sup_{m\geq1}\frac1m\sup_{\rho\in\Gcal_m}
 \Ic(\rho,\Phi_{\eta,\nu}^{\otimes m})=0.
\ee
If also $\eta>\eta_*(\nu)$, the finite choices of $s,d$ above give
\bb
 Q(\Phi_{\eta,\nu})\geq\Ic(\rho_{s,d},\Phi_{\eta,\nu})>0
 =\Ic^{\G,\infty}(\Phi_{\eta,\nu}).
\ee
The binary squeezed-state construction proves positive capacity throughout $\eta_*(\nu)<\eta\leq\eta_{\rm G}(\nu)$ for every $\nu>0$, with a strict advantage over all multimode Gaussian inputs when $\nu\geq\nu_{\rm c}$. This completes Theorem~\ref{analytic_thm}; the following subsections use finite-support Fock-comb inputs to obtain explicit rate bounds and reach lower transmissivities.
\end{proof}

\subsection{Finite-support input and cutoff-erasure lower bound}
\label{finite_witness_proof}

We prove Theorem~\ref{main_thm} for the finite-support input
\[
 \ket{\psi_j}=\frac{2\ket j+2\ket{j+4}+\ket{j+8}}{3}\quad(j=0,1),
 \qquad
 \rho_\star=\frac13\proj{\psi_0}+\frac23\proj{\psi_1}.
\]
At $(\eta,\nu)=(0.8,1)$, we prove the stronger cutoff bound
\[
 \Ic(\rho_\star,\T_{14}\!\circ\Phi_{0.8,1})\geq R_\star=5\times10^{-4},
 \qquad
 \Ic^{\G,\infty}(\Phi_{0.8,1})=0.
\]
The Gaussian equality follows from Theorem~\ref{gaussian_two_uses_thm}. Here we derive the finite matrices that determine the left-hand side; Sec.~\ref{certification_appendix} proves the required lower bound by interval arithmetic.

\begin{proof}[Proof of Theorem~\ref{main_thm}]
We first reduce the coherent information to a finite entropy calculation.

\emph{1. A lower bound from the cutoff.} The map in Eq.~\eqref{cutoff} is completely positive and trace preserving; for example, it has Kraus operators $K_0=P_M$ and $K_n=\ket\perp\!\bra n$ for $n>M$. For a purification $\psi_{RA}$ of $\rho_A$, the coherent information is the negative conditional entropy, i.e., $\Ic(\rho,\mathcal N)=-S(R|B)_{(\mathrm{id}\otimes\mathcal N)(\psi)}$. Processing $B$ through another channel can only increase $S(R|B)$~\cite{WildeBook2017,KhatriWilde2024}, and hence decrease $\Ic$. Therefore
\bb
 \Ic(\rho,\Phi_{0.8,1})\geq
 \Ic(\rho,\T_{14}\!\circ\Phi_{0.8,1}).
\ee
The right-hand side concerns a finite-dimensional output and the six-dimensional input subspace supporting $\rho_\star$. The erasure flag retains the discarded probability, so no postselection or renormalisation is involved.

\emph{2. Matrix entries from the channel dilation.} Every thermal attenuator admits the decomposition
\bb
 \Phi_{\eta,\nu}=\mathcal A_G\circ\mathcal L_\tau,
 \qquad G=1+(1-\eta)\nu,
 \qquad \tau=\eta/G,
 \label{factorization}
\ee
where $\mathcal L_\tau$ is a pure-loss channel and $\mathcal A_G$ is a quantum-limited amplifier (see e.g.~\cite{Rosati2018,NohAlbertJiang2019,MeleLamiGiovannetti2025}).  At $(\eta,\nu)=(0.8,1)$, one has $G=1.2$ and $\tau=2/3$.  A convenient Fock-basis Kraus representation for the two quantum-limited factors is~\cite{IvanSabapathySimon2011,MeleLamiGiovannetti2025}
\begin{align}
 L_\ell&=\sum_{n\geq\ell}
 \sqrt{\binom n\ell(1-\tau)^\ell\tau^{n-\ell}}\,
 \ket{n-\ell}\!\bra n,\label{losskraus}\\
 A_k&=\sum_{n\geq0}
 \sqrt{\binom{n+k}{k}\frac{(G-1)^k}{G^{n+k+1}}}\,
 \ket{n+k}\!\bra n.\label{ampkraus}
\end{align}
We obtain these formulas by taking Fock-basis matrix elements of the beam-splitter and two-mode-squeezing dilations of the pure-loss and quantum-limited amplifier channels, respectively~\cite{IvanSabapathySimon2011,MeleLamiGiovannetti2025}. Composing the two factors gives
\bb
 \Phi_{\eta,\nu}(X)
 =\sum_{k,\ell \geq 0} A_k^{\vphantom{\dag}} L_\ell^{\vphantom{\dag}} X L_\ell^\dagger A_k^\dagger.
 \label{thermal_kraus_composition}
\ee
Because $\rho_\star$ has finite Fock support, each matrix element retained in output levels $0,\ldots,14$ is a finite sum. The erasure block follows from trace preservation: for any operator $X$ supported on the input subspace, its flag coefficient is $\Tr X-\Tr[P_{14}\Phi_{0.8,1}(X)P_{14}]$. Thus the full cutoff output, including its erasure flag, is obtained from finite expressions.

\emph{3. Evaluation of the finite-dimensional entropies.} Let
\bb
 \ket\Psi_{RA}
 =
 \frac1{\sqrt3}\ket0_R\ket{\psi_0}_A
 +
 \sqrt{\frac23}\ket1_R\ket{\psi_1}_A,
\ee
and define
\begin{align}\label{eq_sigma}
 \sigma_{B'}=(\T_{14}\!\circ\Phi_{0.8,1})(\rho_\star),\qquad
 \sigma_{RB'}=\bigl(\mathrm{id}_R\otimes(\T_{14}\!\circ\Phi_{0.8,1})\bigr)(\proj\Psi).
\end{align}
The cutoff output has dimension $16$, i.e., the Fock levels $0,\ldots,14$ together with the erasure flag, so $\sigma_{B'}$ and $\sigma_{RB'}$ are, respectively, $16\times16$ and $32\times32$ density matrices. A direct numerical evaluation gives $S(\sigma_{B'})\approx 2.6777$, $S(\sigma_{RB'})\approx 2.6772$, and hence
\bb
 \Ic(\rho_\star,\T_{14}\!\circ\Phi_{0.8,1})
 \approx 5\times10^{-4}.
 \label{numerical_main_witness}
\ee

For the finite-support input $\rho_\star$, the cutoff reduces the desired capacity lower bound to the entropy difference of two density matrices of dimensions $16$ and $32$. The next subsection bounds both entropies by interval arithmetic, establishing $Q(\Phi_{0.8,1})\geq R_\star$ through data processing and the finite-energy coding theorem~\cite[Theorem~5]{WildeQi2018}.
\end{proof}

\subsection{Rigorous entropy bounds by interval arithmetic}
\label{certification_appendix}

We bound the exact coherent information in Eq.~\eqref{numerical_main_witness} by interval arithmetic. Figure~\ref{threshold_figure} shows our bounds for $\nu=1$, including the positive rates proved at $\eta=0.8$ (Theorem~\ref{main_thm}) and $\eta=0.7841$.

The proof follows the entropy calculation itself: enclose the matrix entries, use them to enclose the eigenvalues, and then bound the entropy of each matrix. All operations use outward rounding, so the final interval contains the exact coherent information. Since the channels are phase covariant and the inputs are Fock combs, the output matrices are block diagonal; we apply the procedure below to each block separately.

For each real symmetric matrix $A$ entering the entropy calculation, each entry $A_{ij}$ has an explicit finite expression from the Kraus representation in Eq.~\eqref{thermal_kraus_composition}. We enclose each exact input quantity in an interval known to contain it and evaluate these expressions with directed rounding, including the square roots. We thus obtain intervals $\mathbf A_{ij}
= [\underline A_{ij},\overline A_{ij}]$ such that $A_{ij}\in\mathbf A_{ij}$ for every pair of indices $i,j$. These entrywise enclosures form the interval matrix $\mathbf A=(\mathbf A_{ij})$.

Outward rounding ensures containment at every step~\cite{MooreKearfottCloud2009,Rump2010}. Gershgorin's theorem~\cite{Varga2004} then gives useful eigenvalue bounds if the off-diagonal entries are small. We therefore first change to a basis in which the matrix is nearly diagonal.

For this purpose, form the midpoint matrix $A_0\coloneqq\operatorname{mid}\mathbf A$, $
(A_0)_{ij} =
 \frac{\underline A_{ij}+\overline A_{ij}}{2}$, and compute an ordinary floating-point eigendecomposition of $A_0$. The approximate eigenvector matrix $U$ suggests a suitable basis. To obtain rigorous bounds in that basis, we replace it with an exactly orthogonal rational matrix.

Starting from $U$, we construct an exactly orthogonal matrix $Q$ that approximates its columns up to their signs, using the standard Householder construction~\cite{GolubVanLoan2013}. For a nonzero vector $v$, define
\bb
 H(v)=I-\frac{2vv^{\mathsf T}}{v^{\mathsf T}v}.
 \label{rational_householder}
\ee
At each step of the construction, we replace the floating-point Householder vector by a sufficiently accurate rational approximation, obtaining
\bb
 Q=H(v_1)\cdots H(v_s),
 \qquad
 v_1,\ldots,v_s\in\mathbb Q^d\setminus\{0\}.
\ee
Since $H(v)^{\mathsf T}H(v)=I$, every factor is exactly orthogonal, and therefore so is $Q$. Moreover, because each $v_j$ is rational, all entries of $H(v_j)$, and hence of $Q$, are rational. The column signs do not affect diagonalisation.

We then define $C\coloneqq Q^{\mathsf T}AQ$. Since $Q$ is exactly orthogonal, $A$ and $C$ have exactly the same spectrum; the rational approximations used to construct $Q$ affect only how close $C$ is to diagonal. A poor approximation may therefore lead to eigenvalue intervals that are too wide, but it cannot invalidate the eigenvalue bounds. Interval matrix multiplication with outward rounding gives $\mathbf C=Q^{\mathsf T}\mathbf A Q$, which encloses the exact transformed matrix $C=Q^{\mathsf T}AQ$ entrywise.

We next apply Gershgorin's theorem~\cite{Varga2004} to the transformed interval matrix $\mathbf C$. For each row $i$, define
\bb
 r_i\coloneqq\sum_{j\neq i}\sup|\mathbf C_{ij}|,
 \qquad
 \Gamma_i\coloneqq\mathbf C_{ii}+[-r_i,r_i],
 \label{gershgorin_interval}
\ee
where $\sup|\mathbf C_{ij}|
 \coloneqq\max\{|x|:x\in\mathbf C_{ij}\}$. Since $C$ is contained entrywise in $\mathbf C$, its entries satisfy $|C_{ij}|\leq \sup|\mathbf C_{ij}|$. Gershgorin's theorem therefore places every eigenvalue of $C$, and hence of $A$, in $\bigcup_i\Gamma_i$. In our computation, we verify by interval arithmetic that the intervals $\Gamma_i$ are pairwise disjoint. The counting statement in Gershgorin's theorem then guarantees that each $\Gamma_i$ contains exactly one eigenvalue, counted with multiplicity. Each eigenvalue is therefore enclosed individually.

Some eigenvalue intervals may extend slightly below zero because of rounding. Since $\sigma_{B'}$ and $\sigma_{RB'}$ are density matrices, their exact eigenvalues lie in $[0,1]$, so we first replace each enclosure $[a,b]$ by $[a,b]\cap[0,1]$.

For each resulting interval, we bound the entropy contribution $h(\lambda)=-\lambda\log_2\lambda$, with $h(0)=0$, using outward-rounded interval arithmetic for all operations, including the logarithm. Since $h$ increases on $[0,e^{-1}]$ and decreases on $[e^{-1},1]$, we determine its range by evaluating the endpoints and, whenever the interval contains $e^{-1}$, the value $h(e^{-1})$. In particular, for $\lambda\in[0,\varepsilon]$ with $\varepsilon\leq e^{-1}$, we obtain $h(\lambda)\in[0,h(\varepsilon)]$, which covers zero and nearly zero eigenvalues.

Overlapping Gershgorin intervals can instead be treated as a group. A group of $q$ overlapping intervals that is separated from all the others contains exactly $q$ eigenvalues, counted with multiplicity, by Gershgorin's theorem~\cite{Varga2004}; we can therefore bound their total contribution to the entropy without enclosing each eigenvalue separately. This also covers degenerate eigenvalues, including degeneracies imposed by symmetry, although the bounds may be wider. For the rigorous amplifier and additive-noise bounds, we further intersect these intervals with Weyl bounds obtained from the Frobenius norm of the off-diagonal part~\cite{GolubVanLoan2013}. Both enclosures contain the same exact ordered eigenvalues, so their intersection remains valid.

Summing the entropy contributions gives
\bb
 S(\sigma_{B'})\in[S_B^-,S_B^+],
 \qquad
 S(\sigma_{RB'})\in[S_{RB}^-,S_{RB}^+].
\ee
Subtracting the two entropy intervals with outward rounding gives
\bb
 \Ic(\rho_\star,\T_{14}\!\circ\Phi_{0.8,1})
 \in
 [S_B^- - S_{RB}^+,\,S_B^+ - S_{RB}^-].
\ee

Following this procedure, we obtain
$
 \Ic(\rho_\star,\T_{14}\!\circ\Phi_{0.8,1})
 \in{} 
 [5.01,\phantom{[}5.02]\times10^{-4}$, with both endpoints rounded outward. Since the lower endpoint is strictly larger than $R_\star=5\times10^{-4}$, the exact coherent information satisfies
$\Ic(\rho_\star,\T_{14}\!\circ\Phi_{0.8,1})
 \geq R_\star>0.$
The interval bounds and data processing prove $Q(\Phi_{0.8,1})\geq R_\star=5\times10^{-4}$, completing Theorem~\ref{main_thm}. The next subsection applies the same entropy bounds to a longer Fock-comb input, extending positive coherent information to a lower transmissivity.

\subsection{A larger Fock-comb input at \texorpdfstring{$\nu=1$}{nu=1}}
\label{threshold_appendix}

To obtain higher rates and reach smaller transmissivities, we optimise mixtures of $K$ orthogonal Fock-comb states,
\bb
 \rho=\sum_{r=0}^{K-1}p_r\proj{\psi_r},\qquad
 \ket{\psi_r}=\sum_{j=0}^{J-1}c_{rj}\ket{mj+r},\qquad 2\leq K\leq m,
\ee
with $\sum_j|c_{rj}|^2=1$, $p_r>0$ and $\sum_r p_r=1$. We vary the spacing $m$, length $J$, amplitudes and probabilities. The inset of Fig.~\ref{wl_comparison_figure} shows the resulting rates; Fig.~\ref{threshold_figure} resolves the small positive values below the thermal threshold. We now give the finite input for the lowest transmissivity at which we prove positivity.

The same method proves a positive rate at $\eta_0=0.7841$, marked by the square in Fig.~\ref{threshold_figure}. It lies below the thermal threshold $\eta_{\rm G}(1)=0.8$, while antidegradability forces zero capacity only for $\eta\leq\eta_{\rm AD}(1)=0.75$. Consider the non-Gaussian state
\bb
 \rho_0=\frac{67}{182}\proj{e_0}+\frac{115}{182}\proj{o_0},
\ee
with
\bb
 \ket{e_0}=\frac{1}{\sqrt{100008178}}\sum_{j=0}^{9}a_j\ket{4j},
 \qquad
 \ket{o_0}=\frac{1}{\sqrt{100004619}}\sum_{j=0}^{9}b_j\ket{4j+1},
\ee
and
\begin{align*}
 (a_j)&=(313,-717,1627,-3027,4551,-5436,4983,-3270,1324,-220),\\
 (b_j)&=(361,-895,1938,-3422,4866,-5464,4646,-2754,945,-106).
\end{align*}
The interval procedure of Sec.~\ref{certification_appendix}, with cutoff $M=50$, gives
\bb
 \Ic(\rho_0,\T_{50}\!\circ\Phi_{0.7841,1})
 \in[1.43,1.44]\times10^{-7}.
 \label{eta07841}
\ee
This proves strictly positive quantum capacity at $\eta=0.7841$ and, by monotonicity~\cite{Lami2020}, for all $\eta\geq0.7841$. Throughout $0.7841\leq\eta\leq0.8$ at $\nu=1$, the Gaussian optimum is zero for any number of channel uses by Theorem~\ref{gaussian_two_uses_thm}, so this input extends the non-Gaussian advantage below the analytic threshold $\eta_*(1)\approx0.7917$.

\begin{figure}[htbp]
 \centering
 \includegraphics[width=0.66\textwidth]{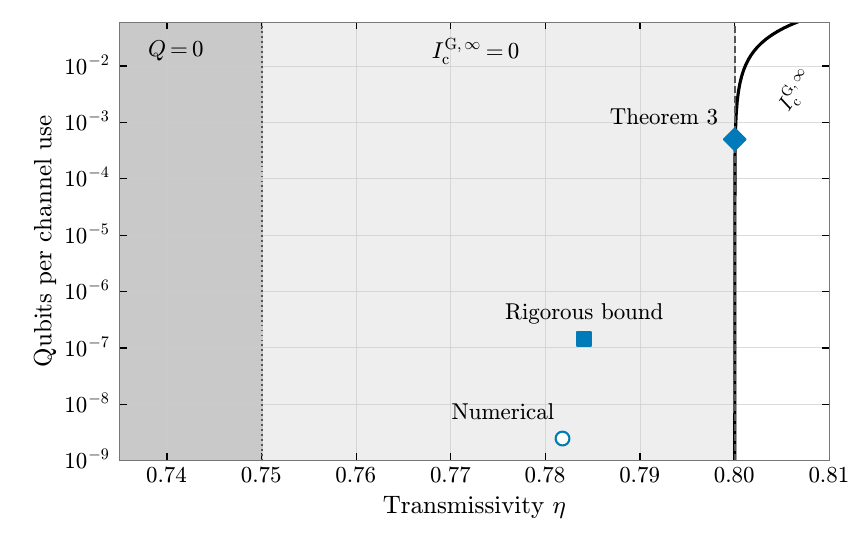}
 \caption{Non-Gaussian inputs transmit quantum information where all Gaussian inputs give zero rate, shown here for bath occupation $\nu=1$. The logarithmic scale makes the small positive rates visible. The black curve is the optimum over all multimode Gaussian inputs; it vanishes at and below the dashed thermal threshold $\eta=0.8$. The filled diamond and square are proved lower bounds; the open circle is a numerical estimate whose positivity has not been proved. Dark gray has zero quantum capacity, with its boundary marked by the dotted line. Light gray marks the remaining region where Gaussian inputs give zero rate.}
 \label{threshold_figure}
\end{figure}

The open circle in Fig.~\ref{threshold_figure} uses a separately optimised, longer comb, given in the accompanying data. Its positive rate is supported by numerical checks, but the proven endpoint remains $\eta=0.7841$. At $\nu=1$, finite-support Fock-comb inputs therefore establish an advantage over all multimode Gaussian inputs throughout $0.7841\leq\eta\leq0.8$. We next apply the same Fock-comb construction to thermal amplification and additive Gaussian noise.

\clearpage
\section{Thermal amplification}
\label{thermal_amplifier_section}

 For thermal amplification, we compare finite-energy non-Gaussian inputs with the optimum over all multimode Gaussian inputs. Corollary~\ref{gr_standard_channels_cor} shows that one channel use already gives this optimum. We define the thermal amplifier $\mathcal A_{\kappa,\nu}$ through its action on the annihilation operators in the Heisenberg picture~\cite{HolevoWerner2001,CarusoGiovannettiHolevo2006}:
\bb
 \hat b=\sqrt{\kappa}\,\hat a+\sqrt{\kappa-1}\,\hat e^\dagger,
 \label{thermal_amplifier}
\ee
where $\kappa>1$ is the gain and $\hat e$ is the annihilation operator of an environmental mode in a thermal state with mean photon number $\nu$. The channel acts as $V\mapsto\kappa V+(\kappa-1)(2\nu+1)\mathds{1}_2$.

For $\nu>0$, the coherent information optimised over thermal input states is~\cite{HolevoWerner2001,Bradler2015}
\bb
 \Ic^{\G,\infty}(\mathcal A_{\kappa,\nu})=
 \Ic^\G(\mathcal A_{\kappa,\nu})
 =\left[\log_2\frac{\kappa}{\kappa-1}-g(\nu)\right]_+,
 \qquad
 \kappa_{\rm th}(\nu)\coloneqq \frac{1}{1-2^{-g(\nu)}}.
 \label{amplifier_benchmark}
\ee

 Equation~\eqref{gr_amplifier_additivity}, with $\gamma=\kappa$, shows that this thermal-input supremum is also the Gaussian optimum per use for any number of modes.

The thermal threshold is therefore $\kappa_{\rm th}(\nu)$, while the channel is antidegradable for $\kappa\geq\kappa_{\rm AD}(\nu)=1+1/(2\nu)$~\cite{LamiKhatriAdessoWilde2019}. To find positive coherent information between these thresholds, we use the Fock-comb construction of Eq.~\eqref{comb_ansatz} and its extension to mixtures of $K$ orthogonal pure states.

For the numerical evaluation, we use the decomposition $\mathcal A_{\kappa,\nu}=\mathcal A_G\circ\mathcal L_\tau$, where $G=\kappa+(\kappa-1)\nu$ and $\tau=\kappa/G$~\cite{MeleLamiGiovannetti2025}, together with the Kraus operators in Eqs.~\eqref{losskraus}--\eqref{ampkraus}.

Figure~\ref{amplifier_figure} shows positive coherent information from these inputs beyond the \space multimode Gaussian threshold $\kappa_{\rm th}(\nu)$ for a broad range of thermal photon numbers $\nu$. A positive rate is proved at one point by the interval method of Sec.~\ref{certification_appendix}. At $\kappa=1.1441$ and $\nu=2.7$, where $\kappa_{\rm th}\approx1.13$ and the thermal-state bound of Eq.~\eqref{amplifier_benchmark} vanishes, we use the rank-two Fock comb of modulus~3 with ten teeth given in the supplementary data (coefficient vectors rounded to multiples of $10^{-6}$ and then exactly normalised) and, with cutoff $M=68$, obtain $\Ic(\rho,\T_{68}\circ\mathcal A_{\kappa,\nu})\in[1.91\times10^{-7},\,1.92\times10^{-7}]$, where $\rho$ is this comb state. The remaining points are evaluated in double precision.

\begin{figure}[ht]
 \centering
{\includegraphics[width=0.55\textwidth]{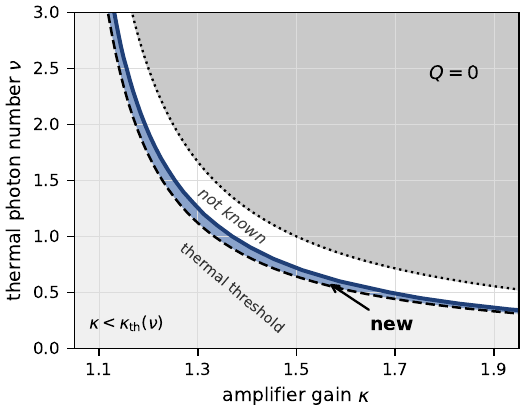}}
\caption{For thermal amplification, non-Gaussian inputs extend the positive-rate region beyond the \space multimode Gaussian threshold. The blue region comes from numerical searches over mixtures of superpositions of equally spaced photon-number states (Fock combs); its boundary joins sampled points. Light gray has positive thermal-input rates, ending at the dashed threshold. Dark gray has zero quantum capacity, with its boundary marked by the dotted curve. White remains unresolved. \space The Gaussian optimum vanishes throughout the blue region for any number of channel uses, so every positive rate there demonstrates a non-Gaussian advantage.}

 \label{amplifier_figure}
\end{figure}

 The Fock-comb input at $(\kappa,\nu)=(1.1441,2.7)$ therefore proves an advantage over every multimode Gaussian input. For additive Gaussian noise, we next compare Fock-comb inputs with the Gaussian optimum for any number of channel uses.

\clearpage
\section{Additive Gaussian noise}
\label{additive_noise_section}

Corollary~\ref{gr_standard_channels_cor} determines the Gaussian optimum for additive noise, including arbitrary correlations across channel uses. We now give non-Gaussian inputs that exceed this benchmark. The channel is the random-displacement map~\cite{Weedbrook2012}
\bb
 \mathcal N_\xi(\rho)\coloneqq\int_{\mathbb C}\frac{d^2\alpha}{\pi\xi}
 e^{-|\alpha|^2/\xi}D(\alpha)\rho D(\alpha)^\dagger,
 \qquad D(\alpha)\coloneqq e^{\alpha\hat a^\dagger-\alpha^*\hat a},
 \qquad \xi>0.
 \label{additive_channel}
\ee
Its covariance action is $V\mapsto V+2\xi\mathds{1}_2$. By Eq.~\eqref{additive_benchmark}, the Gaussian optimum per use is $[-\log_2(e\xi)]_+$, for any number of channel uses. It therefore vanishes for $\xi\geq\xi_{\rm th}\coloneqq1/e$, while the channel is antidegradable for $\xi\geq0.5$~\cite{LamiKhatriAdessoWilde2019}. The interval between these thresholds is where we seek positive non-Gaussian rates.

\begin{figure}[!htbp]
 \centering
{\includegraphics[width=0.55\textwidth]{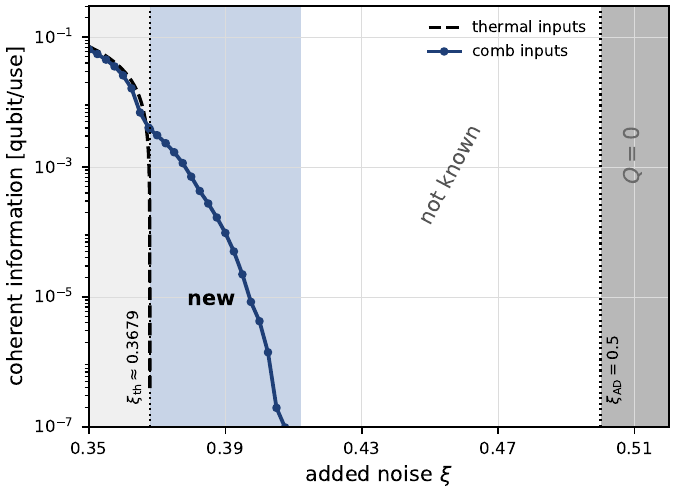}}
\caption{Non-Gaussian inputs tolerate additive noise beyond the limit for Gaussian inputs with any number of modes. Blue points show rates found numerically using Fock-comb inputs, which are mixtures of superpositions of equally spaced photon-number states. The dashed curve is the multimode Gaussian optimum and vanishes at $\xi_{\rm th}\approx0.3679$. Blue shading covers the noise range reached by these inputs, including rates too small to appear in the plot. Light gray has positive Gaussian rates; dark gray has zero quantum capacity. The dotted lines mark the Gaussian threshold and the boundary of the known zero-capacity region $\xi=0.5$; white remains unresolved.}

 \label{additive_figure}
\end{figure}

The numerical evaluation uses $\mathcal N_\xi=\mathcal A_{1+\xi}\circ\mathcal L_{1/(1+\xi)}$~\cite{MeleLamiGiovannetti2025} and the Kraus operators in Eqs.~\eqref{losskraus}--\eqref{ampkraus}.

Figure~\ref{additive_figure} shows positive coherent information from Fock-comb inputs beyond the Gaussian threshold $\xi_{\rm th}\approx0.3679$, up to $\xi\approx0.41$. At two noise values, interval arithmetic proves that the coherent information is positive. The following intervals enclose $\Ic(\rho,\T_M\circ\mathcal N_\xi)$, where $\rho$ is the specified comb state; by data processing, their lower endpoints also bound the coherent information before the cutoff. At $\xi=0.3825$, we use a rank-two comb of modulus~3 with ten teeth and cutoff $M=63$ to obtain $\Ic\in[4.28\times10^{-4},\,4.29\times10^{-4}]$; at $\xi=0.37$, we use a rank-two comb of modulus~4 with thirty teeth and cutoff $M=152$ to obtain $\Ic\in[3.09\times10^{-3},\,3.10\times10^{-3}]$. As for the amplifier, we use rounded coefficient vectors followed by exact normalisation for these states, and evaluate the remaining points in double precision. Both noise values exceed $\xi_{\rm th}$, and both lower endpoints are positive. For additive noise at $\xi=0.37$ and $0.3825$, finite-support non-Gaussian inputs therefore give positive coherent information where the Gaussian optimum is zero for any number of channel uses. We next return to thermal attenuation, derive a capacity upper bound, and use it to determine the capacity through first order in the loss.

\clearpage
\section{An improved upper bound and the weak-loss capacity expansion}
\label{weak_loss_section}
\suppressfloats[t]

We construct one family of degradable extensions of the thermal attenuator and evaluate their capacities. Optimising the squeezing improves the earlier bounds compared in Fig.~\ref{wl_comparison_figure}. A fixed, explicit squeezing makes the bound agree with the thermal-input rate through the term linear in $\varepsilon\coloneqq1-\eta$ as $\varepsilon\to0$. Throughout this section $\nu>0$ and $0.5<\eta<1$.

The extension method gives the receiver extra information until the enlarged channel becomes degradable; its optimised single-use coherent information then bounds the capacity of the original channel~\cite{SmithSmolin2008Extensions}. For thermal attenuation, Rosati, Mari and Giovannetti~\cite{Rosati2018} gave the receiver the bath purifier, and Fanizza, Kianvash and Giovannetti (FKG)~\cite{Fanizza2021} obtained a tighter bound by retaining only part of it. Here we first squeeze the two complementary modes together, then send part of one output to the receiver. The squeezing changes how much quantum side information is needed for degradability; we first describe the extension and state the resulting upper bound.

\subsection{The extension and its upper bound}
\label{wl_extension_subsection}

Let $EF$ be a two-mode squeezed vacuum whose two marginals have mean occupation $\nu$. Thus $\langle\hat e^\dagger\hat e\rangle=\langle\hat f^\dagger\hat f\rangle=\nu$ and $\langle\hat e\hat f\rangle=\sqrt{\nu(\nu+1)}$, with zero first moments. A beam splitter couples the input $A$ to $E$. We then apply a two-mode squeezer to the lost output $C$ and the purifier $F$, and mix the resulting $F_r$ with a fresh vacuum $V$:
\bb
 \hat b&=\sqrt\eta\,\hat a+\sqrt\varepsilon\,\hat e,
 &\hat c&=-\sqrt\varepsilon\,\hat a+\sqrt\eta\,\hat e,\\
 \hat c_r&=h\hat c+j\hat f^\dagger,
 &\hat f_r&=j\hat c^\dagger+h\hat f,\\
 \hat g&=\sqrt\lambda\,\hat f_r+\sqrt{1-\lambda}\,\hat v,
 &\hat h&=-\sqrt{1-\lambda}\,\hat f_r+\sqrt\lambda\,\hat v,
 \label{wl_dilation}
\ee
where $h=\cosh r$, $j=-\sinh r$, $r\geq0$ and $0<\lambda<1$. The receiver holds $BG$ and the full complement is $C_rH$, as shown in Fig.~\ref{wl_extension_figure}. Calling this extension $\mathcal N_{r,\lambda}$, we have
\bb
 \Phi_{\eta,\nu}=\Tr_G\circ\mathcal N_{r,\lambda},
 \qquad Q(\Phi_{\eta,\nu})\leq Q(\mathcal N_{r,\lambda}).
 \label{wl_extension}
\ee
The extra mode is retained as a quantum system. For $r>0$ it also carries conjugate input signal; it is not a classical record of the noise.

\begin{figure}[htbp]
\centering
\begin{tikzpicture}[x=1cm,y=1cm,font=\normalsize,
    sqwire/.style={line width=.65pt},
    sqarrow/.style={sqwire,-{Latex[length=1.8mm,width=1.3mm]}}]
  \filldraw[fill=black!4,draw=black!55,rounded corners=2pt]
    (13.65,1.15) rectangle (16.05,5.25);
  \node at (14.85,3.2) {Receiver};
  \node[anchor=east] at (1.05,4.8) {$A$};
  \node[anchor=east] at (1.05,0) {$\lvert0\rangle_V$};
  \draw[sqarrow] (1.15,4.8) -- (3.8,4.8);
  \draw[sqarrow] (1.15,0) -- (10.2,0);

  \draw[sqwire,fill=white] (1.25,1.2) rectangle (2.95,3.6);
  \node at (2.1,2.4) {$\lvert\psi_\nu\rangle_{EF}$};
  \draw[sqarrow] (2.95,3.2) -- node[above] {$E$} (3.8,3.2);
  \draw[sqarrow] (2.95,1.6) -- node[above,pos=.7] {$F$} (7,1.6);

  \draw[sqwire,fill=white] (3.8,2.9) rectangle (5.2,5.1);
  \node at (4.5,4) {$\mathrm{BS}_{\eta}$};
  \draw[sqarrow] (5.2,4.8) -- (14.65,4.8);
  \node[anchor=west] at (14.75,4.8) {$B$};
  \draw[sqarrow] (5.2,3.2) -- node[above] {$C$} (7,3.2);
  \draw[sqwire,fill=white] (7,1.3) rectangle (8.4,3.5);
  \node at (7.7,2.4) {$S_r$};
  \draw[sqarrow] (8.4,3.2) -- (12.25,3.2);
  \node[anchor=west] at (12.35,3.2) {$C_r$};
  \draw[sqarrow] (8.4,1.6) -- node[above] {$F_r$} (10.2,1.6);

  \draw[sqwire,fill=white] (10.2,-.3) rectangle (11.6,1.9);
  \node at (10.9,.8) {$\mathrm{BS}_{\lambda}$};
  \draw[sqarrow] (11.6,1.6) -- (14.65,1.6);
  \node[anchor=west] at (14.75,1.6) {$G$};
  \draw[sqarrow] (11.6,0) -- (12.25,0);
  \node[anchor=west] at (12.35,0) {$H$};
\end{tikzpicture}
\caption{Giving the receiver an extra quantum mode $G$ provides an upper bound on the thermal attenuator's capacity: discarding $G$ recovers the ordinary output $B$. The source prepares the thermal bath $E$ and its purifier $F$ in a two-mode squeezed vacuum. The first beam splitter mixes the input $A$ with $E$. The squeezer $S_r$ then mixes the lost output $C$ with $F$, and a second beam splitter sends a fraction $\lambda$ of $F_r$ to the receiver as $G$; $V$ is a fresh vacuum mode. The remaining outputs $C_r,H$ form the environment. We choose the squeezing and splitting so that $BG$ can simulate their joint state, making the extension degradable as shown in Fig.~\ref{wl_degrading_figure}.}
\label{wl_extension_figure}
\end{figure}

For the parameters admitted in the proposition below, Gaussian unitaries on the two output pairs give a canonical form with effective transmissivity $\kappa$ and noise occupations $n_R,n_E$. With $T$ as an auxiliary quantity, they are
\bb
 T&=(h\sqrt{\eta\nu}+j\sqrt{\nu+1})^2,
 &\kappa&=\eta-\varepsilon\lambda j^2,
 &n_R&=\frac{\lambda T}{\kappa},
 &n_E&=\frac{\nu(1-\lambda)}{1+\lambda j^2}.
 \label{wl_effective}
\ee
The output unitaries collect the input signal into one mode of each pair; $n_R$ and $n_E$ are the occupations of the remaining, input-independent modes. The physical bath occupation remains $\nu$.

\begin{grproposition}[(Squeezed-environment upper bound)]
\label{wl_bound_prop}
Suppose $T>0$, $0<\lambda<1$, $\kappa>0.5$, and
\bb
 \kappa^2 n_R(n_R+1)=(1-\kappa)^2 n_E(n_E+1).
 \label{wl_balance}
\ee
Then $\mathcal N_{r,\lambda}$ is degradable and
\bb
 Q(\Phi_{\eta,\nu})\leq Q(\mathcal N_{r,\lambda})
 =\log_2\frac\kappa{1-\kappa}+g(n_R)-g(n_E).
 \label{wl_upper_evaluation}
\ee
\end{grproposition}

For a fixed $r$, the balance condition~\eqref{wl_balance} is the quadratic $a_2\lambda^2+a_1\lambda-a_0=0$, where
\bb
 a_2&=(T-\varepsilon\nu)(T+\varepsilon\nu-\varepsilon j^2),
 &a_1&=\eta T+\varepsilon^2\nu(2\nu+1-j^2),
 &a_0&=\varepsilon^2\nu(\nu+1).
 \label{wl_quadratic}
\ee
We use the following positive root
\bb
 \lambda(r)=\frac{2a_0}{a_1+\sqrt{a_1^2+4a_2a_0}},
 \label{wl_t_root}
\ee
whenever $a_1>0$, the discriminant is positive, and the resulting parameters satisfy the proposition. Write $\mathcal A_{\eta,\nu}$ for these admissible values of $r$, and $U(r;\eta,\nu)$ for the right-hand side of Eq.~\eqref{wl_upper_evaluation} evaluated at $\lambda(r)$. This set is nonempty: $r=0$ gives the unique $0<\lambda<0.5$ satisfying $(2\eta-1)\lambda(1+\nu\lambda)=(\nu+1)\varepsilon^2(1-2\lambda)$. Thus
\bb
 L(\eta,\nu)\coloneqq\left[\log_2\frac\eta\varepsilon-g(\nu)\right]_+
 \leq Q(\Phi_{\eta,\nu})\leq U_{\rm opt}(\eta,\nu)\coloneqq\inf_{r\in\mathcal A_{\eta,\nu}}U(r;\eta,\nu).
 \label{wl_sandwich}
\ee
Thermal inputs give the lower bound $L(\eta,\nu)$~\cite{HolevoWerner2001,WildeQi2018}, while each admissible squeezed extension gives an upper bound on $Q(\Phi_{\eta,\nu})$. To prove the upper bound, we next reconstruct the complementary output from the received pair $BG$ and then evaluate the extension's capacity.

\subsection{A physical degrading map}
\label{wl_degrading_subsection}

The added squeezer makes both received modes depend on the input. A change of output modes collects this signal into one mode and exposes the noise correlations needed for degradation. This is an output unitary, so it changes neither the capacity nor degradability.

In the $q$ quadratures, the signal vectors of $BG$ and $C_rH$ are
\bb
 x_R=(\sqrt\eta,-j\sqrt{\varepsilon\lambda})^{\mathsf T},
 \qquad x_E=(-h\sqrt\varepsilon,j\sqrt{\varepsilon(1-\lambda)})^{\mathsf T}.
\ee
Their squared norms with respect to $S=\operatorname{diag}(1,-1)$ are $\kappa$ and $1-\kappa$. The two-mode squeezers with $q$ matrices
\bb
 R_R=\frac1{\sqrt\kappa}\begin{pmatrix}x_{R1}&-x_{R2}\\-x_{R2}&x_{R1}\end{pmatrix},
 \qquad R_E=\frac1{\sqrt{1-\kappa}}\begin{pmatrix}-x_{E1}&x_{E2}\\x_{E2}&-x_{E1}\end{pmatrix}
 \label{wl_boosts}
\ee
send these vectors to $(\sqrt\kappa,0)^{\mathsf T}$ and $(-\sqrt{1-\kappa},0)^{\mathsf T}$. Their $p$ matrices are $SR_RS=R_R^{-\mathsf T}$ and $SR_ES=R_E^{-\mathsf T}$, so they are physical symplectic transformations. %
After $R_R$, multiply both quadratures of the second receiver mode by $\operatorname{sgn}(h\sqrt{\eta\nu}+j\sqrt{\nu+1})$; after $R_E$, leave the phases unchanged. This is either the identity or a phase shift by $\pi$. The sign is well defined because $T>0$. Denote the resulting unitaries by $U_B,U_E$ and the output pairs by $\widetilde B\widetilde G$ and $\widetilde C\widetilde H$.

For completeness, the receiver's original $q$-noise matrix is
\bb
 Y_R=\begin{pmatrix}
 \varepsilon v&\sqrt{\varepsilon\lambda}(\sqrt\eta\,vj+wh)\\
 \sqrt{\varepsilon\lambda}(\sqrt\eta\,vj+wh)&1-\lambda+\lambda(\eta vj^2+vh^2+2\sqrt\eta\,whj)
 \end{pmatrix},
 \quad v=2\nu+1,\quad w=2\sqrt{\nu(\nu+1)}.
\ee
Applying these transformations to $Y_R$ and to the complementary noise in Eq.~\eqref{wl_dilation} gives, with $\bar\nu\coloneqq n_R+n_E$,
\bb
 \widetilde Y_R&=\begin{pmatrix}
 (1-\kappa)(2\bar\nu+1)&2\sqrt{(1-\kappa)n_R(\bar\nu+1)}\\
 2\sqrt{(1-\kappa)n_R(\bar\nu+1)}&1+2n_R
 \end{pmatrix},\qquad
 \widetilde Y_E=\begin{pmatrix}
 \kappa(2\bar\nu+1)&-2\sqrt{\kappa n_E(\bar\nu+1)}\\
 -2\sqrt{\kappa n_E(\bar\nu+1)}&1+2n_E
 \end{pmatrix}.
 \label{wl_canonical}
\ee
The $p$-noise matrices are $S\widetilde Y_RS$ and $S\widetilde Y_ES$; the mixed $q,p$ blocks vanish. These identities specify the full Gaussian channels, including their action on non-Gaussian inputs. In the canonical modes, the effective bath has occupation $\bar\nu$, with shares $n_R$ and $n_E$ at the receiver and complement. This change of modes is the only extra step introduced by the squeezing in Fig.~\ref{wl_extension_figure}.

We can now reproduce the complement directly. Equation~\eqref{wl_balance}, together with $\kappa>0.5$, implies $n_R<n_E$. Choose $\sigma>0$ and a beam-splitter transmissivity $\tau$ as
\bb
 \cosh^2\sigma=\frac{1+n_E}{1+n_R},\qquad \tau=\frac{1-\kappa}{\kappa}\in(0,1).
\ee
With a fresh vacuum mode $W$, the two retained outputs of the degrading circuit are
\bb
 \widetilde c'&=-\sqrt{\frac{1-\kappa}{\kappa}}\,\widetilde b+
 \sqrt{\frac{2\kappa-1}{\kappa}}\bigl(\cosh\sigma\,\hat w-\sinh\sigma\,\widetilde g^\dagger\bigr),\\
 \widetilde h'&=\cosh\sigma\,\widetilde g-\sinh\sigma\,\hat w^\dagger.
 \label{wl_degrading_gain}
\ee
This is a two-mode squeezer followed by a phase shift and a beam splitter, with one unused output discarded. It therefore defines a quantum channel $\widetilde{\mathcal D}$. Figure~\ref{wl_degrading_figure} shows how the received modes $B,G$ enter the degrading map; the complementary modes $C_r,H$ are not used.

\begin{figure}[t]
\centering
\begin{tikzpicture}[x=.99cm,y=1cm,font=\normalsize,
    dgwire/.style={line width=.65pt},
    dgarrow/.style={dgwire,-{Latex[length=1.7mm,width=1.2mm]}}]
  \node at (4.2,7.05) {Extension $\mathcal N_{r,\lambda}$};
  \node at (12.4,7.05) {Degrading map $\mathcal D$};
  \filldraw[fill=black!3,draw=black!55,densely dashed,rounded corners=2pt]
    (8.55,2.45) rectangle (16.25,6.65);

  \node[anchor=east] at (.3,6) {$A$};
  \draw[dgarrow] (.4,6) -- (2.55,6);
  \draw[dgwire,fill=white] (.6,2.7) rectangle (1.8,4.8);
  \node at (1.2,3.75) {$\lvert\psi_\nu\rangle$};
  \draw[dgarrow] (1.8,4.5) -- node[above] {$E$} (2.55,4.5);
  \draw[dgarrow] (1.8,3) -- node[above] {$F$} (4.7,3);
  \draw[dgwire,fill=white] (2.55,4.2) rectangle (3.45,6.3);
  \node at (3,5.25) {$\mathrm{BS}_{\eta}$};
  \draw[dgarrow] (3.45,6) -- node[above,pos=.88] {$B$} (8.8,6);
  \draw[dgarrow] (3.45,4.5) -- node[above] {$C$} (4.7,4.5);
  \draw[dgwire,fill=white] (4.7,2.7) rectangle (5.6,4.8);
  \node at (5.15,3.75) {$S_r$};
  \draw[dgarrow] (5.6,4.5) -- (7.65,4.5);
  \node[anchor=west] at (7.75,4.5) {$C_r$};
  \draw[dgarrow] (5.6,3) -- node[above] {$F_r$} (6.6,3);
  \node[anchor=east] at (.3,.8) {$\lvert0\rangle_V$};
  \draw[dgarrow] (.4,.8) -- (6.6,.8);
  \draw[dgwire,fill=white] (6.6,.5) rectangle (7.5,3.3);
  \node at (7.05,1.9) {$\mathrm{BS}_{\lambda}$};
  \draw[dgarrow] (7.5,3) -- node[above] {$G$} (8.8,3);
  \draw[dgarrow] (7.5,.8) -- (7.9,.8);
  \node[anchor=west] at (8,.8) {$H$};

  \draw[dgwire,fill=white] (8.8,2.7) rectangle (9.55,6.3);
  \node at (9.175,5.05) {$U_B$};
  \draw[dgarrow] (9.55,6) -- node[above] {$\widetilde B$} (11.1,6);
  \draw[dgarrow] (9.55,3) -- node[above] {$\widetilde G$} (11.15,3);
  \node[anchor=east] at (10.55,4.5) {$\lvert0\rangle_W$};
  \draw[dgarrow] (10.65,4.5) -- (11.15,4.5);
  \draw[dgwire,fill=white] (11.15,2.7) rectangle (12,4.8);
  \node at (11.575,3.75) {$S_{\sigma}$};
  \draw[dgarrow] (12,4.5) -- node[above] {$\widetilde W$} (13.4,4.5);
  \draw[dgarrow] (12,3) -- node[above,pos=.45] {$\widetilde H'$} (15.4,3);
  \draw[dgwire,fill=white] (11.1,5.72) rectangle (11.8,6.28);
  \node at (11.45,6) {$\pi$};
  \draw[dgarrow] (11.8,6) -- (13.4,6);
  \draw[dgwire,fill=white] (13.4,4.2) rectangle (14.3,6.3);
  \node at (13.85,5.25) {$\mathrm{BS}_{\tau}$};
  \draw[dgarrow] (14.3,6) -- node[above] {$\widetilde C'$} (15.4,6);
  \draw[dgarrow] (14.3,4.5) -- (14.7,4.5) -- (14.7,4);
  \node at (14.7,3.7) {discard};
  \draw[dgwire,fill=white] (15.4,2.7) rectangle (16.1,6.3);
  \node at (15.75,5.05) {$U_E^{-1}$};
  \draw[dgarrow] (16.1,6) -- (16.75,6);
  \node[anchor=west] at (16.85,6) {$C_r'$};
  \draw[dgarrow] (16.1,3) -- (16.75,3);
  \node[anchor=west] at (16.85,3) {$H'$};
  \node at (12.4,1.75) {$\tau=(1-\kappa)/\kappa$};
  \node at (12.4,.8) {$\mathcal D\circ\mathcal N_{r,\lambda}=\mathcal N_{r,\lambda}^{c}$};
\end{tikzpicture}
\caption{The receiver can simulate the entire environmental output of the extended channel. The shaded circuit uses only the received modes $B,G$ and a fresh vacuum $W$; it has no access to the actual environment $C_r,H$. For the parameter choices in Proposition~\ref{wl_bound_prop}, its outputs $C_r',H'$ have exactly the same joint state as $C_r,H$ for every input, including correlations with a reference system. This is the degrading map. The reversible transformations $U_B,U_E^{-1}$ change the output modes; between them, a squeezer, phase shift and beam splitter produce the required correlations. Here $\kappa$ is the effective signal transmissivity after $U_B$. Degradability reduces the extension's capacity calculation to a single channel use.}
\label{wl_degrading_figure}
\end{figure}

To check the channel identity, write $\widetilde c'=-\sqrt{1-\kappa}\,\hat a+\hat z$. Substitution of Eq.~\eqref{wl_canonical} into Eq.~\eqref{wl_degrading_gain}, using the balance condition, gives
\bb
 \langle\hat z^\dagger\hat z\rangle=\kappa\bar\nu,
 \qquad\langle\widetilde h'^\dagger\widetilde h'\rangle=n_E,
 \qquad\langle\hat z\widetilde h'\rangle=-\sqrt{\kappa n_E(\bar\nu+1)}.
 \label{wl_noise_moments}
\ee
All other noise moments vanish or follow from commutation relations. These are precisely the complementary moments in Eq.~\eqref{wl_canonical}, and the signal coefficients also agree. Hence the canonical channel $\mathcal M_{r,\lambda}\coloneqq U_B\circ\mathcal N_{r,\lambda}$ has complementary channel $\mathcal M_{r,\lambda}^c\coloneqq U_E\circ\mathcal N_{r,\lambda}^c$, which is obtained through $\widetilde{\mathcal D}$. In the original modes, $\mathcal D\coloneqq U_E^{-1}\circ\widetilde{\mathcal D}\circ U_B$ reproduces the full complementary output $C_rH$ from the received pair $BG$ for every input. The extension $\mathcal N_{r,\lambda}$ is therefore degradable, and we now use this property to maximise its coherent information over all input states.

\subsection{Evaluating the capacity of the extension}
\label{wl_capacity_subsection}

Degradability reduces the capacity calculation to one use of the extension, but the input can still be an arbitrary state. We show that phase averaging cannot lower its coherent information and that replacing the averaged state by a thermal state of the same energy cannot lower it either. The thermal coherent information is nondecreasing with energy, so its large-energy limit gives the capacity.

\paragraph{From many uses to one.}
At finite input energy all output entropies are finite. For a degradable channel, coherent information is concave in the input~\cite{YardHaydenDevetak2008}: for any finite ensemble, the concavity difference is the receiver's Holevo information minus the complement's Holevo information, which is nonnegative by data processing through $\widetilde{\mathcal D}$. The canonical degrading map also gives subadditivity across channel uses~\cite{DevetakShor2005}. For a two-use input, write $\mathsf B_j=\widetilde B_j\widetilde G_j$ and $\mathsf E_j=\widetilde C_j\widetilde H_j$. Then
\bb
 \Ic(\rho_{12},\mathcal M_{r,\lambda}^{\otimes2}) - \Ic(\rho_1,\mathcal M_{r,\lambda}) - \Ic(\rho_2,\mathcal M_{r,\lambda})
 = I(\mathsf E_1:\mathsf E_2) - I(\mathsf B_1:\mathsf B_2)\leq0,
 \label{wl_subadditivity}
\ee
because the two local degrading maps cannot increase mutual information. Repeating the argument proves subadditivity of the extension's coherent information for any finite number of channel uses. It remains to optimise the single-use input, which we do by phase averaging and comparison with a thermal state.

\paragraph{From an arbitrary input to a thermal state.}
The canonical channel has the same phase covariance as the extension of Ref.~\cite{Fanizza2021}: the two receiver modes rotate by opposite angles, and so do the two complementary modes. Average an input $\rho$ of mean photon number $N$ over the four angles $0,\pi/2,\pi,3\pi/2$. The resulting state $\bar\rho$ has zero first moments and covariance $(2N+1)\mathds{1}_2$, while concavity gives
$\Ic(\rho,\mathcal M_{r,\lambda})\leq\Ic(\bar\rho,\mathcal M_{r,\lambda})$.
Thus $\bar\rho$ has the same first and second moments as the thermal state $\tau_N$, although it need not itself be Gaussian.

To compare the two states, recall that a faithful (full-rank) Gaussian state $\omega_{\G}$ has a logarithm that is quadratic in the quadratures. If $\omega$ has the same first and second moments, this implies
$D(\omega\Vert\omega_{\G})=S(\omega_{\G})-S(\omega)$~\cite[Eq.~(1)]{GenoniParisBanaszek2008}.
Apply the identity at the receiver and complementary outputs. Their Gaussian reference states are faithful, as verified below, and data processing through $\widetilde{\mathcal D}$ gives
\bb
 \Ic(\tau_N,\mathcal M_{r,\lambda}) - \Ic(\bar\rho,\mathcal M_{r,\lambda})
 &= D\bigl(\mathcal M_{r,\lambda}(\bar\rho)\Vert\mathcal M_{r,\lambda}(\tau_N)\bigr)
 - D\bigl(\mathcal M_{r,\lambda}^c(\bar\rho)\Vert\mathcal M_{r,\lambda}^c(\tau_N)\bigr)\geq0.
 \label{wl_gaussian_optimality}
\ee
This proves thermal optimality at fixed $N$ for the extension.

The required faithfulness follows directly from the output covariances. Both Gaussian outputs have the form
\bb
 V_X(N) = \begin{pmatrix}a_X\mathds{1}_2&c_XZ\\c_XZ&b_X\mathds{1}_2\end{pmatrix},
 \qquad Z = \operatorname{diag}(1,-1),\qquad X = R,E,
 \label{wl_output_covariances}
\ee
with
\bb
 a_R &= 2\kappa N + 1 + 2(1-\kappa)\bar\nu,
 &b_R &= 1 + 2n_R,
 &c_R &= 2\sqrt{(1-\kappa) n_R(\bar\nu + 1)},\\
 a_E &= 2(1-\kappa) N + 1 + 2\kappa\bar\nu,
 &b_E &= 1 + 2n_E,
 &c_E &= -2\sqrt{\kappa n_E(\bar\nu + 1)}.
 \label{wl_output_parameters}
\ee
For the receiver, the two determinants controlling strict uncertainty are
\bb
 (a_R - 1)(b_R + 1) - c_R^2 &= 4\bigl[\kappa N(1 + n_R) + (1-\kappa)n_E\bigr]>0,\\
 (a_R + 1)(b_R - 1) - c_R^2 &= 4\kappa n_R(N + 1)>0.
 \label{wl_faithfulness}
\ee
Together with $a_R,b_R>1$, these inequalities establish strict uncertainty and hence faithfulness. Exchanging $\kappa$ with $1-\kappa$ and $n_R$ with $n_E$ gives the corresponding check for the complement, also at $N=0$. The relative-entropy comparison is therefore valid at every finite energy, and thermal states maximise the extension's coherent information at that energy. We now show that increasing the thermal input energy approaches the capacity of the extension.

\paragraph{Taking the input energy large.}
A hotter thermal state is a Gaussian mixture of displacements of a colder one. Displacement covariance leaves coherent information unchanged for each displaced state, and concavity therefore shows that $\Ic(\tau_N,\mathcal M_{r,\lambda})$ is nondecreasing in $N$~\cite[Eqs.~(327)--(330)]{WildeQi2018}. The integral mixture can be approximated by finite mixtures under a common finite-energy bound, for which output entropies converge.

The photon-number observable has a finite Gibbs partition function at every positive inverse temperature, and bounded input energy gives bounded output energies and entropies. The energy-constrained degradable-channel coding theorem~\cite[Theorem~3]{WildeQi2018} therefore applies. Taking the supremum over finite input energies gives
\bb
 Q(\mathcal M_{r,\lambda}) = \sup_{N<\infty}\Ic(\tau_N,\mathcal M_{r,\lambda}) = \lim_{N\to\infty}\Ic(\tau_N,\mathcal M_{r,\lambda}).
 \label{wl_extension_capacity}
\ee
The equality also holds for the operational capacity without an energy bound. To see the converse, consider entanglement generation with a finite-dimensional reference $R$, which includes entanglement transmission as a special case. A pure preparation on $RA^n$ suffices: for a mixed preparation, one pure component achieves at least the average overlap with the target maximally entangled state after decoding. For each fixed number of channel uses, projecting this pure state onto a sufficiently large photon-number subspace and normalising changes the overlap arbitrarily little. The truncated input has finite energy, so Eq.~\eqref{wl_subadditivity} bounds its coherent information by $n$ times the right-hand side of Eq.~\eqref{wl_extension_capacity}, uniformly in the cutoff. Data processing through the decoder and continuity of conditional entropy with finite-dimensional $R$~\cite[Lemma~2]{Winter2016Continuity} give the coherent-information converse. The bound remains valid as the cutoff and input energy grow with the number of channel uses.

We can now evaluate the limit explicitly. The symplectic eigenvalues of Eq.~\eqref{wl_output_covariances} are
\[
 \frac{\sqrt{(a_X+b_X)^2-4c_X^2}\pm(a_X-b_X)}{2}.
\]
At fixed channel parameters, one grows as $a_X+O(N^{-1})$ and the other tends to $b_X$. Using $g(x)=\log_2x+\log_2e+o(1)$ for the diverging contribution gives
\bb
 Q(\mathcal M_{r,\lambda}) = \log_2\frac{\kappa}{1-\kappa} + g(n_R) - g\bigl(n_E\bigr).
 \label{wl_extension_limit}
\ee
The diverging entropy contributions leave the finite difference $\log_2(\kappa/(1-\kappa))$; the other two terms are the entropies of the finite canonical modes. Output unitaries do not change capacity, so this is also $Q(\mathcal N_{r,\lambda})$. Thermal inputs of increasing energy approach the capacity of $\mathcal N_{r,\lambda}$, and discarding $G$ gives the thermal-attenuator upper bound in Proposition~\ref{wl_bound_prop}. We now choose the squeezing explicitly and compare this upper bound with the thermal-input lower bound to determine the weak-loss capacity through first order.

\begin{samepage}
\subsection{The weak-loss expansion}
\label{wl_expansion_subsection}

We now keep the physical bath occupation $\nu>0$ fixed and let $\varepsilon\to0$. We restate the result proved here.
\begin{manualthm}{\ref{weak_loss_thm}}
For every fixed $\nu>0$, as $\varepsilon\to0$,
\bb
 Q(\Phi_{1-\varepsilon,\nu})=\log_2\frac1\varepsilon-g(\nu)-\frac\varepsilon{\ln2}
 +O_\nu\!\left(\varepsilon^2\log\frac1\varepsilon\right).
 \label{wl_main}
\ee
The implicit constant may depend on $\nu$.
\end{manualthm}
\end{samepage}

\begin{proof}
Choose the nonzero squeezing
\bb
 r_\nu=2\operatorname{arsinh}\sqrt\nu,\qquad h=2\nu+1,\qquad j=-2\sqrt{\nu(\nu+1)}.
 \label{wl_fixed_squeeze}
\ee
For this choice, $T=\nu[(2\nu+1)\sqrt\eta-2(\nu+1)]^2$ tends to $\nu$ as $\eta\to1$. Consequently $a_1\to\nu>0$, the discriminant in Eq.~\eqref{wl_t_root} tends to $\nu^2$, and its positive root gives
\bb
 \lambda&=(\nu+1)\varepsilon^2+O_\nu(\varepsilon^3),
 &\kappa&=1-\varepsilon-4\nu(\nu+1)^2\varepsilon^3+O_\nu(\varepsilon^4),\\
 n_R&=\nu(\nu+1)\varepsilon^2+O_\nu(\varepsilon^3),
 &n_E&=\nu+O_\nu(\varepsilon^2).
 \label{wl_t_asymptotic}
\ee
This choice keeps the complementary occupation close to the physical bath occupation, while the receiver's extra canonical occupation is only of order $\varepsilon^2$. In particular, $0<\lambda<1$ and $\kappa>0.5$ for all sufficiently small positive $\varepsilon$, so this choice is admissible. The effective transmissivity differs from $\eta$ only at cubic order, and hence
\bb
 \log_2\frac\kappa{1-\kappa}=\log_2\frac\eta\varepsilon+O_\nu(\varepsilon^2).
\ee
The small canonical occupation $n_R$ is responsible for the logarithm in the remainder: $g(x)=x\log_2(1/x)+x/\ln2+O(x^2)$ near zero. In contrast, $g$ is differentiable at the fixed positive occupation $\nu$. Equations~\eqref{wl_upper_evaluation} and~\eqref{wl_t_asymptotic} therefore give
\bb
 U(r_\nu;1-\varepsilon,\nu)-\left(\log_2\frac{1-\varepsilon}{\varepsilon}-g(\nu)\right)
 =2\nu(\nu+1)\varepsilon^2\log_2\frac1\varepsilon+O_\nu(\varepsilon^2).
 \label{wl_gap_leading}
\ee
The thermal lower bound is positive for sufficiently small $\varepsilon$, and
\bb
 \log_2\frac{1-\varepsilon}{\varepsilon}-g(\nu)
 =\log_2\frac1\varepsilon-g(\nu)-\frac\varepsilon{\ln2}+O(\varepsilon^2).
 \label{wl_lower_expansion}
\ee
Finally, $L\leq Q\leq U_{\rm opt}\leq U(r_\nu;1-\varepsilon,\nu)$ by Eq.~\eqref{wl_sandwich}. The gap in Eq.~\eqref{wl_gap_leading} is smaller than the linear correction, which proves Theorem~\ref{weak_loss_thm}.
\end{proof}

Equation~\eqref{wl_gap_leading} gives the gap for the explicit choice~\eqref{wl_fixed_squeeze}; it does not determine the next capacity term or the exact asymptotics of the optimised bound. The squeezed extension gives a finite-loss upper bound and, at fixed $\nu>0$, limits any improvement over the thermal-input rate to $O_\nu(\varepsilon^2\log(1/\varepsilon))$, fixing the capacity through the term linear in the loss. Section~\ref{numerics_guide_section} lists the programs that evaluate the bounds and reproduce the figures.

\section{Guide to the numerics}
\label{numerics_guide_section}

The accompanying \texttt{Supplementary\_data/} folder contains all the data and programs sufficient to reproduce the numerical results reported in this paper. The input states are reported in \texttt{data/}, and the programs are in \texttt{numerics/}. 

\paragraph{Figures and input states.}
\texttt{fig1\_attenuator.py}, \texttt{fig2\_capacity\_bounds.py}, \texttt{figS3\_attenuator.py}, \texttt{figS4\_amplifier.py} and \texttt{figS5\_additive\_noise.py} reproduce Figs.~\ref{highground_figure}(a), \ref{wl_comparison_figure}, \ref{threshold_figure}, \ref{amplifier_figure} and~\ref{additive_figure}, respectively. For each non-Gaussian input, the data specify its amplitudes, mixing probabilities, channel parameters and output cutoff. The programs recompute the cutoff-erasure coherent information defined in Sec.~\ref{finite_witness_proof}. The light-blue region in Fig.~\ref{highground_figure}(a), by contrast, follows analytically from Theorem~\ref{analytic_thm}. The diagrams in Figs.~\ref{state_preparations_figure}(b), \ref{gaussian_pair_preparation_figure}--\ref{squeezed_pair_figure} and~\ref{wl_extension_figure}--\ref{wl_degrading_figure} are drawn directly in LaTeX and require no numerical data.

\paragraph{Individual rates.}
The following programs evaluate individual non-Gaussian rates and comparisons:
\begin{itemize}[nosep]
\item \texttt{result\_theorem3.py} and \texttt{result\_eta07841.py} give the rigorous lower bounds in Theorem~\ref{main_thm} and Eq.~\eqref{eta07841}, using the interval method of Sec.~\ref{certification_appendix}.
\item \texttt{result\_certificate\_amplifier\_additive.py} applies the same method to the amplifier and additive-noise examples in Secs.~\ref{thermal_amplifier_section} and~\ref{additive_noise_section}.
\item \texttt{result\_eta07818.py} evaluates the input at $(\eta,\nu)=(0.7818,1)$ shown in Fig.~\ref{threshold_figure}, including larger output cutoffs. \texttt{result\_eta080\_rate.py} evaluates additional optimised inputs at $(\eta,\nu)=(0.8,1)$, using the same Fock-comb families as the inset of Fig.~\ref{wl_comparison_figure}. Both are numerical estimates.
\item \texttt{result\_phase\_flip.py} reproduces the sign change in ``\hyperlink{main_phase_flip}{Mechanism and applications}'': changing only the sign of the $\ket5$ component of $\ket{\psi_1}$ in Eq.~\eqref{codewords} leaves all photon-number probabilities unchanged, but numerically makes the coherent information negative.
\end{itemize}
The interval calculations establish positive rates at the specified channel parameters. The other evaluations allow the numerical estimates, and their dependence on the cutoff, to be checked separately.

\paragraph{The attenuator Gaussian threshold.}
\texttt{result\_gaussian\_regularization.py} verifies the scalar sign inequalities used in Sec.~\ref{gr_certificate_subsection}, including the enclosure of $\nu_{\rm c}$. This part of the Gaussian-superadditivity proof uses exact-integer interval arithmetic and reads no saved numerical results. Gaussian additivity for amplification and additive noise is proved analytically in Lemma~\ref{gr_nonattenuating_additivity_lem} and requires no numerical calculation.

\paragraph{Capacity upper bounds.}
For Fig.~\ref{wl_comparison_figure}, \texttt{fig2\_capacity\_bounds.py} chooses the extension parameters numerically and evaluates the upper bound in Eq.~\eqref{wl_upper_evaluation} with outward rounding. Every admissible choice gives a valid bound; finding the global minimum is not required. The earlier bounds compared are PLOB~\cite{Pirandola2017} and FKG~\cite{Fanizza2021}, including the two attenuator-composition branches and the additive-noise comparison of KFG~\cite{Kianvash2024}. Only the smallest of these values is plotted. The additional optimisation with entangled auxiliary states mentioned in Ref.~\cite{Fanizza2021} is not included. The README records the formulas used in this comparison.

Finally, \texttt{certify\_weak\_loss\_example.py} reproduces the $9\,\mu\mathrm{m}$ \hyperlink{main_weak_loss_example}{example} in the main text, with a thermal bath at $300\,\mathrm{K}$ and $\eta=0.99$, using the Planck occupation before rounding. It compares our upper bound and FKG~\cite{Fanizza2021} with the thermal-input lower bound and bounds the absolute error of the linearised formula obtained from Eq.~\eqref{weak_loss_main}.

\endgroup

\end{document}